\documentclass[11pt,letterpaper]{article}
\pdfoutput=1
\usepackage[T1]{fontenc}
\usepackage[ascii]{inputenc}
\usepackage{lmodern}
\usepackage[dvipsnames,svgnames]{xcolor}
\usepackage{amsmath,amssymb,amsthm,mathtools,enumitem,mleftright}
\usepackage[margin=0.9in]{geometry}
\usepackage[bookmarksnumbered,linktocpage,hypertexnames=false,colorlinks=true,linkcolor=NavyBlue,urlcolor=NavyBlue,citecolor=ForestGreen,anchorcolor=green,breaklinks=true,pagebackref=false,pdfusetitle]{hyperref}
\usepackage[capitalize,nameinlink]{cleveref}
\usepackage[english]{babel}
\usepackage{comment}
\newcommand{\C}{\mathbb{C}}
\newcommand{\R}{\mathbb{R}}
\newcommand{\N}{\mathbb{N}}

\newcommand{\eps}{\varepsilon}

\newcommand{\calB}{\mathcal{B}}
\newcommand{\calC}{\mathcal{C}}

\newcommand{\calK}{\mathcal{K}}
\newcommand{\calL}{\mathcal{L}}

\newcommand{\calN}{\mathcal{N}}
\newcommand{\calR}{\mathcal{R}}

\newcommand{\calQ}{\mathcal{Q}}

\let\cC\calC
\let\cH\calH

\let\al\alpha

\newcommand{\linspan}{\textnormal{span}}
\newcommand{\Law}{\textnormal{Law}}
\DeclareMathOperator{\gr}{gr}
\DeclareMathOperator{\grdeg}{gr}
\newcommand{\ip}[1]{{\langle{#1}\rangle}}
\let\psdgeq\succeq

\newcommand{\gam}{\gamma}

\newtheorem*{theorem*}{Theorem}
\newtheorem{theorem}{Theorem}[section]
\crefname{theorem}{Theorem}{Theorems}
\newtheorem{lemma}[theorem]{Lemma}
\crefname{lemma}{Lemma}{Lemmas}
\AddToHook{env/lemma/begin}{\crefalias{theorem}{lemma}}
\newtheorem{proposition}[theorem]{Proposition}
\crefname{proposition}{Proposition}{Propositions}
\AddToHook{env/proposition/begin}{\crefalias{theorem}{proposition}}
\newtheorem{corollary}[theorem]{Corollary}
\crefname{corollary}{Corollary}{Corollaries}
\AddToHook{env/corollary/begin}{\crefalias{theorem}{corollary}}
\newtheorem{example}[theorem]{Example}
\crefname{example}{Example}{Examples}
\AddToHook{env/example/begin}{\crefalias{theorem}{example}}

\crefname{question}{Question}{Questions}
\AddToHook{env/question/begin}{\crefalias{theorem}{question}}
\theoremstyle{definition}

\crefname{definition}{Definition}{Definitions}
\AddToHook{env/definition/begin}{\crefalias{theorem}{definition}}
\theoremstyle{remark}

\crefname{remark}{Remark}{Remarks}
\AddToHook{env/remark/begin}{\crefalias{theorem}{remark}}
\crefname{equation}{Eq.}{Eqs.}
\numberwithin{equation}{section}

\DeclarePairedDelimiter\norm{\lVert}{\rVert}

\DeclareMathOperator{\Tr}{Tr}
\DeclareMathOperator*{\argmax}{\arg\max}
\DeclareMathOperator*{\argmin}{\arg\min}
\DeclareMathOperator{\supp}{supp}

\DeclareMathOperator{\diag}{diag}

\DeclareMathOperator{\E}{\mathbb{E}}
\DeclareMathOperator{\He}{He}

\DeclareMathOperator{\OPT}{OPT}
\DeclareMathOperator{\SDP}{SDP}
\DeclareMathOperator{\sign}{sign}

\definecolor{niceish}{HTML}{1cb5bd}

\title{Primal/Dual Method for the Grothendieck Constant}
\author{Steven Heilman\thanks{
USC. \texttt{stevenmheilman@gmail.com}}\and Chris Jones\thanks{UC Davis. \texttt{chijones@ucdavis.edu}} \and Giulio Malavolta\thanks{Bocconi University. \texttt{giulio.malavolta@unibocconi.it}}}

\date{}

\begin{document}
\maketitle

\begin{abstract}
    We determine the hundredths digit of the real Grothendieck constant by proving $1.773 \leq K_G \leq 1.7799$.
    Furthermore, numerical heuristics suggest the estimate $K_G \approx 1.779$.
    We start by simplifying the Gaussian duality theory of $K_G$,
    based on Hermite projection games and Krivine rounding schemes.
    The bounds are obtained by a heuristic primal/dual search for these objects,
    followed by rigorous {\em certification} of their values.
    The search and certification are performed using AI tools, and in particular
    the rigorous certificates are very large (analytical reductions to thousands of numerical inequalities).
\end{abstract}

\section{Introduction}

For a matrix $A \in \R^{m \times n}$, let
\[
\OPT(A)=\max_{x_i,y_j\in\{-1,1\}}
       \sum_{i,j}A_{ij}x_iy_j,
 \qquad
 \SDP(A)=\sup_{\substack{u_i, v_j \in \R^{m+n}\\\norm{u_i}=\norm{v_j}=1}}
       \sum_{i,j}A_{ij}\langle u_i,v_j\rangle
\]
The quantity $\OPT(A)$ is also called the ``infinity-to-1 norm'' of $A$.
We refer to signs $x_i, y_j \in \{-1,1\}$ achieving $\OPT(A)$ as ``infinity-to-1 optimizers'' of $A$.
The real Grothendieck constant $K_G$ is:
\begin{equation}
    K_G=\sup_{m,n \in \N}\sup_{A \in \R^{m \times n}}\frac{\SDP(A)}{\OPT(A)}\,. \label{eq:kg}
\end{equation}
$K_G$ is a basic mathematical constant arising in several fields including functional analysis \cite{pisier2012grothendieck}, quantum information \cite{AcinGisinToner2006}, and approximation algorithms \cite{AlonNaor2006,KhotNaor2012}.
The bound $K_G < \infty$ was first proven by Grothendieck and is now known as ``Grothendieck's inequality''~\cite{grothendieck1956resume}.
For many years it was conjectured that $K_G = \frac{\pi}{2\ln(1+\sqrt{2})}\approx 1.7822$ \cite{Krivine1977} until a breakthrough work by Braverman et al.\ \cite{braverman2013grothendieck} proved that this bound could be lowered by a positive $\eps > 0$. Recently, upper and lower bounds moved by additional epsilons~\cite{heilman2026lower, jones2026grothendieck, heilman2026upper} until Saha et al.\ \cite{saha2026new} proved the macroscopic bounds $1.7135 \leq \frac{6\pi}{11} \leq K_G \leq 1.7819$.
Our main result is to prove that the first three digits of $K_G$ are 1.77.

\begin{theorem}\label{thm:kg-177}
    $1.773 \leq K_G \leq 1.7799$
\end{theorem}
Our main conceptual innovation is a primal/dual method for algorithmically searching for upper and lower bounds on $K_G$. This method provides a new convex duality theory for $K_G$ (\cref{thm:duality}), unifying and simplifying existing results on universal methods for upper \cite{naor2014krivine} and lower \cite{raghavendra2009towards} bounds for $K_G$.
Moreover, running the method produces a pair of objects with nearly-matching value which heuristically suggest $K_G \approx 1.779$.

To convert the primal/dual objects into rigorous bounds on $K_G$,
we need to \emph{certify} their values.
Certification seems difficult: there is a gap between \cref{thm:kg-177}
and the estimate $K_G \approx 1.779$ because
we were not able to certify the instances found by the primal/dual search.
In the proof of \cref{thm:kg-177}, we use simpler instances for which the certification procedure terminated.

The gap between search and certification can be interpreted using computational complexity theory.
A key step in the primal/dual search is to compute maximizers of $\OPT(Q)$,
which is an NP-hard task in general.
While brute-force search can give us a good idea of the value of $\OPT(Q)$,
for a rigorous proof, we need to prove an upper bound $\OPT(Q) \leq C$, called ``certifying'' the instance $Q$.
Deciding $\OPT(Q) \leq C$ is co-NP-hard, hence in general, certification can be as hard as proving that SAT instances are unsatisfiable.
It seems unlikely that the instances $Q$ we obtained would be as difficult to certify as worst-case SAT instances.
However, our certificate for $\OPT(Q)$ is still a ``brute-force proof''
involving extensive numerical inequalities proven using AI tools,
which would have been impossible to obtain by hand.

The main body of the paper contains the general primal/dual theory, then outlines the instances and certificates used for \cref{thm:kg-177}.
The lower and upper bounds are proven in \cref{thm:lower} and \cref{thm:upper}, respectively,
in the appendices with a numerical component on GitHub (see \cref{sec:certificate}).

\subsection{Is there a short formula for \texorpdfstring{$K_G$}{KG}?}

\newcommand{\GW}{\alpha_{\mathrm{GW}}}

Ideally, $K_G$ would have a simple formula like $\sqrt{\pi}$
(this guess is now ruled out by \cref{thm:kg-177}).
Below are formulas for three other constants related to quadratic optimization,
which might suggest what a formula for $K_G$ could look like.

\paragraph{Little Grothendieck constant.}
The little Grothendieck constant $k_G$ is
\begin{align}\label{eq:little-grothendieck}
    k_G = \sup_{n \in \N} \sup_{A \in \R^{n \times n}_{\text{sym} }, A \psdgeq 0} \frac{\SDP(A)}{\OPT(A)}
\end{align}
where $\OPT(A), \SDP(A)$ are as in \cref{eq:kg}. It was proven by Grothendieck that $k_G = \frac{\pi}{2}$ \cite{grothendieck1956resume}.

\paragraph{Goemans--Williamson constant.}
Let $\calL = \bigcup_{n\in\N}\{L \in \R_{\text{sym}}^{n \times n} : L \boldsymbol{1} = 0, L \psdgeq 0, L_{ij} \leq 0 \text{ for }i \neq j\}$ be the set of Laplacian matrices for weighted graphs.
The constant $\GW$ is\footnote{When $A$ is symmetric positive semidefinite, without loss of generality we can add the constraints $x_i = y_i$ to $\OPT(A)$ and $u_i = v_i$ to $\SDP(A)$. The standard definition of $\GW$ uses these symmetric variants.}
\begin{align}\label{eq:gw}
    \GW = \inf_{L \in \calL} \frac{\OPT(L)}{\SDP(L)}\,.
\end{align}
Then $\GW$ admits the well-known formula \cite{goemans1995improved,feige2002optimality}
\[
    \GW = \min_{t \in [-1,1)} \frac{2\arccos t}{\pi(1-t)} \approx 0.87856\,.
\]
The three constants $k_G, \GW, K_G$ are all integrality gaps of the SDP relaxation of $\OPT(A)$, just restricted to different sets of matrices.
The relatively simple formulas for $k_G$ and $\GW$
can be traced to the simplicity of the optimal rounding algorithm (hyperplane rounding) for these special cases. Optimal lower bound instances are not too complicated either ($\Pi_1$ and the Feige-Schechtman instance \cite{feige2002optimality}, respectively).
For $K_G$, these constructions are known not to achieve tight bounds,
and the nearly-tight integrality gap instances we construct appear significantly more complicated.

\paragraph{Parisi constant.}

Another useful comparison is the Parisi constant $P^*$ defined by \cite{parisi1979infinite}
\begin{equation}\label{eq:parisi}
    P^* = \lim_{n \to \infty} n^{-3/2} \E \max_{x_i \in \{-1,1\}}\sum_{i,j=1}^n A_{ij}x_ix_j \approx 1.5266
\end{equation}
where $A \in \R^{n \times n}_{\text{sym}}$ is random with $A_{ij} \overset{\text{i.i.d.}}{\sim} \calN(0,1), i \leq j$ and $A_{ij} = A_{ji}$.
Even though $P^*$ is defined by a high-dimensional optimization problem,
its value is given by the maximum value of a one-dimensional stochastic control problem \cite{mourrat2025inverting} (the famous \emph{Parisi formula} is the dual problem, which expresses $P^*$ as a minimization problem over nondecreasing functions $\gamma:[0,1] \to [0,\infty)$).

The Parisi formula is only semi-explicit, since it involves an explicit convex optimization problem,
but the optimizer $\gam$ apparently does not admit a simpler mathematical description besides being the solution to the optimization problem.
Similarly, if there is very limited structure in optimal constructions for $K_G$, perhaps the best we can hope for is to express $K_G$ as the value of a numerically tractable optimization problem.

\section{Primal/dual method}
\label{sec:duality}

(The proofs for this section are at the end of the section.)

\subsection{Duality theory for \texorpdfstring{$K_G$}{KG}}

It was known from prior work that two types of Gaussian objects
are sufficient to precisely bound $K_G$ from both sides,
\emph{Krivine rounding schemes} for the upper bound \cite{naor2014krivine}
and operators on Gaussian space for the lower bound~\cite{raghavendra2009towards} which we call \emph{Hermite projection games} \cite{jones2026grothendieck}.
We start by giving a simpler, unified proof that these two classes of objects suffice using convex duality.

Let $\gam_{n}$ be ${n}$-dimensional standard Gaussian measure.
A Hermite projection game is an operator on $L_2(\gam_{n})$ of the form, for some coefficients $q_m \in [-1,1]$,
\[
    Q = \sum_{\substack{m = 0}}^\infty q_m \Pi_m^{({n})}
\]
where $\Pi_{m}^{({n})}$ is the projection operator on $L_2(\gam_{n})$ to the span of the degree-${m}$ Hermite polynomials.
{We usually omit the superscript ${n}$ in $\Pi_m^{({n})}$.}
Generalize to bounded operators on $L_2(\gam_n)$ the notation,
\[
 \OPT(Q):=\sup_{f,g: \R^n \to \{\pm 1\}}\langle f,Qg\rangle,\qquad
 \SDP(Q):=\sup_{k\in\N}\sup_{\substack{f,g: \R^n \to \R^k\\ \norm{f(X)} = \norm{g(X)} = 1}}
\sum_{j=1}^k\langle f_j,Qg_j\rangle\,.
\]
The operator $Q$ is a witness for the lower bound $K_G \geq \SDP(Q)/\OPT(Q)$, after passing from operators to matrices by standard discretization arguments.
We also write
\[
    v_{n}(q) = \OPT(Q) = \sup_{f,g : \R^n \to \{\pm 1\}} \sum_{\substack{m = 0}}^\infty q_m \langle{\Pi_m f, \Pi_m g}\rangle, \qquad v_\infty(q) = \sup_{n \in \N} v_n(q)\,.
\]
We will always take the limit $n \to \infty$ when calculating the lower bound proven by $Q$.

A Krivine rounding scheme is specified by a probability distribution
$\mu$ on pairs of Boolean functions $f,g: \R^{n} \to \{-1,1\}$.
The rounding algorithm associated to $\mu$ is:
given input unit vectors with arbitrary dimension, written $u,v \in \ell_2, \norm{u} = \norm{v} = 1$, then
\begin{equation}\label{eq:krivine-scheme}
    \text{sample }(f,g) \sim \mu \text{ and } G_{ij} \overset{\text{i.i.d.}}{\sim} \calN(0,1)\text{ for } i \in [n], j \in \N \text{ then output } f(Gu), g(Gv) \in \{\pm 1\}\,.
\end{equation}

The rounding algorithm for $\mu$ proves an upper bound on $K_G$ in terms of an associated ``kernel function'' $H_\mu$.
For two inputs with $\ip{u,v} = t$ the expected correlation between the output signs is
\begin{align*}
    H_\mu(t) &:= \E_{(f,g) \sim \mu} \E_{(X,Y) \sim t\text{-correlated Gaussians}}f(X)g(Y)\\
    \intertext{Let $c(\mu)=(c_m(\mu))_{m\in\N}$ be the Taylor coefficients of $H_\mu(t)$. The Taylor coefficients have the formula}
    H_\mu(t) &= \sum_{m=0}^\infty c_m(\mu) t^m = \sum_{m=0}^\infty \E_{(f,g) \sim \mu} \langle\Pi_m f, \Pi_m g \rangle t^m\,.
\end{align*}
We also write $c(f,g)$ and $H_{f,g}(t)$ for the case where $\mu$ is $(f,g)$ with probability 1.

\begin{lemma}\label{lem:ub-easy}
    Let $n \in \N$ and let $\mu$ be a distribution on pairs of odd functions $f,g: \R^n \to \{\pm 1\}$.
    Then $K_G^{-1} \geq c_1(\mu) - \sum_{\textnormal{odd }m \geq 3} |c_m(\mu)|$.
\end{lemma}

We prove using convex duality that these two methods suffice to bound $K_G$ from both sides.

\begin{theorem}
\label{thm:duality}
    Let $\mathcal C\subseteq\ell_1$ be the set of sequences $c(\mu)$, where ${n}\in\N$ and $\mu$ ranges over distributions on pairs of odd functions $f,g:\R^{n}\to\{\pm1\}$.
    Let $\calQ = \{q \in \ell_\infty: q_1 = 1, \forall m \in \N. \; |q_m| \leq 1, q_{2m} =0\}$. 
    Then $v_\infty(q) = \sup_{c \in \calC} \ip{c, q}$ and
    \begin{equation}\label{eq:duality1}
        \frac{1}{K_G} = \sup \Big\{c_1 - \sum_{\textnormal{odd }m \geq 3} |c_m|: c \in {\mathcal{C}}\Big\} = \inf\{{v_\infty}(q): q \in \calQ\}\,.
    \end{equation}
    Furthermore, the infimum over $q$ is attained.
\end{theorem}

\begin{samepage}
\begin{corollary}\hspace{0.1cm}\label{cor:krivine-schemes}
\begin{enumerate}[label=(\roman*)]
    \item (Hermite projection games attain the lower bound \cite{raghavendra2009towards})\footnote{\cref{cor:krivine-schemes} strengthens \cite[Theorem 1.4]{raghavendra2009towards} by proving that the infimum over $q \in \calQ$ is attained.} There is $Q = \sum_{\textnormal{odd }m\ge1}q_m \Pi_m^{(n)}$ such that $q_1 = 1$ and $\SDP(Q)/\OPT(Q) = K_G - o_{n}(1)$.
    \item  (Krivine schemes attain the upper bound \cite{naor2014krivine}) For every $\eps>0$, there exist ${n}\in\N$ and a distribution $\mu$ on pairs of odd functions $f,g:\R^{n}\to\{-1,1\}$ such that \[ c_1(\mu)-\sum_{\substack{m\ge3\\m\textnormal{ odd}}}|c_m(\mu)| \ge\frac1{K_G+\eps}.\]
\end{enumerate}
\end{corollary}
\end{samepage}

Via complementary slackness, convex duality implies that for any near-optimal pair of $Q$ and $\mu$, most
of the probability mass of $\mu$ is on near-optimizers of $Q$.
Hence hard lower bound instances $Q$ are also the source of good rounding algorithms,
via their infinity-to-1 optimizers.

\begin{proposition}\label{prop:complementary-slackness}
    Let $q \in \calQ$ such that ${v_\infty}(q) = \frac{1}{K_G}$.
    Let $\mu$ be a distribution on $f,g : \R^{n} \to \{-1,1\}$ such that $c_1(\mu) \geq \frac{1}{K_G} - \eps$ and $\sum_{\textnormal{odd }m \geq 3} |c_m(\mu)| \leq \eps$.
    Then $\E_{(f,g) \sim \mu}\ip{f, Qg} \geq \frac{1}{K_G} - 2\eps.$
\end{proposition}

\subsection{Numerical primal/dual search}

By \cref{thm:duality},
\begin{equation}\label{eq:duality-abstract}
    \frac{1}{K_G} = \sup_{c \in \mathcal{C}} \inf_{q \in \mathcal{Q}} \langle c, q\rangle = \inf_{q \in \mathcal{Q}} \sup_{c \in \mathcal{C}} \langle c, q \rangle
\end{equation}
This suggests that we should search for extremizers $c,q$ using a primal/dual algorithm, using the alternating update rules (assuming the extrema are attained),
\begin{equation}\label{eq:abstract-primal-dual}
    q = \argmin_{q \in \mathcal{Q}} \langle c, q\rangle, \qquad \qquad c = \argmax_{c \in \mathcal{C}} \langle c, q \rangle\,.
\end{equation}

A very important caveat is that there is probably no efficient implementation of this algorithm.
Calculating $\argmax_{c \in \mathcal{C}} \langle c, q \rangle$ corresponds
to calculating the infinity-to-1 optimizers of $Q = \sum_i q_i \Pi_i$
which is an NP-hard task in general.
Therefore, we use heuristics to search for
good infinity-to-1 optimizers of $Q$.

We use a slight variant of \cref{eq:abstract-primal-dual}. Fix a parameter $D$. The iteration has two steps.
\begin{enumerate}[label=\arabic*.]
    \item In the first step, fix $Q$ and heuristically compute its infinity-to-1 optimizers $(f_i, g_i)_{i \in I}$.
    Add $(f_i, g_i)_{i \in I}$ to a pool of rounding functions, denoted $\mathcal{R}$.
    Heuristically calculating the $\infty$-to-1 optimizers is very slow and is guided by AI-heuristics.

    \item In the second step, fix the pool $\mathcal{R}$ and re-compute $Q$ to minimize the worst value among the rounding functions in the pool.
    Let $\calQ_{\le D} = \{q \in \calQ: \forall m > D. \; q_m = 0\}$.
    \begin{equation}\label{eq:objective-q}
        q= \argmin_{\substack{q\in\mathcal Q_{\le D}}} \max_{(f,g) \in \mathcal{R}} \sum_{\substack{m = 1\\m\text{ odd}}}^{D} q_m c_m(f,g)\qquad Q=\sum_{m=1}^Dq_m\Pi_m\,.
    \end{equation}
    This computation can be formulated as a linear program parameterized by the coefficients $(c_{m}(f,g))_{1 \leq m \leq D}$ for $(f,g) \in \mathcal{R}$.
\end{enumerate}
\noindent
At the end of the iteration, calculate $\mu$ by the formula
\begin{equation}\label{eq:objective-mu}
    \mu = \argmax_{\supp(\mu) \subseteq \calR}c_1(\mu) - \sum_{m \geq 3}^{D} |c_m(\mu)|\,.
\end{equation}
This can be formulated as a linear program in $\mu$ again parameterized by $(c_{m}(f,g))_{1 \leq m \leq D}$.

Running the algorithm leads to a pair of candidates with heuristic primal and dual estimates of $K_G$ matching on the first four digits $K_G \approx 1.779$.
See \cref{sec:numerics} for numerical details of the run.

\subsection{Proofs for the duality theory}

\begin{lemma}\label{lem:hfg-analytic}
    Let $f,g : \R^{n} \to \{\pm1\}$ be measurable. Then $\sum_m |c_m(f,g)| \leq 1$. Thus, $\calC \subseteq \ell_1$.
    The series for $H_{f,g}$ and $H_\mu$ converge absolutely
    and uniformly for $t \in \C, |t|\le1$ and are holomorphic for $|t|<1$.
\end{lemma}
\begin{proof}
    Recall $c_i(f,g) = \ip{\Pi_i f, \Pi_i g}$.
    Using Cauchy--Schwarz twice,
    \[
        \sum_k |\ip{\Pi_k f, \Pi_k g}| \leq \sum_{k} \norm{\Pi_k f}_2 \norm{\Pi_k g}_2 \leq \left(\sum_k \norm{\Pi_k f}_2^2\right)^{1/2} \left(\sum_k \norm{\Pi_k g}_2^2\right)^{1/2} = \norm{f}_2 \norm{g}_2 = 1\,.
    \]
    Then for a general $\mu$ we have $\sum_m|c_m(\mu)|
        \le\E_{(f,g)\sim\mu}\sum_m|c_m(f,g)|\le1$.
    Absolute summability gives the stated convergence and holomorphy.
\end{proof}

\begin{lemma}\label{lem:odd-rounding}
Let $q \in \calQ$ and $Q = \sum_{\textnormal{odd }m \geq 1}q_m \Pi_m$. Then
\[
    \OPT(Q) = \sup_{f,g : \R^n \to \{\pm 1\}} \langle f, Qg\rangle = \sup_{\textnormal{odd }f,g : \R^n \to \{\pm1\}} \langle f, Qg\rangle\,.
\]
\end{lemma}
\begin{proof}
    Given $f,g : \R^n \to \{\pm 1\}$ their odd parts $f_{\text{odd}}(X) = \frac{f(X) - f(-X)}{2}$ and $g_\text{odd}(X) = \frac{g(X) - g(-X)}{2}$
    satisfy $\langle f, Q g\rangle  = \langle f_{\text{odd}}, Qg_{\text{odd}}\rangle$.
    Furthermore, $f_{\text{odd}}(X), g_{\text{odd}}(X) \in \{-1,0,1\}$.

    Fix any odd Boolean function $a$, for example $a(X)=\sign(X_1)$. Define
    \[
     f_\pm=f_{\mathrm{odd}}\pm a\,\mathbf1_{\{f_{\mathrm{odd}}=0\}},
     \qquad
     g_\pm=g_{\mathrm{odd}}\pm a\,\mathbf1_{\{g_{\mathrm{odd}}=0\}}.
    \]
    These functions are odd and Boolean, and their respective averages
    are $f_{\mathrm{odd}}$ and $g_{\mathrm{odd}}$. Hence, $\ip{f_\pm, Qg_\pm} \geq \ip{f_{\text{odd}}, Qg_{\text{odd}}}$ for
    at least one of the four choices.
    Taking suprema proves the claim.
\end{proof}

\begin{proof}[Proof of \cref{lem:ub-easy}.]
    Let $A \in \R^{m \times n}$ and let $u_i, v_j \in \ell_2, \norm{u_i} = \norm{v_j} = 1$ be an optimal SDP solution.
    Taking tensor powers shows
    \begin{equation}\label{eq:tensor}
        \left|\sum_{i,j}A_{ij}\ip{u_i,v_j}^k\right|
        =\left|\sum_{i,j}A_{ij}
        \ip{u_i^{\otimes k},v_j^{\otimes k}}\right|
        \le\SDP(A)\qquad(k\ge1).
    \end{equation}
    Let $\tilde{x}_i, \tilde{y}_j \in \{\pm 1\}$ be the randomized signs output by Gaussian rounding in \cref{eq:krivine-scheme}.
    We have
    \[
        \OPT(A) \geq \sum_{i,j}A_{ij}\E \tilde x_i \tilde y_j = \sum_{i,j}A_{ij}H_\mu(\ip{u_i,v_j})
        \overset{\text{\cref{eq:tensor}}}{\ge}\left(c_1(\mu)
        -\sum_{\substack{m\ge3\\m\textnormal{ odd}}}|c_m(\mu)|\right)
        \SDP(A).
    \]
    Therefore $\frac{\OPT(A)}{\SDP(A)} \geq c_1(\mu) - \sum_{\text{odd }m \geq 3} |c_m(\mu)|$ for all $A \ne 0$ which implies the same bound on $K_G^{-1}$.
\end{proof}

\begin{proof}[Proof of \cref{thm:duality}.]
    By \cref{lem:odd-rounding}, $v_\infty(q) = \sup_{c \in \calC} \ip{c,q}$. We prove
    \[
        \frac{1}{K_G} \leq \inf\{v_\infty(q) : q \in \calQ\} = \sup\Big\{c_1 - \sum_{\text{odd }m \geq 3}|c_m| : c \in \cC\Big\} \leq \frac{1}{K_G}\,.
    \]
    The right inequality is \cref{lem:ub-easy}.
    The middle equality is convex duality. Indeed, $c_1 - \sum_{\text{odd }m \geq 3}|c_m| = \inf_{q \in \calQ} \ip{c,q}$ for all $c \in \ell_1$.
    This implies the weak duality $\sup_{c \in \cC} \inf_{q \in \calQ} \ip{c,q} \leq \inf_{\substack{q \in \calQ}} \sup_{c \in \cC} \ip{c,q}$.
    Strong duality holds by Sion's minimax theorem: $\cC$ is convex, $\calQ$ is convex and compact in the weak-$^*$ topology on $\ell_\infty = (\ell_1)^*$,
    and $\ip{c,q}$ is affine and continuous in each variable separately.
    
    For the left inequality, let $q \in \calQ$ and let $Q = \sum_m q_m \Pi_m^{(n)}$.
    Then $\OPT(Q) =  v_n(q) \leq {v_\infty}(q)$.
    On the other hand, $\SDP(Q) \geq 1-o_{n}(1)$ by $q_1 = 1$ and a standard estimate of the SDP value of the vector solution $f(x) = g(x) = \frac{x}{\norm{x}}$.
    Thus $\SDP(Q)/\OPT(Q) \geq \frac{1 - o_{n}(1)}{{v_\infty}(q)}$ and $K_G \geq \frac{1}{{v_\infty}(q)}$ as desired.

    We have that $v_\infty(q) = \sup_{c \in \calC}\ip{c,q}$ is weak-$^*$ lower semicontinuous since it is a supremum of weak-$^*$ continuous functions.
    Then it attains a minimum on the weak-$^*$ compact set $\calQ$.
\end{proof}

\begin{proof}[Proof of \cref{cor:krivine-schemes}.]
    (Hermite projection games) Let $q \in \calQ$ attain the infimum in \cref{eq:duality1}. Then $Q  = \sum_m q_m \Pi_m^{(n)}$ has $\SDP(Q) \geq 1- o_n(1)$ as in the previous proof, and $\OPT(Q) = v_n(q) \leq v_\infty(q)$. Therefore $\SDP(Q)/\OPT(Q) \geq \frac{1 - o_{n}(1)}{{v_\infty}(q)} = K_G - o_{n}(1)$.

    (Krivine rounding schemes) By definition of $\calC$ and \cref{thm:duality}.
\end{proof}

\begin{proof}[Proof of \cref{prop:complementary-slackness}.]
    The assumptions on $\mu$ with $q_1 = 1, |q_m| \leq 1$ imply $\ip{c(\mu), q} = \sum_{i} c_i(\mu) q_i \geq K_G^{-1} - 2\eps$.
    On the other hand,
    \begin{align*}
        \sum_{i}c_i(\mu) q_i &= \E_{(f,g) \sim \mu} \sum_i q_i\ip{\Pi_i f, \Pi_i g} = \E_{(f,g) \sim \mu} \ip{f, Qg}\,.
    \end{align*}
    We conclude that $\E_{(f,g) \sim \mu} \ip{f, Qg} \geq K_G^{-1} - 2\eps$.
\end{proof}

\section{Asymptotic Boolean functions}
\label{sec:asymptotic}

In this section, we lift Hermite projection games and Krivine rounding schemes
to an explicit asymptotic space in the limit ${n} \to \infty$.
We call ``asymptotic Boolean functions'' the limits of functions $f:  \R^{n} \to \{-1,1\}$.
These frequently arise as $\infty$-to-1 optimizers and therefore are also used in the numerical primal/dual search.
Asymptotic Boolean functions are closely related to the ``limiting Krivine schemes'' used by Saha et al.\ \cite{saha2026new}.

\begin{example}\label{ex:pi-d}
Consider the Hermite projection operator $\Pi_{d}^{(n)}$ with ${d} > 1$.
As $n \to \infty$, the $\infty$-to-1 norm of $\Pi_d$ is asymptotically achieved only by functions $f,g : \R^n \to \{\pm 1\}$ using a diverging number of variables, such as
\begin{equation}\label{eq:pi-d}
    f(X) = g(X) = \sign\left(\frac{1}{\sqrt{{n}}}\sum_{i = 1}^{n} \He_{d}(X_i)\right)
\end{equation}
where $\He_d$ are the Hermite polynomials.
It can be shown that $f,g$ in \cref{eq:pi-d} are asymptotic optimizers of $\OPT(\Pi_{d}^{(n)})$.
Calculating their value shows $\SDP(\Pi_{d})/\OPT(\Pi_{d}) = \frac{\pi}{2} - o_n(1)$. By the little Grothendieck inequality,
$\frac{\pi}{2}$ is the maximum possible integrality gap for PSD matrices.
The standard construction for little Grothendieck uses ${d}=1$ but it turns out that arbitrary ${d}$ has the same gap.
\end{example}

We use the following standard notions from Gaussian analysis.
Let $X$ be a set of variables. Sample $X_i \overset{\text{i.i.d.}}{\sim} \calN(0,1)$.
Let $\He_k$ be the unnormalized probabilists' Hermite polynomials.
Let $\psi_k = \He_k/\sqrt{k!}$ be their normalizations.
For a finitely-supported multi-index $\al \in \N^X$
let $X^\alpha = \prod_i X_i^{\alpha_i}$ and similarly for $\He_\alpha$ and $\psi_\alpha$. Let
$|\alpha|=\sum_i\alpha_i$ and $\alpha! = \prod_i \alpha_i!$.
Then $\{\psi_\alpha(X) : \alpha \in \N^X, |\alpha|<\infty\}$ is an orthonormal basis of
$L_2(X)$.
Let $\widehat f_\alpha=\E[f\psi_\alpha]$.
Let $\cH_{d}(X)$ be the closed span of the degree-${d}$ Hermite polynomials in $X$, which is also called the ``degree-${d}$ Gaussian chaos'' for $X$.

\subsection{Graded Gaussian space}

The idea of the asymptotic limit is to replace the summation
in \cref{eq:pi-d} by a new Gaussian variable, since it is asymptotically Gaussian.
In the case of \cref{eq:pi-d}, we assign the new Gaussian variable ``grade $d$''
to represent that it is a degree-$d$ polynomial in the ``original'' variables.

Let $X$ be a set of variables such that each variable has an associated positive integer \emph{grade}.
Sample $X_i \overset{\text{i.i.d.}}{\sim} \calN(0,1)$ independently across all variables regardless of grade.
We will notate variables in $X$ as $X_i^{({d})}$ where ${d}$ indicates the grade.

Let $\grdeg(\alpha)$ be the total grade of a multi-index $\al \in \N^X$.
Polynomials in $X$ now have two notions of degree, their standard degree,
and their graded-degree, or total grade.
The Hermite polynomial $\He_\alpha(X)$ is defined as before, and $\psi_\alpha(X)$ are still an orthonormal basis for $L_2(X)$.

The following notions behave differently on the graded space.
Let $\Pi_i$ be the operator on $L_2(X)$ which projects to the span of graded-degree $i$ Hermite polynomials.
On the graded space, oddness means
$f(\widetilde X) = -f(X)$ under the graded sign transformation $\widetilde X_i^{({d})} = (-1)^{d} X_i^{({d})}$.
An odd function satisfies $\Pi_{2m} f = 0$ since by oddness and $\psi_\alpha(\widetilde X)=(-1)^{\grdeg(\alpha)}\psi_\alpha(X)$ we have for all even $\gr(\alpha)$,
\begin{equation}\label{root:odd-hermite}
 \E[f(X)\psi_\alpha(X)]
 =-(-1)^{\grdeg(\alpha)}\E[f(X)\psi_\alpha(X)]
 =0.
\end{equation}

Let $W_{d}(f) = \norm{\Pi_{d} f}_2^2$.
For variables with grade $k$, the Ornstein--Uhlenbeck noise operator acts by sampling $t$-correlated Gaussians $Y$ with covariance $\E[X_i Y_i] = t^k$.
Under these definitions, we still have the noise correlation and Parseval identities,
\begin{equation}\label{eq:correlation}
 \E_{(X,Y) \sim t\text{-correlated}}[f(X)g(Y)]=\sum_{{d}\ge0}t^{d}\langle\Pi_{d} f,\Pi_{d} g\rangle\,, \qquad \qquad \E[f(X)^2] = \sum_{{d} \geq 0} W_{d}(f)\,.
\end{equation}

Now set $X$ to be a collection of countably many variables of each graded degree.
Write $X_1, X_2, X_3,\dots$ for the grade-1
variables in $X$ and $X^{({d})}_1, X_2^{({d})}, X_3^{({d})},\dots$ for the (independent) grade-${d}$ variables in $X$.
Define the asymptotic limit spaces
\begin{equation}\label{eq:limit-spaces}
    \overline{\cH}_{d} = \overline{\linspan}\left(\He_\al(X) : \grdeg(\al) = {d}\right)\,, \qquad  \qquad \overline{L_2} = L_2(X)\,.
\end{equation}
The notation $\overline{\linspan}$ in \cref{eq:limit-spaces} denotes the closure of the linear span in $L_2(X)$ which is distinct from the notation $\overline{\cH}_{d}$ and $\overline{L_2}$.
We write either $f$ or $f(X)$ for $f \in \overline{L_2}$.

As shown in \cref{ex:pi-d}, these spaces contain
natural limits of Gaussian functions.
More interestingly, in the sense of laws, $\overline{\cH}_{d}$ contains \emph{all possible} limits
of degree-${d}$ Gaussian polynomials.
We prove this in another paper by one of the current authors \cite{jones2026regularity}.

\begin{theorem}[{\cite{herry2024limit, jones2026regularity}}]\label{thm:closure}
Let $d_i\in \N$ for $i \in \N$.
\begin{enumerate}[label=(\roman*)]
    \item (Closure of the asymptotic space) Let $p_{n,i}\in\overline{\cH}_{{d}_i}$ such that $\sup_n\norm{p_{n,i}}_2<\infty$ for all $i$.
    Then there exists a subsequence and $p_i\in\overline{\cH}_{{d}_i}$ such that, along the subsequence, $(p_{n,i})_{i \in \N}\xrightarrow{d}(p_i)_{i \in \N}$
    with convergence of all mixed (finite) moments.
    Consequently, $\{\Law(p):p\in\overline{\cH}_{d}\}$ is weakly closed.

    \item (Approximability by finitary functions) For all $p_i \in \overline{\cH}_{d_i}$
there exists a sequence $p_{n,i} \in \cH_{d_i}(X_1, \dots, X_{n})$ such that $(p_{n,i})_{i \in \N} \xrightarrow{d} (p_i)_{i \in \N}$ with convergence of all mixed (finite) moments.
\end{enumerate}
\end{theorem}

As a corollary, $\overline{L_2}$ contains
all limits of Gaussian functions with uniform tail bounds on their Hermite coefficients.

\subsection{Asymptotic Boolean functions on Gaussian space}

An asymptotic Boolean function is $f \in \overline{L_2}$ such that $f \in \{-1,1\}$ almost surely.
While this contains all finitary Boolean functions $f : \R^{n} \to \{-1,1\}$,
it also contains asymptotic functions which use graded variables,
such as $f(X) = \sign(X_1^{(2)})$ and $f(X) = \sign(X_1^{(1)}X_1^{(2)} - X_1^{(3)})$.

Extend the notations $c(\mu), c(f,g), H_\mu(t), H_{f,g}(t)$ using $c_m(f,g) = \langle \Pi_m f, \Pi_m g\rangle$.
The proof of \cref{lem:hfg-analytic} still gives
$\sum_m |c_m(\mu)|\le1$,
and hence uniform convergence of $H_{\mu}$ on \(|t|\le1\) and holomorphy on \(|t|<1\).

In the remainder of the section, we prove several Lemmas \ref{lem:asymptotic-c}, \ref{lem:maximizers-exist}, \ref{cor:krivine-schemes-exact} showing that maximizers are attained within the class of asymptotic Boolean functions.

\begin{lemma}\label{lem:asymptotic-c}
Let $\mathcal K \subseteq \ell_1$ be the set of $c(\mu)$ for distributions $\mu$
on pairs of odd $f,g \in \overline{L_2}, f(X),g(X) \in \{\pm 1\}$.
Then $\mathcal K=\overline{\mathcal C}$ where $\overline{\calC}$ is the closure of $\calC$ in $\ell_1$.
\end{lemma}

\begin{proof}
We use the auxiliary set
\[
 \mathcal B=
 \{c(f,g):f,g\in\overline{L_2}\text{ odd},\ |f|,|g|\le1\}.
\]

First we claim $\mathcal B \subseteq \mathcal K$.
Let $f,g \in \overline{L_2}$ odd, $|f|,|g|\le 1$.
We use an odd version of threshold rounding.
Let $a$ be any odd Boolean function, such as $a(X)=\sign X_1$.
For $s\in[-1,1]$, define
\[
 R_s h(X)=a(X)\sign\bigl(a(X)h(X)-s\bigr),
\]
with $\sign(0)=1$. If $h$ is odd and bounded by $1$, then
$R_s h$ is odd and Boolean, and $\E_sR_s h=h$ for uniform
$s\in[-1,1]$. Thus, for independent uniform $s,t$ we have
\begin{equation}\label{eq:odd-threshold-rounding}
 \E_{s,t}c(R_s f,R_t g)=c(f,g).
\end{equation}
Thus $\calB \subseteq \calK$. The same proof shows that if odd $|f|,|g| \leq 1$ depend on only finitely many grade-one variables, then $c(f,g) \in \mathcal C$.

We next prove compactness of $\mathcal B$.
Take a sequence of odd $(f_n,g_n), |f_n|,|g_n| \leq 1$. By \cref{thm:closure}(i),
after passing to a subsequence there are
$p_d,r_d\in\overline{\mathcal H}_d$ such that
\begin{equation}\label{eq:ou-1}
 (\Pi_d f_n,\Pi_d g_n)_{d \in \N}
 \xrightarrow{d}(p_d,r_d)_{d\in \N}
\end{equation}
with convergence of the mixed moments of every finite
subtuple. For each fixed $D$,
\[\sum_{d \le D} \|p_d\|_2^2 = \lim_{n \to \infty} \sum_{d \le D} \|\Pi_d f_n\|_2^2 \overset{\text{Parseval}}{\leq} 1\,.\]
Taking sup over $D$ shows $\sum_d\|p_d\|_2^2\le1$.
So $f=\sum_dp_d$ is well-defined in $\overline{L_2}$ and $\Pi_d f = p_d$.
The same holds for $g$.
Both are odd since their even-degree components vanish.

Note that we may \emph{not} have $f_n \xrightarrow{d} f$ since Hermite mass
for $f_n$ could escape to higher and higher degree.
We still prove $|f|,|g| \leq 1$. Fix $0<r<1$ and define the graded Ornstein--Uhlenbeck operator $T_r = \sum_d r^d \Pi_d$.
\cref{eq:ou-1} together with
\[
 \sup_n\left\|T_rf_n-\sum_{d\le D}r^d\Pi_df_n\right\|_2
 \le r^{D+1},
\]
implies $T_rf_n\xrightarrow{d}T_rf$.
For the nonnegative continuous function
$\varphi(x)=(|x|-1)_+$, the Portmanteau inequality gives
\[
 \E\varphi(T_rf)
 \le\liminf_{n\to\infty}\E\varphi(T_rf_n)=0,
\]
since $|T_rf_n|\le1$. Finally, $T_rf\to f$ in $L_2$
as $r\uparrow1$, and $\varphi$ is $1$-Lipschitz, so
$\E\varphi(f)=0$. Hence $|f|\le1$ almost surely.
The argument for $g$ is identical.

We have
\[c_d(f_n,g_n) = \E[(\Pi_d f_n)(\Pi_d g_n)] \underset{n \to\infty}{\longrightarrow} \E[(\Pi_d f)(\Pi_d g)] = c_d(f,g)
\]
Thus every sequence in $\calB$ has a coordinatewise convergent subsequence.
Since the coordinatewise topology is metrizable, sequential compactness implies $\calB$ is (topologically) compact.

We now prove
$\mathcal B\subseteq\overline{\mathcal C}$.
Fix odd $(f,g), |f|,|g| \leq 1$, and put
\[
 f^{(D)}=\sum_{d\le D}\Pi_df,\qquad
 g^{(D)}=\sum_{d\le D}\Pi_dg,\qquad
 \varepsilon_D=
 \|f-f^{(D)}\|_2+\|g-g^{(D)}\|_2.
\]
Then $\varepsilon_D\to0$ as $D \to \infty$.
\cref{thm:closure} (ii) gives odd finite-dimensional
polynomials $F_n,G_n$ of degree at most $D$ such that, also with convergence of all mixed moments,
\[
 (\Pi_dF_n,\Pi_dG_n)_{d\le D}
 \xrightarrow{d}
 (\Pi_df,\Pi_dg)_{d\le D}.
\]
In particular,
\[
 c(F_n,G_n)\longrightarrow c(f^{(D)},g^{(D)})
 \quad\text{in }\ell_1,
 \qquad
 \|F_n\|_2\to\|f^{(D)}\|_2\le1,
\]
and similarly for $G_n$.

Let $\chi(u)=\max\{-1,\min\{u,1\}\}$ be the clipping function to $[-1,1]$.
Convergence in distribution and uniformly bounded fourth
moments imply
 \begin{align}\label{eq:clip1}
 \|F_n-\chi(F_n)\|_2
 &\longrightarrow
 \|f^{(D)}-\chi(f^{(D)})\|_2
 \leq\|f^{(D)}-f\|_2,\\ \notag
 \|G_n-\chi(G_n)\|_2
 &\longrightarrow
 \|g^{(D)}-\chi(g^{(D)})\|_2
 \leq\|g^{(D)}-g\|_2.
 \end{align}
The right inequalities use that clipping is the pointwise projection onto $[-1,1]$.
Orthogonality and Cauchy--Schwarz give the general continuity relations,
\begin{equation}\label{eq:coefficient-continuity}
 \begin{aligned}
 \|c(f,g)\|_1&\le\|f\|_2\|g\|_2,\\
 \|c(f,g)-c(f',g')\|_1
 &\le\|f-f'\|_2\|g\|_2+\|f'\|_2\|g-g'\|_2.
 \end{aligned}
\end{equation}

Using \cref{eq:clip1} and \eqref{eq:coefficient-continuity}, we obtain
\[
 \limsup_{n \to \infty}
 \|c(\chi(F_n),\chi(G_n))-c(f,g)\|_1
 \le2\varepsilon_D.
\]
The clipped functions are bounded by 1, odd, and finite-dimensional,
so their coefficient sequences belong to $\mathcal C$ by
the sentence after \eqref{eq:odd-threshold-rounding}. Letting $D\to\infty$
proves the inclusion.

It remains to pass from $\calB$ to $\calK$.
The set $\mathcal K$ is exactly the set of barycenters of
probability measures on $\mathcal B$.
Indeed, if $c(\mu) \in \calK$, then $c(\mu) = \E_{(f,g) \sim \mu}c(f,g)$
is also a mixture of elements $c(f,g) \in \calB$.
Conversely, a probability
measure on $\mathcal B$ can be lifted\footnote{
The space of bounded odd pairs is a closed subset of
$\overline{L_2}^{\,2}$, hence Polish. Its coefficient map
onto $\mathcal B$ is continuous by
\eqref{eq:coefficient-continuity}. The Jankov--von Neumann
theorem gives a section measurable for the completion of
the given probability measure, which supplies the lift.}
to a measure $\mu$ on odd pairs $(f,g), |f|,|g| \leq 1$,
and then threshold rounding in \eqref{eq:odd-threshold-rounding} ensures $f,g \in \{\pm1\}$ without changing $c(\mu)$.

Probability measures on the compact metric space
$\mathcal B$ form a compact space for weak convergence.
Their barycenters depend continuously on each coordinate,
so $\mathcal K$ is coordinatewise compact and convex.

The inclusion map $\mathcal B\to\ell_1$ is Borel measurable,
and $\|c\|_1\le1$ on $\mathcal B$, so these barycenters are
also Bochner integrals in $\ell_1$. Since
$\overline{\mathcal C}$ is closed and convex
and contains $\mathcal B$, it contains every such
barycenter. Hence
$\mathcal K\subseteq\overline{\mathcal C}$.
Conversely, $\mathcal C\subseteq\mathcal K$, and
coordinatewise compactness makes $\mathcal K$ closed
in $\ell_1$. Therefore
$\mathcal K=\overline{\mathcal C}$.
\end{proof}

We prove that the asymptotic optimum of a finitely-supported Hermite projection game is achieved by an asymptotic Boolean function.
We call these ``asymptotic $\infty$-to-1 optimizers'' of $Q$.
Moreover, there exist optimizers which are $\overline{\cH}_{\le D}$ threshold functions (essentially, polynomial threshold functions).

\begin{lemma}\label{lem:maximizers-exist}
    Let $q \in \ell_\infty$ with $q_{2{m}}=0$ for ${m}\ge0$ and assume $\max\supp(q) = D< \infty$.
    Then
    \[
        {v_\infty}(q) = \max\left\{\sum_{m = 0}^\infty q_m \langle \Pi_m f, \Pi_m g\rangle : f,g \in \overline{L_2}, f(X), g(X) \in \{\pm 1\} \textnormal{ a.s.} \right\}\,.
    \]
    Furthermore, there exist optimizers of the form $f = \sign(p)$ and $g = \sign(r)$
    for $p,r \in \overline{\cH}_{\le D}$.
\end{lemma}

\begin{proof}
Let $Q=\sum_dq_d\Pi_d$ on $\overline{L_2}$. First we claim (this still holds if $\max \supp(q) = \infty$)
\begin{equation}\label{eq:asymptotic-payoff}
 \OPT(Q)=\sup_{c\in\mathcal B}\langle c,q\rangle=v_\infty(q).
\end{equation}
The first equality follows since taking the odd parts of candidates $f,g$ for $\OPT(Q)$ preserves their value, and $f,g \in [-1,1]$ can be replaced by a distribution of $f,g \in \{\pm 1\}$ using the threshold rounding in \eqref{eq:odd-threshold-rounding}.
For the second equality, the inclusion in
\cref{lem:asymptotic-c} and $\ell_1$-continuity of the pairing
give the upper bound
$\sup_{c \in \mathcal B}\langle c,q\rangle
\le\sup_{c \in \mathcal C}\langle c,q\rangle=v_\infty(q)$.
For the reverse inequality, a sequence of $c(\mu^{(k)}) \in \calC$ approaching $v_\infty(q)$
has that $c(\mu^{(k)}) = \E_{(f,g) \sim \mu^{(k)}}c(f,g)$ is an average
of $c(f,g) \in \mathcal{B}$,
so the inequality follows.

When $q$ has finite support, the map $c \mapsto \langle c, q \rangle$ is coordinatewise
continuous. Compactness of $\mathcal B$ and
\eqref{eq:asymptotic-payoff} therefore give $c \in \calB$ achieving value $v_\infty(q)$.
Then any pair $(f,g), |f|, |g| \leq 1$ corresponding to $c$ achieves the maximum value $v_\infty(q)$.

To see that $f,g$ can be polynomial threshold functions, take two successive best responses:
\[
 p=Qg,\qquad f=\sign(p),\qquad
 r=Qf,\qquad g=\sign(r).
\]
Neither replacement decreases the value, and
\eqref{eq:asymptotic-payoff} says that no Boolean pair can
exceed $v_\infty(q)$. Hence $(f,g)$ is optimal.
Both $p$ and $r$ have graded degree at most $D$.
\end{proof}

\begin{lemma}\hspace{0.1cm}\label{cor:krivine-schemes-exact}
\begin{enumerate}[label=(\roman*)]
    \item There is an operator $Q = \sum_{\textnormal{odd }m \geq 1}q_m \Pi_m$ on $\overline{L_2}$ such that $q_1 = 1$
    and $\SDP(Q)/\OPT(Q) = K_G$.
    \item There is a distribution $\mu$ supported on odd functions $f,g \in \overline{L_2}, f(X), g(X) \in \{-1,1\}$ such that $H_\mu(t) = \frac{t}{K_G}$.
\end{enumerate}
\end{lemma}

\begin{proof}
\medskip
\noindent(i)
Choose $q\in\mathcal Q$ attaining the infimum in
\cref{thm:duality}, so $v_\infty(q)=K_G^{-1}$.
For $Q=\sum_dq_d\Pi_d$,
\eqref{eq:asymptotic-payoff} gives $\OPT(Q)=K_G^{-1}$.
Its operator norm is at most $1$, so $\SDP(Q)\le1$.
The standard SDP construction $f(X) = g(X) = (X_1, \dots, X_n)/\sqrt{X_1^2 + \cdots + X_n^2}$ with $n \to \infty$
shows $\SDP(Q)\ge q_1=1$.
Therefore $\SDP(Q)/\OPT(Q)=K_G$.

\medskip\noindent(ii)
For every finitely supported real sequence $b$ supported
on odd degrees, Grothendieck's inequality gives
\[
 \sup_{c\in\mathcal B}\langle c,b\rangle
 =v_\infty(b)
 \ge\OPT(Q_b)
 \ge K_G^{-1}\SDP(Q_b)
\]
where $Q_b = \sum_m b_m \Pi_m$.
By the standard SDP construction, $\SDP(Q_b) \geq b_1$. Therefore, for all such $b$,
\[
 \sup_{c\in\mathcal B}\langle c,b\rangle\ge K_G^{-1} b_1.
\]
Because $\mathcal K$ is coordinatewise compact and convex,
finite-dimensional separation implies
$K_G^{-1} e_1\in\mathcal K$.
By the definition of $\mathcal K$, there is a distribution
$\mu$ on odd asymptotic Boolean pairs with
$c(\mu)=K_G^{-1}e_1$. Hence $H_\mu(t)=t/K_G$ on $[-1,1]$.
\end{proof}

\section{Overview of the rigorous certificates}
\label{sec:certificate}

\subsection{Certifying the lower bound}

To prove a lower bound on $K_G$,
we prove an upper bound on $\OPT(Q)$ for a Hermite projection game $Q = \sum_{m} q_m \Pi_m$.\footnote{
The previous best lower bound by \cite{saha2026new} can be interpreted as certifying $\OPT(\Pi_1 - 0.5\Pi_3) \leq \frac{11}{6\pi}$.}
Set $q_m=0$ except at the following degrees; every decimal is exact:
\begin{equation}\label{eq:coefficients}
\begin{array}{r|r@{\qquad}r|r}
 m&q_m&m&q_m\\\hline
 1&1&13&0.072282834573\\
 3&-0.654509751718&15&-0.004420811526\\
 5&-0.159925966134&17&-0.058106583958\\
 7&0.047785078787&19&-0.041412929910\\
 9&0.086465878916&21&0.016668224599\\
 11&0.076371662489&&
\end{array}
\end{equation}

\begin{theorem}\label{thm:lower}
${v_\infty}(q) \le \frac{1000}{1773}$. Therefore, $K_G \geq 1.773$.
\end{theorem}
\begin{proof}
    By a standard discretization argument to approximate $Q$ by finite matrices \cite[Theorems 2.4 and 2.5]{pisier2012grothendieck}, we have $K_G \geq \frac{1}{v_\infty(q)}$.
    We prove $v_\infty(q) \leq \frac{1000}{1773}$ in \eqref{eq:lower-bound}.
\end{proof}

The proof of \eqref{eq:lower-bound} occupies \cref{sec:lower}.
Proposition \ref{thm:reduction} reduces estimating $\OPT(Q)$ to a 44-dimensional optimization problem, which is then reduced to checking a finite number of inequalities numerically.
The code for the numerical portions is available at \url{https://github.com/christophsjones/RealGrothendieck}.

\subsection{Certifying the upper bound}\label{sec:upper-rounding}

To prove an upper bound on $K_G$, we estimate $c_1(\mu) - \sum_{\text{odd }m \geq 3}|c_m(\mu)|$ for a mixture of rounding functions $\mu$.

\paragraph{Polynomially correlated Gaussians.}
The rounding functions use Gaussians whose noise correlation
is an arbitrary polynomial $P(t)$.
We use graded Gaussians from \cref{sec:asymptotic} for their construction.
Given an odd polynomial $P(t) = \sum_{m \geq 1} p_m t^m$ such that $\sum_{m \geq 1}|p_m| \leq 1$,
define graded Gaussian polynomials $W(X), W'(X)$ by
\begin{align}\label{eq:w-definition}
    W(X) = W_P(X) &:= \sum_{{m}\ge1}\sqrt{|p_m|}\,X_1^{({m})}
       +\sqrt{1-\sum_{{m}\ge1}|p_m|}\,X_2,\\
    W'(X) = W'_P(X) &:=\sum_{{m}\ge1}\sign(p_m)\sqrt{|p_m|}\,X_1^{({m})}
       +\sqrt{1-\sum_{{m}\ge1}|p_m|}\,X_3.\notag
\end{align}
Recall that for $t$-correlated Gaussians $(X,Y)$,
we have $\E[X_1^{(m)}Y_1^{(m)}] = t^m.$
Then,
\begin{equation}\label{eq:w-correlation}
    \E[W(X)W'(Y)]=P(t).
\end{equation}

\paragraph{Rounding functions.}
We use a mixture $\mu$ of 32 rounding functions, with all but one involving a correlation polynomial $P$.
For $i = 0,\dots, 13$, we use a polynomial of degree at most five $P_i = p_1 t + p_3 t^3 + p_5 t^5$.
For $i = 15, \dots, 31$, we use a polynomial $P_i = p_1 t + p_m t^m$ for a single nonzero $m$.
Their coefficients are given in \cref{root:recipe}.
The rounding function for $i = 14$ does not use a correlation polynomial $P$.
\cref{root:recipe} also gives the weights $\lambda_i$ of the mixture $\mu$. These values satisfy
\begin{equation}\label{root:normalization}
 \sum_{i=0}^{31}\lambda_i=1,\qquad\lambda_i>0,\qquad
 \sum_{m\ge1}|p_m|\le1\quad(i\ne14).
\end{equation}
The rounding functions $(f_i(X),g_i(X))$ are of three types.

\begin{enumerate}[label=\arabic*.]
    \item \textbf{Graded linear threshold.}
    For $i = 17, \dots, 31$, take $f_i(X) = \sign W(X)$ and $g_i(X) = \sign W'(X)$.

    \item \textbf{Graph threshold.} For $i = 0,\dots, 13, 15, 16$, take
    \[f_i(X) = F(W, X_4, X_5, X_6), \qquad g_i(X) = F(W'(X), -X_4, -X_5, -X_6)\]
    where
    \begin{align}\label{root:threshold}
     F(w,x_4,x_5,x_6)&:=\sign\left(w+\sum_{j=4}^6r(x_j)\right),\\
    \label{eq:upper-threshold-polynomial}
     r(s):=&-0.002608895910895478\psi_1(s)-0.0999107408341197\psi_3(s)\\
     &-0.020910560460516325\psi_5(s)-0.005853637316670958\psi_7(s)\notag\\
     &+0.0057627068821593744\psi_9(s)+0.0013019893459673402\psi_{11}(s)\notag\\
     &-0.0019418422133255495\psi_{13}(s)+0.0006843699264111792\psi_{15}(s).\notag
    \end{align}
    The indices $X_4, X_5, X_6$ are used only so that they are independent from \cref{eq:w-correlation}.

    \item \textbf{Bivariate threshold.} For $i = 14$, take
    \[f_i(X) = F_*(X_1^{(3)}, X_1^{(5)})\qquad g_i(X) = F_*(-X_1^{(3)}, X_1^{(5)})\]
    where $F_*$ is a discretization of a polynomial threshold function,
    \begin{equation}\label{root:staircase-polynomial}
    \begin{aligned}
    F_*(x,y) \approx{}& \sign R(x,y),\\
    R(x,y)={}&-0.072399902677\psi_1(y)+0.073767456734\psi_3(y)\\
     &-0.71483345535\psi_1(x)+0.008648628127\psi_1(x)\psi_2(y)\\
     &+0.043678832526\psi_2(x)\psi_1(y)-0.004614640469\psi_2(x)\psi_3(y)\\
     &+0.036990218101\psi_3(x)+0.001978721807\psi_3(x)\psi_2(y)\\
     &+0.036193792294\psi_4(x)\psi_1(y)+0.131004390897\psi_5(x)\\
     &-0.011093407361\psi_7(x).
    \end{aligned}
    \end{equation}
    The discretization is defined as follows.
    For $k\in\{0,1,\ldots,2047\}$, set $y=(2k+1)/512$ and take the increasing list
    of real roots in $x$ of $R(x,y)$.
    There are five roots for $0\le k\le1052$, three for $1053\le k\le1535$,
    and one for $1536\le k<2048$. For each $k$, round these roots
    to the nearest multiples of $10^{-12}$, then apply the numerical corrections
    specified in Appendix~\ref{app:upper-bound}. Denote the resulting numbers, in increasing order, by
    $r_{k,1},r_{k,2},\ldots$. Define
    \[
     F_*(x,y):=
     \begin{cases}
      1+2\displaystyle\sum_{a:\,r_{k,a}\le x}(-1)^a,
           &k/256\le y<(k+1)/256,\quad 0\le k<2048,\\
      -\sign(x),&y\ge8,\\
      -F_*(-x,-y),&y<0.
     \end{cases}
    \]
\end{enumerate}

Let $\mu$ be the distribution of rounding functions above.
The final rounding algorithm is the Krivine scheme for $\mu$ defined in \cref{eq:krivine-scheme}. This Krivine scheme uses graded Gaussians which can be implemented in two different ways.
First, a single grade-$d$ Gaussian for an SDP vector $u$ can be approximated using a large number of independent Gaussian projections, as in \cref{eq:pi-d},
\[
    X_1^{(d)} \approx \frac{1}{\sqrt{Nd!}}\sum_{j=1}^N \He_d((Gu)_j)\,.
\]
The second way directly implements a pair of polynomially correlated Gaussians by preprocessing the SDP vectors (a trick going back to Krivine \cite{Krivine1977}).
For $P(t) = \sum_{m}p_m t^m$, the preprocessing is
\[
\begin{aligned}
 u&\longmapsto
 \bigoplus_{{m}\ge1}\sqrt{|p_m|}\,u^{\otimes {m}}
 \oplus\left(\sqrt{1-\sum_{{m}\ge1}|p_m|},0\right),\\
 v&\longmapsto
 \bigoplus_{{m}\ge1}\sign(p_m)\sqrt{|p_m|}\,v^{\otimes {m}}
 \oplus\left(0,\sqrt{1-\sum_{{m}\ge1}|p_m|}\right).
\end{aligned}
\]
Let $\tilde u, \tilde v$ be the output of this preprocessing. When $W(X), W'(X)$ appear in the Krivine scheme, we use $\ip{\tilde u, g}, \ip{\tilde v, g}$ where $g$ is an independent standard Gaussian vector. These satisfy the desired correlation identity $\E_g \ip{\tilde u, g}\ip{\tilde v, g} = P(\ip{u,v})$.
No approximation error is needed with this approach.
Polynomially correlated Gaussians via the preprocessing approach also appear as ``limiting Krivine schemes'' in the work of Saha et al.\ \cite{saha2026new}.

\paragraph{Analysis.}
Define
\[
 H_i:=H_{f_i,g_i},\qquad
 H:=H_\mu=\sum_{i=0}^{31}\lambda_iH_i,\qquad c_m:= c_m(\mu) = [t^m]H
\]
where $[t^m]H$ denotes the coefficient of the monomial $t^m$.
All supported $f_i, g_i$ are odd: all $P_i$ are odd,
hence $W(X), W'(X)$ are odd, and $F,F_*$ are also odd.
Therefore, $H_i$ and $H$ are also odd and the only nonzero $c_m$ are the odd terms.
By \cref{lem:hfg-analytic} extended to functions on graded Gaussian space by definition \eqref{eq:correlation},
the series for $H_i,H$ converge absolutely and uniformly on
$|t|\le1$ and are holomorphic on $|t|<1$.

\begin{theorem}\label{thm:upper}
For the rounding scheme above,
\[
 c_1-\sum_{\substack{m\ge3\\m\ \mathrm{odd}}}|c_m|
 >0.561831528950464991>\frac{1000000}{1779893}.
\]
In particular, $K_G<1.779893$.
\end{theorem}
\begin{proof}\label{root:assembly}
Combining the bounds \cref{eq:upper-certified-head},
\cref{eq:upper-certified-nonadditive}, \cref{eq:upper-certified-rare},
\cref{heavy:final-number} gives,
\begin{align*}
 c_1-\sum_{m\ge2}|c_m|
 &\ge c_1-\sum_{m=2}^{681}|c_m|
       -\lambda_{14}\sum_{m>681}|[t^m]H_{14}|
       -\sum_{i=17}^{31}\lambda_i\sum_{m>681}|[t^m]H_i|\\
 &\quad-\sum_{i=1}^{13}\lambda_i\sum_{m>681}|[t^m]H_i|
       -\sum_{m>681}\left|[t^m]\sum_{i\in\{0,15,16\}}\lambda_iH_i\right|\\
 &>0.5618773799453352-0.000003878150168295\\
 &\quad-0.000034097111321496-0.000007875733380418\\
 &=0.561831528950464991>\frac{1000000}{1779893}.
\end{align*}
By \cref{lem:asymptotic-c} and the $\ell_1$ continuity of $c_1 - \sum_{\text{odd }m \ge 3}|c_m|$,
\cref{lem:ub-easy} continues to apply if $\mu$ is a distribution on pairs of graded odd Boolean functions.
Then, $K_G \leq \left(c_1 - \sum_{m \geq 3} |c_m|\right)^{-1} \leq 1.779893$.
\end{proof}

The coefficient bounds are proven in \cref{app:upper-bound}.
Coefficients up to degree 681 are estimated numerically,
using integral representations of the coefficients then bounds on the quadrature error.
Analytical methods are used to bound the infinite tails of the coefficient series,
reducing to additional numerical checks for the different function families.

\paragraph{Acknowledgments.} GPT-6 Astra was used as described in the text for the numerical search and certification procedures.
The numerical search is described in \cref{sec:numerics}.
The certificates, in \cref{sec:lower,app:upper-bound}, were verified
for correctness by the authors,
and they have been substantially edited to improve their exposition.
Additionally, GPT-6 Astra assisted with the proofs in \cref{sec:duality,sec:asymptotic}.

S.H. is supported by NSF Grant CCF AF 2448108.  G.M.\ is supported by the European Research Council through an ERC Starting Grant (Grant agreement No.~101077455, ObfusQation) and funded by the Deutsche Forschungsgemeinschaft (DFG, German Research Foundation) under Germany's Excellence Strategy - EXC 2092 CASA - 390781972.

\noindent\begin{minipage}{\linewidth}
\begingroup\scriptsize
\setlength{\arraycolsep}{2pt}
\[
\resizebox{\linewidth}{!}{$%
\begin{array}{r|r|rrrrr}
i&{\lambda_i}&t&t^3&t^5&{t^{13}}&{t^{15}}\\\hline
0&{\frac{60643327718562175000}{25000000000000000247431}}&0.946&-0.054&0&{0}&{0}\\[4pt]
1&{\frac{175365095292992000}{25000000000000000247431}}&0&-1&0&{0}&{0}\\[4pt]
2&{\frac{4573126272489115}{1219512195121951231582}}&0.05&-0.95&0&{0}&{0}\\[4pt]
3&{\frac{248005501245969575}{25000000000000000247431}}&{\frac{7500000000000001}{100000000000000001}}&{-\frac{92500000000000000}{100000000000000001}}&0&{0}&{0}\\[4pt]
4&{\frac{1746960750902429000}{25000000000000000247431}}&0.1&-0.9&0&{0}&{0}\\[4pt]
5&{\frac{3119884864162846250}{25000000000000000247431}}&0.2&-0.8&0&{0}&{0}\\[4pt]
6&{\frac{1167759537924566000}{25000000000000000247431}}&0.225&-0.775&0&{0}&{0}\\[4pt]
7&{\frac{10622692935124497500}{25000000000000000247431}}&{\frac{35000000000000003}{100000000000000000}}&{-\frac{6499999999999999}{10000000000000000}}&0&{0}&{0}\\[4pt]
8&{\frac{4510989226396620750}{25000000000000000247431}}&{\frac{2006830122278547}{40000000000000000}}&{-\frac{7540766094444311}{10000000000000000}}&{\frac{4893815937465127}{25000000000000000}}&{0}&{0}\\[4pt]
9&{\frac{225271185467334750}{25000000000000000247431}}&{\frac{1741176476977869}{100000000000000013}}&{-\frac{92295288728215510}{100000000000000013}}&{\frac{5963534794806634}{100000000000000013}}&{0}&{0}\\[4pt]
10&{\frac{2590644248486109750}{25000000000000000247431}}&{\frac{3393956633363427}{100000000000000000}}&{-\frac{655490081870693}{1000000000000000}}&{\frac{3105703517956727}{10000000000000000}}&{0}&{0}\\[4pt]
11&{\frac{159466984575209675}{25000000000000000247431}}&{-\frac{36112846919679523}{10000000000000000103}}&{-\frac{9881992863232164000}{10000000000000000103}}&{-\frac{81894289848156580}{10000000000000000103}}&{0}&{0}\\[4pt]
12&{\frac{1252554238454943650}{25000000000000000247431}}&{\frac{4901570661097}{5000000000000000}}&{-\frac{5234190673248801}{10000000000000000}}&{\frac{951201237085801}{2000000000000000}}&{0}&{0}\\[4pt]
13&{\frac{349449731405024675}{25000000000000000247431}}&{\frac{280336034129396}{312500000000000021}}&{-\frac{140417768196820875}{312500000000000021}}&{\frac{171801895769049750}{312500000000000021}}&{0}&{0}\\[4pt]
14&{\frac{892199070421725000}{25000000000000000247431}}&\text{--}&\text{--}&\text{--}&{\text{--}}&{\text{--}}\\[4pt]
15&{\frac{18675444698504427500000}{25000000000000000247431}}&{\frac{45146000000000005}{50000000000000002}}&{-\frac{4853999999999997}{50000000000000002}}&0&{0}&{0}\\[4pt]
16&{\frac{5572553297529536750000}{25000000000000000247431}}&{\frac{22573500000000000}{25000000000000001}}&{-\frac{2426500000000001}{25000000000000001}}&0&{0}&{0}\\[4pt]
17&{\frac{661943041586152825000}{25000000000000000247431}}&{\frac{8859999999999999}{10000000000000000}}&{-\frac{1140000000000001}{10000000000000000}}&{0}&{0}&{0}\\[4pt]
18&{\frac{166403349365097000}{25000000000000000247431}}&{\frac{240000000000000000}{250000000000000009}}&{-\frac{10000000000000009}{250000000000000009}}&{0}&{0}&{0}\\[4pt]
19&{\frac{24994962065681940}{25000000000000000247431}}&0.1&{0}&{0}&-0.9&{0}\\[4pt]
20&{\frac{9752885915637175}{25000000000000000247431}}&0.525&{0}&{0}&-0.475&{0}\\[4pt]
21&{\frac{30936434154596385}{25000000000000000247431}}&{\frac{5499999999999999}{10000000000000000}}&{0}&{0}&{-\frac{45000000000000007}{100000000000000000}}&{0}\\[4pt]
22&{\frac{27931141611749985}{25000000000000000247431}}&{\frac{17500000000000000}{25000000000000001}}&{0}&{0}&{-\frac{7500000000000001}{25000000000000001}}&{0}\\[4pt]
23&{\frac{23042926715956245}{50000000000000000494862}}&0.725&{0}&{0}&-0.275&{0}\\[4pt]
24&{\frac{120200964730604775}{25000000000000000247431}}&{\frac{7999999999999999}{10000000000000000}}&{0}&{0}&{-\frac{20000000000000007}{100000000000000000}}&{0}\\[4pt]
25&{\frac{1360595407560664500}{25000000000000000247431}}&{\frac{9249999999999999}{10000000000000000}}&{0}&{0}&{-\frac{7500000000000007}{100000000000000000}}&{0}\\[4pt]
26&{\frac{213458751479788850}{25000000000000000247431}}&{\frac{237500000000000000}{250000000000000011}}&{0}&{0}&{-\frac{12500000000000011}{250000000000000011}}&{0}\\[4pt]
27&{\frac{8400770105211650}{25000000000000000247431}}&0.1&{0}&{0}&{0}&0.9\\[4pt]
28&{\frac{15098122612496557}{50000000000000000494862}}&0.375&{0}&{0}&{0}&0.625\\[4pt]
29&{\frac{63773334657951175}{50000000000000000494862}}&0.4&{0}&{0}&{0}&0.6\\[4pt]
30&{\frac{149703036787065950}{25000000000000000247431}}&{\frac{6749999999999999}{10000000000000000}}&{0}&{0}&{0}&{\frac{32500000000000007}{100000000000000000}}\\[4pt]
31&{\frac{97307307105401800}{25000000000000000247431}}&{\frac{17500000000000000}{25000000000000001}}&{0}&{0}&{0}&{\frac{7500000000000001}{25000000000000001}}\\[4pt]
\end{array}%
$}
\]
\endgroup
{\refstepcounter{table}\label{root:recipe}%
Table~\thetable: {Weights and polynomial coefficients used in the rounding scheme.} Over 96\% of the probability mass is on $i=15$ and $i=16$.}
\end{minipage}

{\small
\bibliographystyle{alpha}
\let\KGoriginalbibitem\bibitem
\def\KGredreference{assaad2026weak}
\renewcommand{\bibitem}[2][]{%
    \par\def\KGcurrentreference{#2}%
    \ifx\KGcurrentreference\KGredreference\color{red}\else\normalcolor\fi
    \if\relax\detokenize{#1}\relax
        \KGoriginalbibitem{#2}%
    \else
        \KGoriginalbibitem[#1]{#2}%
    \fi}
\bibliography{references}

@article{AcinGisinToner2006,
  author  = {Ac{\'i}n, Antonio and Gisin, Nicolas and Toner, Benjamin},
  title   = {{Grothendieck}'s constant and local models for noisy
             entangled quantum states},
  journal = {Physical Review A},
  volume  = {73},
  number  = {6},
  pages   = {062105},
  year    = {2006},
  doi     = {10.1103/PhysRevA.73.062105}
}

@article{AlonNaor2006,
  author  = {Alon, Noga and Naor, Assaf},
  title   = {Approximating the cut-norm via {Grothendieck}'s inequality},
  journal = {SIAM Journal on Computing},
  volume  = {35},
  number  = {4},
  pages   = {787--803},
  year    = {2006},
  doi     = {10.1137/S0097539704441629}
}

@article{KhotNaor2012,
  author  = {Khot, Subhash and Naor, Assaf},
  title   = {{Grothendieck}-type inequalities in combinatorial
             optimization},
  journal = {Communications on Pure and Applied Mathematics},
  volume  = {65},
  number  = {7},
  pages   = {992--1035},
  year    = {2012},
  doi     = {10.1002/cpa.21398}
}

@article{saha2026new,
  title={New Lower and Upper Bounds for the {G}rothendieck Constant},
  author={Saha, Rahul and Li, Alan and Xue, Anton and Chaudhuri, Swarat and Klivans, Adam and Kothari, Pravesh K and Meka, Raghu},
  journal={arXiv preprint arXiv:2608.11158},
  year={2026}
}

@article{jones2026regularity,
  title={Regularity lemma for {G}aussian polynomials},
  author={Jones, Chris},
  journal={To appear},
  year={2026}
}

@inproceedings{raghavendra2009towards,
  title={Towards computing the {G}rothendieck constant},
  author={Raghavendra, Prasad and Steurer, David},
  booktitle={Proceedings of the Twentieth Annual ACM-SIAM Symposium on Discrete Algorithms},
  pages={525--534},
  year={2009},
  organization={SIAM}
}

@article{naor2014krivine,
  title={Krivine schemes are optimal},
  author={Naor, Assaf and Regev, Oded},
  journal={Proceedings of the American Mathematical Society},
  volume={142},
  number={12},
  pages={4315--4320},
  year={2014}
}

@article{herry2024limit,
  title={Limit distributions for polynomials with independent and identically distributed entries},
  author={Herry, Ronan and Malicet, Dominique and Poly, Guillaume},
  journal={arXiv preprint arXiv:2412.06749},
  year={2024}
}

@article{HaleTrefethen2008,
  author={Hale, Nicholas and Trefethen, Lloyd N.},
  title={New Quadrature Formulas from Conformal Maps},
  journal={SIAM Journal on Numerical Analysis},
  volume={46},
  number={2},
  pages={930--948},
  year={2008},
  doi={10.1137/07068607X}
}

@book{grothendieck1956resume,
  title={R{\'e}sum{\'e} de la th{\'e}orie m{\'e}trique des produits tensoriels topologiques},
  author={Grothendieck, Alexandre},
  volume={2},
  year={1956},
  publisher={Soc. de Matem{\'a}tica de S{\~a}o Paulo}
}

@article{pisier2012grothendieck,
  title={{Grothendieck}'s theorem, past and present},
  author={Pisier, Gilles},
  journal={Bulletin of the American Mathematical Society},
  volume={49},
  number={2},
  pages={237--323},
  year={2012}
}

@inproceedings{braverman2013grothendieck,
  title={The {G}rothendieck constant is strictly smaller than {K}rivine's bound},
  author={Braverman, Mark and Makarychev, Konstantin and Makarychev, Yury and Naor, Assaf},
  booktitle={Forum of Mathematics, Pi},
  volume={1},
  pages={e4},
  year={2013},
  organization={Cambridge University Press}
}

@article{goemans1995improved,
  title={Improved approximation algorithms for maximum cut and satisfiability problems using semidefinite programming},
  author={Goemans, Michel X and Williamson, David P},
  journal={Journal of the ACM (JACM)},
  volume={42},
  number={6},
  pages={1115--1145},
  year={1995},
  publisher={ACM New York, NY, USA}
}

@article{feige2002optimality,
  title={On the optimality of the random hyperplane rounding technique for {MAX} {CUT}},
  author={Feige, Uriel and Schechtman, Gideon},
  journal={Random Structures \& Algorithms},
  volume={20},
  number={3},
  pages={403--440},
  year={2002},
  publisher={Wiley Online Library}
}

@article{jones2026grothendieck,
  title={The {G}rothendieck Constant is Strictly Larger than {D}avie-{R}eeds' Bound},
  author={Jones, Chris and Malavolta, Giulio},
  journal={arXiv preprint arXiv:2603.30039},
  year={2026}
}

@article{heilman2026lower,
  title={A Lower Bound for {G}rothendieck's Constant},
  author={Heilman, Steven},
  journal={arXiv preprint arXiv:2603.22616},
  year={2026}
}

@article{heilman2026upper,
  title={An upper bound on {G}rothendieck's constant},
  author={Heilman, Steven},
  journal={arXiv preprint arXiv:2606.00247},
  year={2026}
}

@article{Krivine1977,
  author  = {Krivine, Jean-Louis},
  title   = {Sur la constante de {G}rothendieck},
  journal = {Comptes Rendus Hebdomadaires des S{\'e}ances de l'Acad{\'e}mie des Sciences. S{\'e}rie A et B},
  volume  = {284},
  number  = {8},
  pages   = {A445--A446},
  year    = {1977}
}

@article{parisi1979infinite,
  title={Infinite number of order parameters for spin-glasses},
  author={Parisi, Giorgio},
  journal={Physical Review Letters},
  volume={43},
  number={23},
  pages={1754},
  year={1979},
  publisher={APS}
}

@article{mourrat2025inverting,
  title={Un-inverting the {Parisi} formula},
  author={Mourrat, Jean-Christophe},
  journal={Annales de l'Institut Henri Poincare (B) Probabilites et statistiques},
  volume={61},
  number={4},
  pages={2709--2720},
  year={2025},
  organization={Institut Henri Poincar{\'e}}
}

@article{assaad2026weak,
  title={Weak limits of multiple Wiener integrals: representation, extraction, and realization},
  author={Assaad, Obayda Julien},
  journal={arXiv preprint arXiv:2608.12492v2},
  year={2026},
  url={https://arxiv.org/abs/2608.12492v2}
}
}

\appendix

\clearpage
\section{Numerical summary of the primal/dual run}
\label{sec:numerics}

We report the statistics used for one representative execution of the primal/dual search.
We have attempted to use more reproducible methods for this run instead of pushing for the absolute tightest bounds on $K_G$.

\paragraph{Runtime statistics.}
The maximum degree parameter was $D=15$.
The initialization was $Q = \Pi_1$.
There were 92 update iterations for $Q$. The final response pool had $|\calR| = 274$ rounding function pairs $(f,g)$. A pair was added to $\calR$ if they significantly exceeded the current maximum value of $Q$ on $\calR$.
A larger pool of 13{,}069 asymptotic Boolean functions (see \cref{sec:asymptotic}) was maintained to heuristically estimate the $\infty$-to-1 norm of $Q$.
The search took 190 minutes using GPT-6 Astra.
After 60 minutes, the algorithm began iteration 82. After 115 minutes, the algorithm began iteration 90. The remaining 75 minutes were spent searching for $\infty$-to-1 optimizers and updating $Q$ only two more times.

The final candidate $\mu$ was computed after the iteration using \cref{eq:objective-mu}.
The LP was equivalent to maximizing $c_1(\mu)$ subject to $c_3 = c_5 = \cdots = c_{15} = 0$ (a useful variant would be to increase $D$ in this calculation in order to minimize more low-degree coefficients of $\mu$, but in this run we kept $D = 15$).
Furthermore, in order to measure the progress of upper bounds throughout the iteration, we additionally ran \cref{eq:objective-mu} at each iteration. The change in the heuristic upper and lower bounds is plotted in \cref{fig:optq-growth}.

\begin{figure}[ht]
    \centering
    \includegraphics[width=\linewidth]{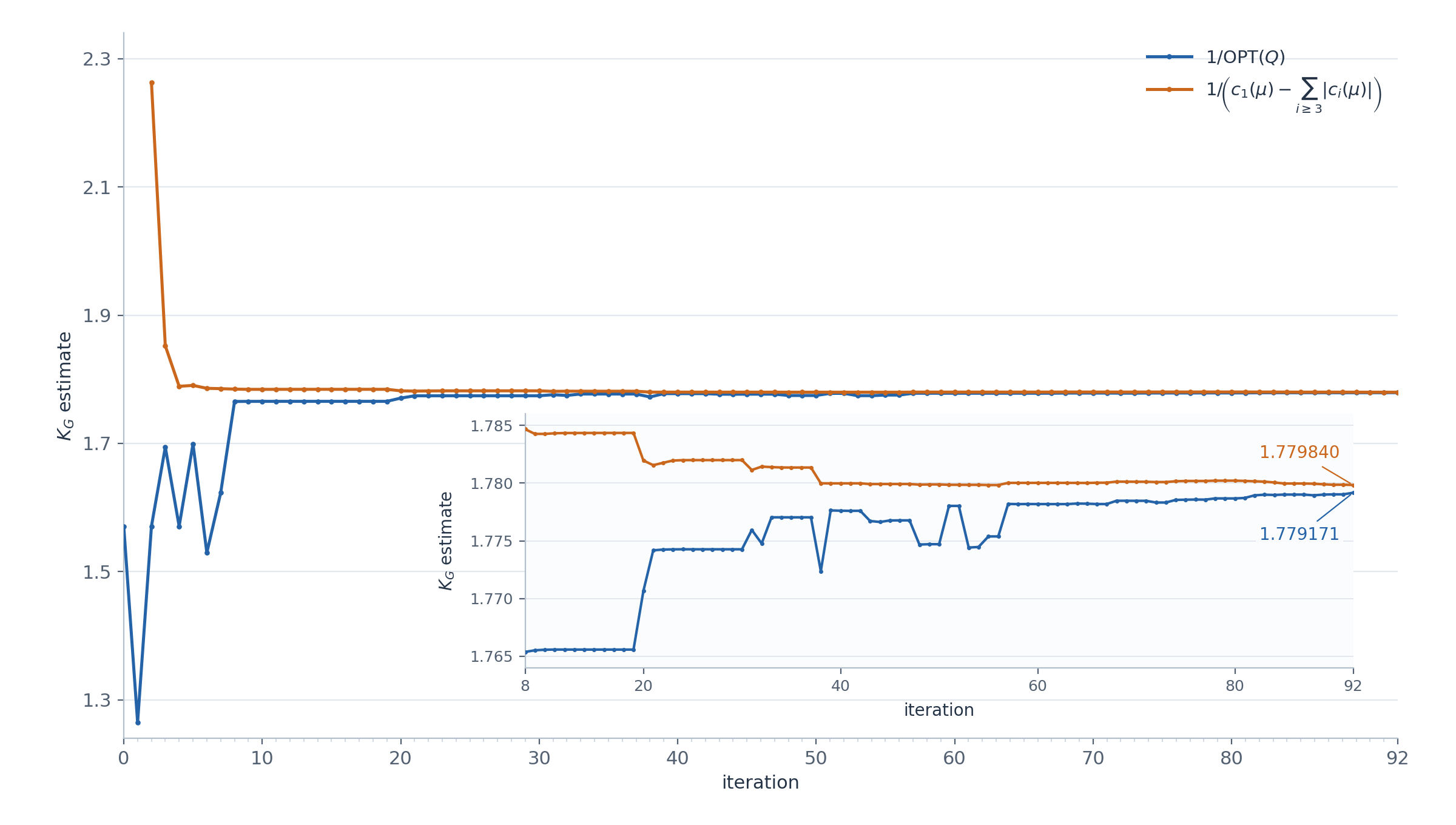}
    \caption{Primal/dual estimates of $K_G$ by iteration. The value $\OPT(Q^{(t)})$ is estimated using the formula $\OPT(Q^{(t)}) \approx \max_{(f,g) \in \calR} \langle f, Q^{(t)} g\rangle$ where $\calR$ is the final pool of pairs of rounding functions. The estimate of $\sum_{i \geq 3}|c_i(\mu)|$ is numerically estimated by truncating to degree at least 63. At the end of the iteration, sources contributing to the remaining gap include a nonzero duality gap, only optimizing cancellations in $\mu$ up to degree 15, and numerical integration error.}
    \label{fig:optq-growth}
\end{figure}

To estimate $\OPT(Q)$ during the search, alternating maximization $f \gets \sign(Qg)$ and $g \gets \sign(Q f)$ is used to obtain a local optimum starting from an initial candidate $f$ or $g$.
Candidate initializations were stored in the larger pool of 13{,}069 functions, and new initializations were enumerated from various function families, such as the families described in \cref{sec:upper-rounding}:
graded linear thresholds, polynomial threshold functions in a small number of variables, and graph thresholds.
An additional class of symmetric polynomial threshold functions was used in the final iterations,
resulting from local search initialized at a graph threshold function.

\paragraph{Outputs.}
The final Hermite projection game has degree $15$ with
$Q=\sum_{\substack{\text{odd }1 \leq m \leq 15}}q_m\Pi_m$
where the coefficients $q_m$ are in Table~\ref{tab:q15-coefficients}.
\begin{table}[htbp]
\centering\small
\begin{tabular}{r r@{\qquad}r r}
$m$ & $q_m$ & $m$ & $q_m$ \\ \hline
1 & $ 1$                  &  9 & $ 0.13649017557899745$  \\
3 & $-0.879274051265278$   & 11 & $-0.049432632669210025$ \\
5 & $ 0.5207487584184273$  & 13 & $-0.10368806013060777$  \\
7 & $-0.3110788699180306$  & 15 & $ 0.12636021488829785$  \\
\end{tabular}
\caption{Coefficients for the candidate degree-\(15\) Hermite game produced by the run.}
\label{tab:q15-coefficients}
\end{table}

During the final iteration, when $Q$ was selected to minimize the value of the pool via \cref{eq:objective-q}, the maximum value among the pool functions was $\max_{(f,g) \in \calR} \langle f, Qg \rangle = 0.5619674$.
A subsequent search for new responses to $Q$ found functions with slightly larger value,
\[
  \max_{(f,g)\in\mathcal R}\langle f,Qg\rangle
  =0.5620595,
  \qquad \frac1{0.5620595}=1.779171\ldots.
\]
Thus, the duality gap at termination was still nonzero.
The $\infty$-to-1 calculations were getting long enough that it was very expensive to continue the iteration.

The final mixture of rounding functions $\mu$ is supported on 8 pairs of rounding functions $(f,g)$ each with the form $\sign(p(X))$ for a graded multivariate polynomial $p$ with graded-degree at most 15.
The mixture weights are given in \cref{tab:rounding-mixture-weights}.
All 8 rounding pairs have value approximately 0.5619674 in the final game $Q$,
the maximum among functions in the pool (as expected by complementary slackness).
The numerical estimates are $c_1(\mu)\approx0.5619674$ and
$\sum_{\text{odd }i = 3}^{63}|c_i(\mu)|\approx0.0001118$.
The resulting estimated upper bound is
\[
    c_1(\mu) - \sum_{i=3}^{63}|c_i(\mu)| \approx 0.5618556, \qquad \frac{1}{0.5618556} = 1.7798167\ldots.
\]

\begin{table}[htbp]
\centering\small
\begin{tabular}{r r l}
$i$ & $w_i$ & Function family \\ \hline
1  & $0.07659708908805163$  & Graded linear threshold \\
2 & $0.38070941881309583$  & Graded linear threshold \\
3 & $0.3422214571899304$   & Graph threshold, ${n}=8$ \\
4 & $0.027934013874205684$ & Bivariate threshold, ${n}=2$ \\
5 & $0.04989420667246783$  & Bivariate threshold, ${n}=2$ \\
6 & $0.05838336532409958$  & Symmetric polynomial threshold, ${n}=32$ \\
7 & $0.017645044351596674$ & Symmetric polynomial threshold, ${n}=32$ \\
8 & $0.04661540468655233$  & Symmetric polynomial threshold, ${n}=8$ \\
\end{tabular}
\caption{Weights of the final rounding mixture. ${n}$ is the number of graded variables used in the function. The full definitions of the rounding functions are available at \url{https://github.com/christophsjones/RealGrothendieck}.}
\label{tab:rounding-mixture-weights}
\end{table}

\section{Lower bound certificate}
\label{sec:lower}

Define the operator $Q$ on $L_2(\gamma_n)$ by
\begin{equation}\label{eq:game}
 {Q=Q^{({n})}}=\sum_mq_m\Pi_m,
 \qquad
 v_{n}(q)={\OPT(Q)}=\sup_{f,g}|\langle f,{Q}g\rangle|
\end{equation}
where the sup is over measurable $f,g\colon\R^{n}\to[-1,1]$.
Define coefficients $q_m$ from \cref{eq:coefficients}.

\begin{lemma}
    \begin{equation}\label{eq:lower-bound}
        v_\infty(q) \leq \frac{1000}{1773}
    \end{equation}
\end{lemma}

\subsection{Reduction to one variable}\label{sec:reduction}

For $N=7,9$, choose rational weights $w^{(N)}\in\R^{(N-1)/2}$ and put
\begin{align}
 B_N&:=\max_{\sigma\in\{-1,1\}^N}
 \frac1N\sum_{r=0}^{N-1}\sum_{j=1}^{(N-1)/2}
 w_j^{(N)}\sigma_r\sigma_{r+j},\qquad
 \sigma_{r+N}=-\sigma_r,\label{eq:rotation-constant}\\
 \Gamma_m&:=\sum_{N\in\{7,9\}}\sum_{j=1}^{(N-1)/2}
 w_j^{(N)}\cos^m(\pi j/N),\qquad B=B_7+B_9.
 \label{eq:rotation-coefficients}
\end{align}
Global sign reversal leaves the objective unchanged, so one may fix $\sigma_0=1$.
Define, with all other entries zero,
\begin{equation}\label{eq:residual-weights}
 \alpha_m:=q_m-\Gamma_m\quad(m=3,5,\ldots,21),\qquad
 \beta_m:=-q_m-\Gamma_m\quad(m=1,3,\ldots,21).
\end{equation}

Set $\psi:=(\psi_0,\ldots,\psi_{21})^T$ and call
$M:=\binom{\E[\psi u]}{\E[\psi v]}\in\R^{44}$ \textbf{feasible} when
$u,v\colon\R\to[-1,1]$ are measurable. Index its coordinates by $0,\ldots,43$
and set, for $p=0,1$ and $0\le j\le10$,
\begin{equation}\label{eq:hk-moments}
 (h_p(M))_j:=\frac{M_{2j+p}+M_{22+2j+p}}2,\qquad
 (k_p(M))_j:=\frac{M_{2j+p}-M_{22+2j+p}}2.
\end{equation}
Choose symmetric rational $11\times11$ matrices $H_p,K_p$, vectors
$\ell_h,\ell_k\in\mathbb Q^{11}$, and rational $t,c_h,c_k,C$.
Set $\varepsilon:=10^{-7}$ and
$D_{p,r}(a):=\diag(a_{2j+p+r})_{j=0}^{10}$. Require
\begin{equation}\label{eq:rotation-tail}
 \Gamma_m\ge0\quad(m>21\text{ odd}),\qquad 1-\Gamma_1>0.
\end{equation}
For each $(A,a,\lambda,c)\in\{(H,\alpha,\ell_h,c_h),
(K,\beta,\ell_k,c_k)\}$, assume, $\forall$ $p=0,1$ and $1\le r\le21$,
\begin{equation}\label{eq:matrix-comparisons}
 A_p\succeq0,\qquad
 A_p+\varepsilon I_{11}\succeq D_{p,r}(a),\qquad
 \begin{pmatrix}
 A_1+\varepsilon I_{11}-D_{1,0}(a)&\lambda\\
 \lambda^T&c+\varepsilon
 \end{pmatrix}\succeq0.
\end{equation}
Here $\succeq0$ means positive semidefinite, and $I_{11}$ is the $11\times11$ identity matrix.
Unsubscripted norms of finite vectors are Euclidean; $\|z\|_1=\sum_j|z_j|$.
Define
\begin{align}
F(M)&:=\sum_{p=0}^1\bigl(h_p(M)^TH_ph_p(M)+k_p(M)^TK_pk_p(M)\bigr)
 +t(h_1(M))_0-2\ell_h^Th_1(M)-2\ell_k^Tk_1(M),
 \label{eq:moment-quadratic}\\
 \Phi(x)&:=B+c_h+c_k+3\varepsilon
 +(1-\Gamma_1)\frac2\pi x^2-t\sqrt{2/\pi}\,x.
 \label{eq:scalar-correction}
\end{align}
Write $F(M)=M^TRM+\ell^TM$ with $R$ symmetric. By the invertible
change of variables \cref{eq:hk-moments}, $R\succeq0$ is equivalent to
$H_p,K_p\succeq0$ for $p=0,1$.

\begin{proposition}[One-variable criterion]\label{thm:reduction}
Suppose a finite family of closed intervals $I$ covers $[0,1]$,
each with dimension-independent parameters satisfying
\cref{eq:rotation-tail}, \cref{eq:matrix-comparisons}, and
\begin{align}
 F(M)&\le C\quad\text{for every feasible }M,\label{eq:one-dimensional}\\
 C+\max_{x\in\partial I}\Phi(x)&\le {v_*}
 =\frac{1000}{1773}-10^{-9}.
 \label{eq:endpoint}
\end{align}
Here $\partial I$ denotes the interval endpoints.
Then $v_{n}(q)\le {v_*}$ for every ${n}\ge1$, hence Theorem \ref{thm:lower}. 
\end{proposition}
\begin{proof}
By \cref{lem:odd-rounding}, to bound $v_n(q)$ it suffices to bound \emph{odd} optimizers $f,g$.
Put $h:=(f+g)/2$, $k:=(f-g)/2$. We have $h,k \in \{-1,0,1\}$ and the relation $h^2+k^2=1$.  

Rotate the domain so that $\Pi_1h(z)$ is a nonnegative multiple of $z_1$.
Then $\|\Pi_1h\|_2=\E[z_1h]\le\E|z_1|=\sqrt{2/\pi}$, so
$x:=\sqrt{\pi/2}\|\Pi_1h\|_2\in[0,1]$ and $\Pi_1h=\sqrt{2/\pi}xz_1$.
Choose an interval $I\ni x$ and use its associated parameters.

For each $N$ and independent $U,V\sim\gamma_{n}$, set
$\sigma_r:=f(\cos(r\pi/N)U+\sin(r\pi/N)V)$.
Oddness gives $\sigma_{r+N}=-\sigma_r$.
The correlation identity \cref{eq:correlation} gives
\begin{equation}\label{eq:bnbound}
 \sum_jw_j^{(N)}\sum_{m\text{ odd}}\cos^m(\pi j/N)\|\Pi_m f\|_2^2
 =\E\!\left[\frac1N\sum_{r=0}^{N-1}\sum_j
w_j^{(N)}\sigma_r\sigma_{r+j}\right]
\stackrel{\cref{eq:rotation-constant}}{\le}\E[B_N]
=B_N.
\end{equation}
All series converge absolutely by Parseval. Average for $f,g$ and sum
over $N$, using $(\|\Pi_m(f)\|_2^2+\|\Pi_m g\|_2^2)/2=\|\Pi_m h\|_2^2+\|\Pi_m k\|_2^2$.
Combining this with
$\langle f,{Q}g\rangle=\sum_mq_m(\|\Pi_m h\|_2^2-\|\Pi_m k\|_2^2)$
and $q_1=1$, $q_m=0$ for $m>21$ from \cref{eq:coefficients}, gives
\begin{equation}\label{eq:payoff-residual}
\begin{aligned}
 \langle f,{Q}g\rangle
 &=\sum_{m\,\rm odd}\Gamma_m(\|\Pi_m h\|_2^2+\|\Pi_m k\|_2^2)+\sum_{m\,\rm odd}\bigl((q_m-\Gamma_m)\|\Pi_m h\|_2^2+(-q_m-\Gamma_m)\|\Pi_m k\|_2^2\bigr)\\
&\stackrel{\cref{eq:bnbound}\wedge\cref{eq:rotation-coefficients}\wedge\cref{eq:residual-weights}\wedge\cref{eq:rotation-tail}}{\le} B+(1-\Gamma_1)\frac2\pi x^2
 +\sum_m\bigl(\alpha_m\|\Pi_m h\|_2^2+\beta_m\|\Pi_m k\|_2^2\bigr).
 \end{aligned}
\end{equation}

Write the Gaussian coordinates as independent
$S\sim\gamma_1$, $Y\sim\gamma_{{n}-1}$, with $Y$ the empty tuple if ${n}=1$.
Let $M(Y)$ be the feasible moments of $f(\cdot,Y),g(\cdot,Y)$. For
$p=0,1$ and $\eta\in\{0,1,\ldots\}^{{n}-1}$, set
\[
\begin{aligned}
 u_p(Y)&:=h_p(M(Y)),&v_p(Y)&:=k_p(M(Y)),\\
 \widehat u_p(\eta)&:=\E[u_p(Y)\psi_\eta(Y)],&
 \widehat v_p(\eta)&:=\E[v_p(Y)\psi_\eta(Y)].
\end{aligned}
\]
The last sum in \cref{eq:payoff-residual} equals
\[
\sum_{\eta}\sum_{p=0}^1[\widehat u_p(\eta)^TD_{p,|\eta|}(\alpha)\widehat u_p(\eta)
+\widehat v_p(\eta)^TD_{p,|\eta|}(\beta)\widehat v_p(\eta)].
\]
Also $(h_1(M(Y)))_0=(M_1(Y)+M_{23}(Y))/2=\E_S[Sh(S,Y)]$, so
$\E_Y(h_1(M(Y)))_0=\sqrt{2/\pi}\,x$.
In the Hermite expansion in $Y$, coordinate $j$ of $h_p(M),k_p(M)$ at
index $\eta$ has total degree $2j+p+|\eta|$ in $(S,Y)$.
For $1\le|\eta|\le21$, apply
$D_{p,|\eta|}(\alpha)\preceq H_p+\varepsilon I_{11}$ and
$D_{p,|\eta|}(\beta)\preceq K_p+\varepsilon I_{11}$ from \cref{eq:matrix-comparisons}.
For $|\eta|>21$, both diagonal matrices vanish and $H_p,K_p\succeq0$ by the same equation.
When $\eta=0$ and $p=0$, $D_{0,0}(\alpha)=D_{0,0}(\beta)=0$.
For $\eta=0$ and $p=1$, the last inequality of \cref{eq:matrix-comparisons},
applied to $(z,-1)$, gives the following bounds for
$z=\widehat u_1(0)$ and $z=\widehat v_1(0)$, respectively:
\[
z^TD_{1,0}(\alpha)z
\leq z^TH_1z
-2\ell_h^Tz+c_h+\varepsilon+\varepsilon\|z\|_2^2
,\quad
z^TD_{1,0}(\beta)z
\leq z^TK_1z
-2\ell_k^Tz+c_k+\varepsilon+\varepsilon\|z\|_2^2.
\]
Parseval adds these bounds. The constants cost $c_h+c_k+2\varepsilon$,
and the quadratic errors cost at most
\[
 \varepsilon\E_Y\sum_p(\|h_p(M(Y))\|^2+\|k_p(M(Y))\|^2)
 \le\varepsilon\E(h^2+k^2)=\varepsilon
\]
by Bessel's inequality. 
Combining these estimates into \cref{eq:payoff-residual} then gives
\begin{equation}\label{eq:slice-bound}
 \langle f,{Q}g\rangle\le\E_YF(M(Y))+\Phi(x)
 \stackrel{\cref{eq:one-dimensional}}{\le} C+\Phi(x).
\end{equation}
By \cref{eq:rotation-tail}, $\Phi$ is convex. Since $x\in I$,
$C+\Phi(x)\le C+\max_{\partial I}\Phi\le {v_*}$ by \cref{eq:endpoint}, proving the claim.
\end{proof}

\subsection{Finite certification on intervals}\label{sec:computation}
Since $R\succeq0$, we reduce $F\le C$ on feasible moments to an
integer quadratic in signs.

\paragraph{Reduction to finitely many signs.}
Partition $[-9,9]$ into $J=[j/2,(j+1)/2]$, $-18\le j\le17$,
and divide each $J$ into $n_J$ panels of length $h_J=1/(2n_J)$.
On a panel $[a,b]$, put $\tau=(s-a)/(b-a)$ and
\[
 (\lambda_0,\lambda_1,\lambda_2):=((1-\tau)^2,2\tau(1-\tau),\tau^2),
 \quad(c_0(p),c_1(p),c_2(p)):=
 \left(p(a),2p\left(\frac{a+b}{2}\right)-\frac{p(a)+p(b)}2,p(b)\right).
\]
The interpolant is $I_2p=\sum_\nu\lambda_\nu c_\nu(p)$, with
$\lambda_\nu\ge0$. Form $V$ from the columns
$c_\nu(\psi)\int_a^b\lambda_\nu\,d\gamma_1$, merging those at shared
endpoints (positive multiples of the same $\psi(s)$) and ordering by
associated points. Set $\mathcal B=\operatorname{diag}(V,V)$; it has
$44$ rows and ${N_{\mathrm{sign}}}=4\sum_Jn_J+2$ columns.

For feasible $M$, set ${\xi=2RM+\ell}$, $p_i=\psi^T\xi^{[i]}$,
and $\sigma_j=\sign((\mathcal B^T\xi)_j)$; brackets indicate
$22$-coordinate blocks, or $22\times22$ matrix blocks. Convexity, $\lambda_\nu\ge0$, and the
error bound \cref{eq:analytic-error} below give
\begin{equation}\label{eq:finite-signs}
 F(M)-F(\mathcal B\sigma)
 \le\xi^TM-\|\mathcal B^T\xi\|_1
 \le\sum_i\int|p_i|\,d\gamma_1-\|\mathcal B^T\xi\|_1
 \le E.
\end{equation}

\paragraph{The analytic error.}
Quadratic interpolation satisfies, for $p\in C^3([0,1])$,
\begin{equation}\label{eq:interpolation-error}
 \int_0^1|p-I_2p|\,ds\le\frac1{96}\int_0^1|p'''|\,ds.
\end{equation}
Indeed, choose $g'''=|p'''|$ and put $w(s)=\sign(1/2-s)$.
The Lagrange remainder for $g\pm p$ gives $|p-I_2p|\le w(g-I_2g)$.
For $H(s)=u^2(3-4u)/24$, $u=\min(s,1-s)$, we have
$H\in C^2$, $H=H'=0$ at the endpoints, $-H'''=w$ almost everywhere,
and $0\le H\le1/96$. Three integrations by parts against $g-I_2g$,
which vanishes at the endpoints, bound the left side by $\int_0^1H|p'''|$, proving the claim.
On a panel in $J$, scaling contributes $h_J^3$ and comparison of
Gaussian densities contributes at most $e^{h_J\max_J|s|}$.

Let $T=\{|s|>9\}$ and let $\mathcal P_0$ partition $\R$ into
the $36$ intervals $J$ and two exterior intervals, up to endpoints. Put
\[
 G_L=\int_L\psi\psi^T\,d\gamma_1,\qquad
 Q_J=\int_J\psi'''(\psi''')^T\,d\gamma_1,
\]
and, for $Q\in\R^{22\times22}$, $Q\succeq0$, and a finite measurable partition $\mathcal P$, define
\[
 \mathcal A(Q;\mathcal P)
 =\sum_{i=1}^2\sqrt{(\ell^{[i]})^TQ\ell^{[i]}}
 +2\sum_{i,j=1}^2\sum_{L\in\mathcal P}
\sqrt{\gamma_1(L)\Tr((R^{[ij]})^TQR^{[ij]}G_L)}.
\]
Splitting each feasible moment over $\mathcal P$ and applying the
triangle inequality gives
$\sum_i\|Q^{1/2}\xi^{[i]}\|\le\mathcal A(Q;\mathcal P)$, since
Cauchy--Schwarz gives, for a real matrix $R_0$ with $22$ columns and measurable $h\colon\R\to[-1,1]$,
\[
 \left\|R_0\int_L\psi h\,d\gamma_1\right\|^2
 \le\gamma_1(L)\Tr(R_0G_LR_0^T)\qquad(|h|\le1).
\]
Another Cauchy--Schwarz inequality bounds the derivative and tail
integrals. Thus the uniform error allowance is any rational $E$ with
\begin{equation}\label{eq:analytic-error}
 E\ge\sqrt{\gamma_1(T)}\mathcal A(G_T;\{\R\})
 +\sum_J\frac{e^{h_J\max_J|s|}h_J^3}{96}
           \sqrt{\gamma_1(J)}\mathcal A(Q_J;\mathcal P_0),
\end{equation}
where $G_\R=I_{22}$.
The polynomial integrands have degree at most $42$; their Gaussian
moments follow from \cref{heavy:gaussian-moments}, using
reflection and infinite-endpoint limits for the tails.
The verifier uses equivalent Hermite identities.

\paragraph{Integer rounding and the target.}
Bars denote rounding $\mathcal B$ to multiples of $2^{-42}$ and
$R,\ell$ to multiples of $2^{-40}$, with $\overline R$ symmetric.
With entrywise absolute values, set
$b_j=\sum_k|\overline{\mathcal B}_{jk}|$, choose rational
$\delta_j\ge\sum_k|\mathcal B_{jk}-\overline{\mathcal B}_{jk}|$,
and put $b_+=b+\delta$.
The bounds $|\overline{\mathcal B}\sigma|\le b$ and
$|\mathcal B\sigma-\overline{\mathcal B}\sigma|\le\delta$
give, by expansion, the uniform input-error allowance
\begin{equation}\label{eq:input-error}
 E_{\mathrm{input}}\ge
 \delta^T|\overline R|(2b+\delta)+b_+^T|R-\overline R|b_+
 +|\overline\ell|^T\delta+|\ell-\overline\ell|^Tb_+.
\end{equation}
After substituting $\sigma_j^2=1$, multiply the rounded quadratic by
$2^{44}$ and round its coefficients to nearest integers, obtaining
$P=a_0+\sum_ja_j\sigma_j+\sum_{j<k}a_{jk}\sigma_j\sigma_k$.
Sign monomials have absolute value one, so the sum of absolute
coefficient-rounding errors bounds the uniform error. Let
$E_{\mathrm{graph}}$ be this sum divided by $2^{44}$, including the
once-rounded total diagonal.
Thus, with $\Delta=E_{\mathrm{input}}+E_{\mathrm{graph}}$,
$|F(\mathcal B\sigma)-2^{-44}P(\sigma)|\le\Delta$.
Equation~\eqref{eq:finite-signs} gives
$F(M)\le E+\Delta+2^{-44}P(\sigma)$, so it suffices to certify
\begin{equation}\label{eq:finite-target}
 \max_{\sigma\in\{-1,1\}^{N_{\mathrm{sign}}}}P(\sigma)
 \le U:=\lfloor2^{44}(C-E-\Delta)\rfloor.
\end{equation}

\paragraph{Graph bounds and sign fixing.}
Use vertices $j^+,j^-$ and terminals of signs $+1,-1$.
For each nonzero $a_{jk}$, join equal copies if $a_{jk}>0$ and
opposite copies otherwise, with capacity $|a_{jk}|$ on both edges.
Join $j^+$ to terminal $\sign(a_j)$ and $j^-$ to the opposite
terminal, with capacity $|a_j|$; omit zero capacities.
A cut assigns vertex signs and sums capacities across opposite signs.
For consistent cuts, the signs are $(\sigma_j,-\sigma_j)$, giving
\begin{equation}\label{eq:cut}
 \max_\sigma P(\sigma)
 \le a_0+\sum_j|a_j|+\sum_{j<k}|a_{jk}|-\min_x\operatorname{cap}(x),
\end{equation}
with equality when minimizing over consistent cuts.
Represent each edge by opposite directed arcs. Check flow nonnegativity,
capacity bounds, and conservation at nonterminals. Summing conservation
shows that a flow and cut of equal value certify a minimum cut.

A minimum cut $x$ also fixes opposite pairs. Replacing each
$x_{j^\pm}$ by $-x_{j^\mp}$ preserves capacity.
For consistent $z$, sort $z$, $x$, and this reflected cut coordinatewise into $a\le y\le b$.
Sorting cannot increase total capacity: each edge's disagreements become
the absolute difference of its endpoint counts of positive signs.
By minimality,
\[
 \operatorname{cap}(y)\le\operatorname{cap}(z)+2\operatorname{cap}(x)
 -\operatorname{cap}(a)-\operatorname{cap}(b)\le\operatorname{cap}(z).
\]
The coordinatewise median satisfies $y_{j^+}=-y_{j^-}$ and
$y_{j^\pm}=x_{j^\pm}$ whenever $x_{j^+}=-x_{j^-}$.
Thus some maximizing assignment obeys all fixings. Substitute them,
retaining constants, and split remaining signs as needed. Check both
children and require every terminal bound to be at most $U$.

\paragraph{Rectangle certificates.}
Checked invariance of $P$ under grid reflection combined with negation
of both sign blocks allows the first block's sign at zero to be fixed:
$\sigma_{j_0}=1$.
For supplied nonzero $v_1,v_2\in\mathbb Z^{N_{\mathrm{sign}}}$, put
$u_k=v_k^T\sigma/\|v_k\|_1\in[-1,1]$.
Recursively split $[-1,1]^2$ at rational interior points, retaining
both closed children. On a leaf $\mathcal R=\prod_k[l_k,h_k]$, add
\begin{equation}\label{eq:rectangle-quadratic}
 2^{44}\sum_{k=1}^2\mu_k(u_k-l_k)(h_k-u_k),\qquad\mu_k\ge0.
\end{equation}
The correction is nonnegative on $\mathcal R$.
Expand with $\sigma_j^2=1$ and round each coefficient contribution,
including the separate constants
$-2^{44}\mu_kl_kh_k$ and
$-2^{44}\mu_k\|v_k\|^2/\|v_k\|_1^2$, to a nearest integer.
For each $k$, charge half the number of nonintegral
contributions, rounded upward, to cover their absolute rounding errors.
Apply graph bounds and sign fixing to the augmented polynomial over
all signs with $\sigma_{j_0}=1$, unrestricted by $\mathcal R$.
Its bound plus the charges must be at most $U$.

Alternatively, supplied integers $a,b$ exclude a leaf if, for $w=av_1+bv_2$,
\[
 \min_{(s_1,s_2)\in\mathcal R}
 \bigl(a\|v_1\|_1s_1+b\|v_2\|_1s_2\bigr)
 >w_{j_0}+\sum_{j\ne j_0}|w_j|.
\]
The left side is evaluated at the minimizing corner; the right side
is $\max_{\sigma:\sigma_{j_0}=1}w^T\sigma$.
The complete cover and checked leaves prove \cref{eq:finite-target}.

\paragraph{Acceptance and reproduction.}\label{sec:reproduction}
The infinite part of \cref{eq:rotation-tail} follows from the
rational inequalities
\begin{equation}\label{eq:tail-test}
 \eta_N:=w_1^{(N)}-
 \sum_{j=2}^{(N-1)/2}\max\{0,-w_j^{(N)}\}>0\qquad(N=7,9).
\end{equation}
Indeed, for $m\ge0$ and $c_j=\cos(\pi j/N)$,
$\sum_jw_j^{(N)}c_j^m\ge c_1^m\eta_N>0$, since $0<c_j/c_1\le1$.

The exact rows and subdivisions are in
{\path{data/lower_rows.json}}; proof trees and verifier
are in {\path{lower/certificates/} and \path{lower/verify.py}, respectively}.
The supplied checks verify coverage of $[0,1]$, rotation maxima,
\cref{eq:tail-test}, $1-\Gamma_1>0$, all reduction matrices,
the unshifted ${R\succ0}$, regenerated error allowances and integer inputs,
the finite certificates, and \cref{eq:endpoint}.
Analytic calculations use outward-rounded $384$-bit intervals and
positive $LDL^T$ pivots for matrix inequalities; {\path{lower/check-tail.py}}
checks \cref{eq:tail-test} with exact rational arithmetic.
All $35$ rows pass, with largest payoff bound
$0.564015352833<{v_*}$, rounded upward.
Hence \cref{eq:finite-signs} and \cref{eq:finite-target} give $F\le C$;
\cref{thm:reduction} and {\cref{thm:duality}} prove \cref{thm:lower}.

{Install the README's dependencies,
and run from the repository root:
\begin{verbatim}
python verify.py --lower --workers 4
\end{verbatim}
This command includes the exact tail check and writes
\path{runs/complete-*/ACCEPTED.json}}
only on success;
\texttt{--data-only} checks input identities and structure, not the proof.

\section{Upper bound certificate}\label{app:upper-bound}
For the Hermite integrals below, $W,X_1,X_2,X_3$ are independent standard
Gaussians and $X=(X_1,X_2,X_3)$ (comparing to the notation of \cref{sec:upper-rounding},
we no longer mention the Gaussians used to define $W$ since they do not appear in the analysis).
Define
\begin{equation}\label{root:native}
 C[a,b]:=\sum_{|\alpha|=b}
 \bigl(\E[F(W,X)\psi_a(W)\psi_\alpha(X)]\bigr)^2
 \quad(a,b\ge0,\ \alpha\in\{0,1,\ldots\}^3).
\end{equation}
The one-coordinate case of \cref{eq:correlation} and
$\psi_a(-x)=(-1)^a\psi_a(x)$ give
\begin{align}
 H_i(t)&=\sum_{a,b\ge0}C[a,b]P_i(t)^a(-t)^b
 \quad(i=0,\ldots,13,15,16),\label{root:additive-kernel}\\
 H_{14}(t)&=\sum_{a,b\ge0}(-1)^a
 \bigl(\E[F_*(X_1,X_2)\psi_a(X_1)\psi_b(X_2)]\bigr)^2t^{3a+5b}.
 \label{eq:upper-staircase-kernel}
\end{align}
For $17\le i\le31$, rotational symmetry of two independent standard
Gaussians gives the sign-disagreement probability $\arccos(P_i(t))/\pi$, hence
\begin{equation}\label{root:linear}
 H_i(t)=1-\frac2\pi\arccos P_i(t)
       =\frac2\pi\arcsin P_i(t)\quad(17\le i\le31).
\end{equation}
Equation~\eqref{root:odd-hermite} and Parseval give
\[
 C[a,b]\ge0,\qquad \sum_{a,b}C[a,b]=\E F^2=1,\qquad
 C[a,b]=0\quad(a+b\text{ even}).
\]

\subsection{Numerical bounds for the coefficients up to degree 681}
\label{head:section}

\begingroup

\theoremstyle{plain}
\newtheorem{headclaim}{Claim}
\begin{lemma}
The coefficients of $H$ satisfy
\begin{equation}\label{eq:upper-certified-head}
 c_1-\sum_{m=2}^{681}|c_m|>0.5618773799453352.
\end{equation}
\end{lemma}
\begin{proof}
Define
\[
 F_{\rm tr}(W,X):=F(W,X)\prod_{j=1}^3\mathbf1_{\{|X_j|\le73/10\}},
 \qquad C_{\rm tr}[a,b]:=\sum_{|\alpha|=b}
 \left(\E[F_{\rm tr}(W,X)\psi_a(W)\psi_\alpha(X)]\right)^2,
\]
and $\operatorname{erfc}(x):=2\pi^{-1/2}\int_x^\infty e^{-t^2}\,dt$.
Choose $\hat C[a,b]$ as in Claim~\ref{head:claim-numerical}.
By Parseval and \cref{head:square-difference},
\begin{equation}\label{head:cutoff}
\begin{aligned}
 \sum_{a,b\ge0}|C_{\rm tr}[a,b]-C[a,b]|
 &\le\|F-F_{\rm tr}\|_2(\|F\|_2+\|F_{\rm tr}\|_2)\\
 &\le2\sqrt{3\operatorname{erfc}(73/(10\sqrt2))}
 <1.858282199124\cdot10^{-6},
\end{aligned}
\end{equation}
since $\|F_{\rm tr}\|_2\le\|F\|_2=1$ and
\[
 \|F-F_{\rm tr}\|_2^2
 =\Pr\!\left(\max_j|X_j|>73/10\right)
 \le\sum_{j=1}^3\Pr(|X_j|>73/10)
 =3\operatorname{erfc}(73/(10\sqrt2)).
\]
For square-summable real vectors $x,y$,
\begin{equation}\label{head:square-difference}
\begin{aligned}
 \sum_j|x_j^2-y_j^2|
 &\le\|x-y\|_2\|x+y\|_2\le\|x-y\|_2(\|x\|_2+\|y\|_2)\\
 &\le\|x-y\|_2(2\|x\|_2+\|x-y\|_2).
\end{aligned}
\end{equation}
For $i\ne14$, set $P_{i,+}(t):=\sum_{m\ge1}|[t^m]P_i(t)|t^m$.
Then
\begin{equation}\label{root:coefficient-majorant}
\begin{aligned}
 |[t^m](P_i(t)^a(-t)^b)|&\le[t^m](P_{i,+}(t)^at^b),\\
 \sum_{m\ge0}[t^m](P_{i,+}(t)^at^b)&=P_{i,+}(1)^a\le1,\qquad
 [t^m](P_{i,+}(t)^at^b)=0\quad(m<a+b).
\end{aligned}
\end{equation}
Only $a+b\le681$ contributes through degree $681$, by
\cref{root:coefficient-majorant}. That equation,
\cref{head:joint-plane}, and \cref{head:cutoff} give
\begin{equation}\label{head:head-error}
\begin{aligned}
 &\sum_{m=1}^{681}\left|[t^m]
 \sum_{i\in\{0,\ldots,13,15,16\}}\lambda_i
 \sum_{a+b\le681}(\hat C[a,b]-C[a,b])P_i(t)^a(-t)^b\right|\\
 &\quad\le\left(\sum_{i\in\{0,\ldots,13,15,16\}}\lambda_i\right)
             \sum_{a+b\le681}|\hat C[a,b]-C[a,b]|\\
 &\quad<\frac{1315400222548137891645}{1351351351351351364726}
 \left(6.078879564946870\cdot10^{-8}+1.858282199124\cdot10^{-6}\right)\\
 &\quad<1.868016346073\cdot10^{-6},
\end{aligned}
\end{equation}
where the fraction is the sum of the indicated weights.
For \cref{root:coefficient-majorant}, induction from $P_i^0=P_{i,+}^0=1$
uses
\[
 |[t^m]P_i^a|\le\sum_{{j}=0}^m|[t^{j}]P_i|\,|[t^{m-{j}}]P_i^{a-1}|
 \le[t^m]P_{i,+}^a.
\]
Since $P_{i,+}(0)=0$,
\[
 [t^m]P_{i,+}^a=0\quad(m<a),\qquad
 [t^m](P_{i,+}^at^b)=
 \begin{cases}[t^{m-b}]P_{i,+}^a,&m\ge b,\\0,&m<b.\end{cases}
\]
The coefficient sum equals $P_{i,+}(1)^a\le1$.
The triangle inequality and \cref{head:finite-margin} yield
\[
\begin{aligned}
 c_1-\sum_{m=2}^{681}|c_m|
 &>0.56187924796168130219-1.868016346073\cdot10^{-6}\\
 &=0.56187737994533522919>0.5618773799453352.
\end{aligned}
\]
\end{proof}

\begin{headclaim}\label{head:claim-numerical}
There exist nonnegative rational numbers $\hat C[a,b]$, $a+b\le681$, with
\begin{equation}\label{head:joint-plane}
 \sum_{a+b\le681}|\hat C[a,b]-C_{\rm tr}[a,b]|
 \le6.078879564946870\cdot10^{-8}
\end{equation}
and
\begin{equation}\label{head:finite-margin}
\begin{aligned}
 &[t]\left(\sum_{i\in\{0,\ldots,13,15,16\}}\lambda_i
   \sum_{a+b\le681}\hat C[a,b]P_i(t)^a(-t)^b
   +\lambda_{14}H_{14}(t)+\sum_{i=17}^{31}\lambda_iH_i(t)\right)\\
 &\quad-\sum_{m=2}^{681}\left|[t^m]\left(
   \sum_{i\in\{0,\ldots,13,15,16\}}\lambda_i
   \sum_{a+b\le681}\hat C[a,b]P_i(t)^a(-t)^b
   +\lambda_{14}H_{14}(t)+\sum_{i=17}^{31}\lambda_iH_i(t)
   \right)\right|\\
 &\qquad>0.56187924796168130219.
\end{aligned}
\end{equation}
\end{headclaim}
\begin{proof}
\smallskip\noindent\textbf{1. Express the Hermite coefficients as integrals.}
Set $\phi(z):=e^{-z^2/2}/\sqrt{2\pi}$ and, for real $s$,
\[
 B_j(s):={\mathrm{i}}^{-j}\E[e^{{\mathrm{i}}sr(X_1)}\psi_j(X_1)\mathbf1_{\{|X_1|\le73/10\}}],
 \qquad B(s):=(B_j(s))_{j=0}^{681},\quad {\mathrm{i}}^2=-1.
\]
For $a+|\alpha|\le681$ odd,
\begin{equation}\label{head:coefficient-integral}
 \E[F_{\rm tr}(W,X)\psi_a(W)\psi_\alpha(X)]
 =\frac{2(-1)^{(a+|\alpha|-1)/2}}{\pi\sqrt{a!}}
       \int_0^\infty e^{-s^2/2}s^{a-1}\prod_{j=1}^3B_{\alpha_j}(s)\,ds.
\end{equation}
For even $a+|\alpha|$ the coefficient vanishes by
$F_{\rm tr}(-W,-X)=-F_{\rm tr}(W,X)$ and \cref{root:odd-hermite}.
Reflection gives $\overline{B_j(s)}={\mathrm{i}}^j(-1)^j {\mathrm{i}}^jB_j(s)=B_j(s)$,
and Bessel's inequality gives $\|B(s)\|_2\le1$.
To prove \cref{head:coefficient-integral}, for $a\ge0$ and $v\in\R$ use
\begin{equation}\label{head:smoothed-sign}
\begin{aligned}
 q_a(v)&:=\E[{\sign}(W+v)\psi_a(W)]
       =\frac{q_0^{(a)}(v)}{\sqrt{a!}},\qquad
 q_0(v)=2\int_0^v\phi(u)\,du,\\
 q_a^{(l)}(v)&=\frac{2(-1)^{a+l-1}}{\sqrt{a!}}
                  \phi(v)\He_{a+l-1}(v)
        \quad(a,l\ge0,\ a+l\ge1),\\
 q_a(v)&=\frac2{\pi\sqrt{a!}}\int_0^\infty
          e^{-s^2/2}s^{a-1}\Im({\mathrm{i}}^ae^{{\mathrm{i}}sv})\,ds.
\end{aligned}
\end{equation}
Conditioning on $X$ and independence give
\[
\begin{aligned}
 \E[F_{\rm tr}\psi_a(W)\psi_\alpha(X)]
 &=\E\!\left[q_a\!\left(\sum_jr(X_j)\right)
       \psi_\alpha(X)\prod_j\mathbf1_{\{|X_j|\le73/10\}}\right],\\
 \E\!\left[e^{{\mathrm{i}}s\sum_jr(X_j)}\psi_\alpha(X)
       \prod_j\mathbf1_{\{|X_j|\le73/10\}}\right]
 &={\mathrm{i}}^{|\alpha|}\prod_jB_{\alpha_j}(s).
\end{aligned}
\]
Fubini follows from $\E|\psi_\alpha(X)|\le1$ and the integrability of
$e^{-s^2/2}s^{a-1}$ for $a\ge1$; for $a=0$, use
$|\sin(s\sum_jr(X_j))|/s\le3\sup_{|x|\le73/10}|r(x)|$ when
$\max_j|X_j|\le73/10$.
The first two identities in \cref{head:smoothed-sign} follow from
\[
 (\phi\He_{a-1})'=-\phi\He_a,\qquad
 \frac2{\sqrt{a!}}\int_{-v}^\infty\He_a(u)\phi(u)\,du
 =\frac{2(-1)^{a-1}}{\sqrt{a!}}\phi(v)\He_{a-1}(v)\quad(a\ge1).
\]
For the Fourier identity, both sides at $a=0$ vanish at $v=0$, and
\[
 \frac d{dv}\left(\frac2\pi\int_0^\infty e^{-s^2/2}
                 \frac{\sin(sv)}s\,ds\right)
 =\frac2\pi\int_0^\infty e^{-s^2/2}\cos(sv)\,ds=2\phi(v).
\]
The bounds $|\sin(sv)|/s\le|v|$ and
$|\partial_v^a(\sin(sv)/s)|\le s^{a-1}$ for $a\ge1$, multiplied by
$e^{-s^2/2}$, are integrable in $s$, locally uniformly for real $v$.

\smallskip\noindent\textbf{2. Define the finite approximation $\hat C[a,b]$.}
Construct $\hat C$ by the operations below, with fixed summation orders
and fixed choices of separate or combined multiplication and addition.
Use the same resulting rational array in \cref{head:finite-margin}.
Let $s_\nu,\omega_\nu$, $1\le\nu\le768$, be the 24-point
Gauss--Legendre nodes and weights on each $[k,k+1]$, $0\le k\le31$.
Their weighted sum equals the integral on that interval for every polynomial
of degree at most $47$. Set
\[
 U_{a\nu}:=\omega_\nu e^{-s_\nu^2/2}s_\nu^{a-1}/\sqrt{a!}
 \qquad(0\le a\le681).
\]
In interval arithmetic, each operation encloses the exact values by
rounding the lower endpoint downward and the upper endpoint upward,
at the stated precision.
For $[a,b]\subset[0,73/10]$, set $h:=(b-a)/2$ and
\begin{equation}\label{head:inner-M}
\begin{aligned}
 M:={}&\frac{(100/99)^{681/2}(10000/199)^{1/4}}{\sqrt{2\pi}}
 \exp\left(-\frac{50}{199}\max(0,a-h/4)^2+50(3h/4)^2\right)\\
 &\quad\cdot\exp\left(24h\sum_{k=1}^{15}
 \frac{|r^{(k)}((a+b)/2)|}{(k-1)!}(5h/4)^{k-1}\right).
\end{aligned}
\end{equation}
Starting with $[0,73/10]$, accept an interval only when its
512-bit interval evaluation proves
$8(b-a)M2^{-64}<10^{-10}(b-a)/(73/5)$; otherwise bisect it.
Use the 32-point
Gauss--Legendre rule, exact for polynomials of degree at most $63$,
on each interval. Index its positive nodes by $j$;
for the node $x$ in $[a,b]$, define
\[
 v_j:=2\phi(x)(\psi_m(x))_{m=0}^{681}
 \int_a^b\prod_{\substack{y\ne x\\y\text{ a quadrature node in }[a,b]}}
               \frac{t-y}{x-y}\,dt.
\]
The integral is the quadrature weight. Reflection supplies the factor two:
${\mathrm{i}}^{-m}(e^{{\mathrm{i}}sr(x)}+(-1)^me^{-{\mathrm{i}}sr(x)})=2\Re({\mathrm{i}}^{-m}e^{{\mathrm{i}}sr(x)})$.
At $512$ bits, enclose the exact quadrature nodes
and use these enclosures to compute $v_j$, $U_{a\nu}$, $4/\pi^2$, and
$\Re({\mathrm{i}}^{-m}e^{{\mathrm{i}}s_\nu r(x)})$. Round the midpoint of each resulting
enclosure to the nearest IEEE 754 binary64 number, which has at most
$53$ significant binary digits; at a tie, choose the number whose last
stored binary digit is zero. Bound its conversion error against the
whole enclosure. Hats denote computed values whose stored real and
imaginary components are interpreted as exact rationals.
Rounding an exact real result $x$ to the nearest binary64 number has
absolute error at most $2^{-53}|x|+2^{-1074}$, provided that $x$ is
within the representable range.
Compute $\hat B_{\nu m}$ by summing the products of $(\hat v_j)_m$
with the stored values of $\Re({\mathrm{i}}^{-m}e^{{\mathrm{i}}s_\nu r(x)})$, using at most
one multiplication and one addition per term. Round each operation,
or round a multiplication and its following addition together. Write
$\hat B_{\nu\cdot}:=(\hat B_{\nu m})_{m=0}^{681}$.
Let $\hat U_{a\nu}$ denote the conversion before omitting entries
below $10^{-18}$, and, for each
$(\nu,\mu)$, set $g(z):=\sum_{j=0}^{681}\hat B_{\nu j}\hat B_{\mu j}z^j$
with exact products. For odd $a+b\le681$, compute $\hat C[a,b]$ from
\begin{equation}\label{head:computed-sum}
 {\widehat{\frac4{\pi^2}}}
 \sum_{\mu:\hat U_{a\mu}\ge10^{-18}}\hat U_{a\mu}
 \left(\sum_{\nu:\hat U_{a\nu}\ge10^{-18}}\hat U_{a\nu}[z^b]
 \left(\sum_{j=0}^{681}\hat B_{\nu j}\hat B_{\mu j}z^j\right)^3\right).
\end{equation}
Round the products defining $g$ and append zeros to obtain $2048$ entries.
The exact formula approximated by the forward transform, two multiplications
for each cube, and the inverse transform is
\begin{equation}\label{head:discrete-transform}
 [z^b]g(z)^3=\frac1{2048}\sum_{k=0}^{2047}
  g(e^{-2\pi {\mathrm{i}}k/2048})^3e^{2\pi {\mathrm{i}}kb/2048}\quad(0\le b\le2043),
\end{equation}
since $\deg(g^3)\le2043$. First reverse the eleven binary digits of
the indices $0,\ldots,2047$. For $k=1,\ldots,11$, split each consecutive
block of $2^k$ entries into halves and pair their $j$th entries,
$0\le j<2^{k-1}$, applying $(a,b)\mapsto(a+wb,a-wb)$ with
$w=e^{-2\pi {\mathrm{i}}j/2^k}$. Obtain $\hat w$ by rounding the real and imaginary
parts of $w$ from their 256-bit enclosure midpoints, with error $<2^{-48}$
against each enclosure.
The inverse conjugates these roots and ends with division by $2048$.
Take real parts, sum first over $\nu$ and then over $\mu$ in
\cref{head:computed-sum}, and multiply by ${\widehat{\frac4{\pi^2}}}$ with the
same rounding convention; replace negative results
by zero and set $\hat C[a,b]=0$ when $a+b$ is even.

\smallskip\noindent\textbf{3. Bound the error in $\hat C$.}
For $s>0$ and real vectors $v,w$ of equal dimension,
\[
 \sum_{a=0}^{681}e^{-s^2}s^{2a-2}/a!\le s^{-2},\qquad
 \|v^{\otimes3}-w^{\otimes3}\|_2
 \le3\|v-w\|_2\max(\|v\|_2,\|w\|_2)^2.
\]
Together with $\|B(s)\|_2\le1$, the 512-bit bound
$\sum_{\nu=1}^{768}\omega_\nu/s_\nu<11.017652258306741$, and
Claim~\ref{head:claim-quadrature}, these inequalities give
\[
\begin{aligned}
 &\left\|\left(
 \E[F_{\rm tr}\psi_a(W)\psi_\alpha(X)]
 -\frac{2(-1)^{(a+|\alpha|-1)/2}}\pi
       \sum_{\nu=1}^{768}U_{a\nu}\prod_{j=1}^3\hat B_{\nu\alpha_j}
 \right)_{\substack{a+|\alpha|\le681\\a+|\alpha|\ \mathrm{odd}}}\right\|_2\\
 &\quad\le2.917598329143\cdot10^{-10}+8.228746\cdot10^{-18}\\
 &\qquad+\frac6\pi\sum_{\nu=1}^{768}\frac{\omega_\nu}{s_\nu}
 (5.604655228362\cdot10^{-12})(1+5.604655228362\cdot10^{-12})^2\\
 &\quad<4.096939258130\cdot10^{-10}.
\end{aligned}
\]
Since $\|F_{\rm tr}\|_2\le1$, \cref{head:square-difference} yields
\[
\begin{aligned}
 &\sum_{\substack{a+b\le681\\a+b\ \mathrm{odd}}}\left|
 C_{\rm tr}[a,b]-\frac4{\pi^2}\sum_{|\alpha|=b}
 \left(\sum_{\nu=1}^{768}U_{a\nu}
             \prod_{j=1}^3\hat B_{\nu\alpha_j}\right)^2\right|\\
 &\quad<2(4.096939258130\cdot10^{-10})
             +(4.096939258130\cdot10^{-10})^2<8.193878518\cdot10^{-10}.
\end{aligned}
\]
\begin{equation}\label{head:nonnegative-squares}
\begin{aligned}
 &\sum_{|\alpha|=b}\left(\sum_{\nu=1}^{768}U_{a\nu}
                      \prod_{j=1}^3\hat B_{\nu\alpha_j}\right)^2\\
 &\quad=\sum_{\nu,\mu=1}^{768}U_{a\nu}U_{a\mu}[z^b]
       \left(\sum_{j=0}^{681}\hat B_{\nu j}\hat B_{\mu j}z^j\right)^3.
\end{aligned}
\end{equation}
Thus Claim~\ref{head:claim-arithmetic} proves \cref{head:joint-plane}:
\[
 \sum_{a+b\le681}|\hat C[a,b]-C_{\rm tr}[a,b]|
 \le8.193878518\cdot10^{-10}+5.996940779766870\cdot10^{-8}
 =6.078879564946870\cdot10^{-8}.
\]

\smallskip\noindent\textbf{4. Evaluate the finite coefficient inequality.}
For $1\le m\le681$, \cref{eq:upper-staircase-kernel} and \cref{root:linear} give
\[
\begin{aligned}
 [t^m]H_{14}(t)&=\sum_{\substack{a,b\ge0\\3a+5b=m}}(-1)^a\frac{u_{ab}^2}{a!b!},\\
 [t^m]H_i(t)&=\frac2\pi\sum_{k=0}^{340}
 \frac{\binom{2k}{k}}{4^k(2k+1)}[t^m]P_i(t)^{2k+1}
 \qquad(17\le i\le31),
\end{aligned}
\]
where $u_{ab}$ is defined in \cref{other:staircase-moments} and
$\E[F_*(X_1,X_2)\psi_a(X_1)\psi_b(X_2)]=u_{ab}/\sqrt{a!b!}$.
To evaluate $u_{ab}$, define $J_j(x):=\int_x^\infty\He_j(t)\phi(t)\,dt$ and use
\begin{equation}\label{head:hermite-integrals}
\begin{aligned}
 \int_x^y\psi_j(t)\phi(t)\,dt
 &=\frac{J_j(x)-J_j(y)}{\sqrt{j!}}
 &&(-\infty\le x<y\le\infty),\\
 J_0(x)&=\tfrac12\operatorname{erfc}(x/\sqrt2),&
 J_j(x)&=\phi(x)\He_{j-1}(x)\quad(j\ge1),\\
 J_j(-\infty)&=\mathbf1_{\{j=0\}},&J_j(+\infty)&=0.
\end{aligned}
\end{equation}
These identities are proved in Claim~\ref{other:claim-staircase}; evaluate them
at $384$ bits on the coordinate intervals of the rectangles where $F_*$ is constant.
For the remaining polynomial arithmetic, retain only degrees through 681:
$P_i(0)=0$ implies $[t^m]P_i^{2k+1}=0$ for $m\le681$, $k>340$, and
products of nonnegative powers cannot lower degree.
Interval evaluation at $256$ bits, using the rational $\hat C[a,b]$, bounds the
first term in \cref{head:finite-margin} below by
$0.561894649254995327808$ and the subtracted sum above by
$1.540129331402561534750\cdot10^{-5}$. Hence
\[
 0.561894649254995327808-1.540129331402561534750\cdot10^{-5}
 >0.56187924796168130219.
\]
\end{proof}

\begin{headclaim}\label{head:claim-quadrature}
The quadratures in step 2 satisfy
\begin{equation}\label{head:outer-error}
\begin{aligned}
 &\left\|\left(
 \E[F_{\rm tr}\psi_a(W)\psi_\alpha(X)]
 -\frac{2(-1)^{(a+|\alpha|-1)/2}}\pi
       \sum_{\nu=1}^{768}U_{a\nu}\prod_{j=1}^3B_{\alpha_j}(s_\nu)
 \right)_{\substack{a+|\alpha|\le681\\a+|\alpha|\ \text{odd}}}\right\|_2\\
 &\qquad<2.917598329143\cdot10^{-10}+8.228746\cdot10^{-18},
\end{aligned}
\end{equation}
and, for every $1\le\nu\le768$,
\begin{equation}\label{head:B-error}
 \|\hat B_{\nu\cdot}-B(s_\nu)\|_2<5.604655228362\cdot10^{-12}.
\end{equation}
\end{headclaim}
\begin{proof}
For an interval $[a,b]$ of length $\ell$, consider the ellipse
\[
 z=\frac{a+b}{2}+\frac{b-a}{4}
       (2e^{{\mathrm{i}}\theta}+2^{-1}e^{-{\mathrm{i}}\theta}),\qquad0\le\theta\le2\pi.
\]
If the integrand is analytic inside and on this ellipse and has norm
at most $M$ there, the $n$-point Gauss--Legendre error is at most
\begin{equation}\label{head:gauss-error}
 \frac\ell2\frac{64M}{15(1-2^{-2})2^{2n}}
 =\frac{128}{45}\ell M2^{-2n}\le8\ell M2^{-2n}.
\end{equation}
The scalar bound is \cite[Theorem~2.1]{HaleTrefethen2008}.
For complex Euclidean or Gaussian ${L_2}$ spaces, apply it after bounded
complex-linear functionals $v$ and use
$\|x\|=\sup_{\|v\|\le1}|v(x)|$, with
$\|v\|:=\sup_{\|x\|\le1}|v(x)|$.

\smallskip\noindent\textbf{The outer quadrature.}
Apply \cref{head:gauss-error} in Gaussian ${L_2}$ to
\[
 s\longmapsto\frac2\pi\frac{\sin(s(W+\sum_{j=1}^3r(X_j)))}s
                  \prod_{j=1}^3\mathbf1_{\{|X_j|\le73/10\}}.
\]
Its singularity at zero is removable. Since $r$ is bounded on
$[-73/10,73/10]$ and $\E e^{s|W|}<\infty$ for every real $s$,
the displayed function is analytic in Gaussian ${L_2}$.
On the ellipses for $[k,k+1]$, $|\Im s|\le3/8$,
so $|\sin(sv)/s|=|v\int_0^1\cos(tsv)\,dt|\le|v|e^{3|v|/8}$ for real $v$.
\[
\begin{aligned}
 \E[e^{3|r(X_1)|/4}\mathbf1_{\{|X_1|\le73/10\}}]&<2.262302379287,\\
 \E[r(X_1)^2e^{3|r(X_1)|/4}\mathbf1_{\{|X_1|\le73/10\}}]&<2023.738286188,\\
 \E e^{3|W|/4}&\le2e^{9/32},\qquad
 \E W^2e^{3|W|/4}\le(25/8)e^{9/32}.
\end{aligned}
\]
By Bessel's inequality, retaining the Hermite coefficients with
$a+|\alpha|\le681$ and $a+|\alpha|$ odd cannot increase the ${L_2}$ norm.
Independence and Minkowski's inequality therefore bound the contribution
of the quadrature on $[0,32]$ to \cref{head:outer-error} by
\[
\begin{aligned}
 &8\cdot32\cdot2^{-48}\frac{2\sqrt2}{\pi}e^{9/64}
 \left(\frac54(2.262302379287)^{3/2}
       +3\sqrt{2023.738286188}\,(2.262302379287)\right)\\
 &\qquad<2.917598329143\cdot10^{-10}.
\end{aligned}
\]
For the first two moments, partition $[0,73/10]$ into intervals of length
$10^{-3}$. On each $[a,b]$, let $R$ be the upper endpoint of the modulus
of the 512-bit interval polynomial evaluation $r([a,b])$ and sum
\[
\begin{aligned}
 2\int_a^b e^{3|r(x)|/4}\phi(x)\,dx&\le2(b-a)\phi(a)e^{3R/4},\\
 2\int_a^b r(x)^2e^{3|r(x)|/4}\phi(x)\,dx&\le2(b-a)\phi(a)R^2e^{3R/4}.
\end{aligned}
\]
The factor two uses oddness of $r$ and evenness of $\phi$.
For $s\ge32$, Bessel's inequality and a geometric sum give
\[
 \sum_{|\alpha|\le681-a}\left|\prod_jB_{\alpha_j}(s)\right|^2\le1,
 \qquad
 \sum_{a=0}^{681}\frac{s^{2a}}{a!}
 \le\frac{s^{1362}}{681!}\sum_{j=0}^{681}(681/s^2)^j
 \le\frac{1024s^{1362}}{343\cdot681!}.
\]
Minkowski's inequality bounds the integral over $s>32$ by
\begin{equation}\label{head:fourier-tail}
\begin{aligned}
 &\left\|\left(\frac2{\pi\sqrt{a!}}\int_{32}^\infty
 e^{-s^2/2}s^{a-1}\prod_{j=1}^3B_{\alpha_j}(s)\,ds
 \right)_{\substack{a+|\alpha|\le681\\a+|\alpha|\ \text{odd}}}\right\|_2\\
 &\quad\le\frac{64}{\pi\sqrt{343\cdot681!}}
              \int_{32}^\infty e^{-s^2/2}s^{680}\,ds
 \le\frac{64e^{-512}32^{681}}{345\pi\sqrt{343\cdot681!}}
 <3\cdot10^{-18}<8.228746\cdot10^{-18}.
\end{aligned}
\end{equation}
The second inequality follows by integrating
\[
 -\frac d{ds}(e^{-s^2/2}s^{679})
 =e^{-s^2/2}s^{680}(1-679/s^2)
 \ge\frac{345}{1024}e^{-s^2/2}s^{680}\quad(s\ge32).
\]

\smallskip\noindent\textbf{The computation of $\hat B$.}
On the ellipse over $[a,b]$, $|z-(a+b)/2|\le5h/4$,
$|\Im z|\le3h/4$, and $|\Re z|\ge\max(0,a-h/4)$.
Mehler's identity and the Taylor expansion of $r$ give
\[
\begin{aligned}
 \sum_{j=0}^{681}|\psi_j(z)|^2
 &\le(100/99)^{681}\sum_{j\ge0}(99/100)^j|\psi_j(z)|^2\\
 &=(100/99)^{681}(10000/199)^{1/2}
   e^{(99/199)(\Re z)^2+99(\Im z)^2},\\
 |\Im r(z)|&\le|\Im z|\sum_{k=1}^{15}
 \frac{|r^{(k)}((a+b)/2)|}{(k-1)!}|z-(a+b)/2|^{k-1}.
\end{aligned}
\]
Multiplying the square root of the first bound by
$|\phi(z)e^{{\mathrm{i}}sr(z)}|\le(2\pi)^{-1/2}
 e^{-(\Re z)^2/2+(\Im z)^2/2+32|\Im r(z)|}$ yields
$\|(\phi(z)\psi_j(z)e^{{\mathrm{i}}sr(z)})_{j=0}^{681}\|_2\le M$ for $0\le s\le32$,
with $M$ as in \cref{head:inner-M}. The bisection accepts 50 intervals,
giving 1600 positive nodes and
\begin{equation}\label{head:inner-total}
 2\sum_{[a,b]}8(b-a)M2^{-64}<3.288380170604\cdot10^{-12}.
\end{equation}
The sum is over the accepted intervals. By reflection and
\cref{head:gauss-error}, it bounds
\[
 \left\|B(s)-\left(\sum_x(v_j)_m
       \Re({\mathrm{i}}^{-m}e^{{\mathrm{i}}sr(x)})\right)_{m=0}^{681}\right\|_2
 \quad(0\le s\le32),
\]
where $j$ indexes $x$ as in step 2.
Put $\gamma_k:=k2^{-53}/(1-k2^{-53})$ for integers $0\le k<2^{53}$.
For an $n$-term dot product, $2n<2^{53}$, using at most $n$
multiplications and $n$ additions with the stated rounding and no overflow,
the error is at most $\gamma_{2n}$ times the sum of the absolute
values of its summands, apart from the additive $2^{-1074}$ rounding errors.
Rounding a multiplication and addition together does not increase this bound.
For every quadrature node and $0\le m\le681$, the input enclosures verify
\[
 \left|\widehat{\Re({\mathrm{i}}^{-m}e^{{\mathrm{i}}s_\nu r(x)})}
             -\Re({\mathrm{i}}^{-m}e^{{\mathrm{i}}s_\nu r(x)})\right|<2^{-48}.
\]
The same 512-bit evaluation gives
\[
 \sum_j\|v_j-\hat v_j\|_2<3.050177189723512\cdot10^{-16},
 \qquad \sum_j\|v_j\|_2<6.454332802795005.
\]
\begin{align}
 \|\hat B_{\nu\cdot}-B(s_\nu)\|_2
 \le{}&3.288380170604\cdot10^{-12}+\sum_j\|v_j-\hat v_j\|_2\notag\\
 &+(2^{-48}+\gamma_{3200}(1+2^{-48}))
       \sum_j(\|v_j\|_2+\|v_j-\hat v_j\|_2)+10^{-290}\notag\\
 &<5.604655228362\cdot10^{-12}.\label{head:B-error-detail}
\end{align}
\end{proof}

\begin{headclaim}\label{head:claim-arithmetic}
\begin{equation}\label{head:arithmetic-error}
 \sum_{\substack{a+b\le681\\a+b\ \mathrm{odd}}}
 \left|\hat C[a,b]-\frac4{\pi^2}
 \sum_{\nu,\mu=1}^{768}U_{a\nu}U_{a\mu}[z^b]
 \left(\sum_{j=0}^{681}\hat B_{\nu j}\hat B_{\mu j}z^j\right)^3\right|
 \le5.996940779766870\cdot10^{-8}.
\end{equation}
\end{headclaim}
\begin{proof}
For each $(\nu,\mu)$, Cauchy--Schwarz and \cref{head:B-error} give
\[
\begin{aligned}
 \sum_j|[z^j]g(z)|&\le\|\hat B_{\nu\cdot}\|_2\|\hat B_{\mu\cdot}\|_2
 \le(1+5.604656\cdot10^{-12})^2,\\
 \sum_j|[z^j]g(z)^3|&\le\left(\sum_j|[z^j]g(z)|\right)^3.
\end{aligned}
\]
Here and below the weight indices range over $0\le a\le681$,
$1\le\nu,\mu\le768$.
The input enclosures and stored numbers give, at $256$ bits,
\begin{equation}\label{head:input-totals}
\begin{aligned}
 \sum_{a,\nu}|U_{a\nu}-\hat U_{a\nu}|&<1.092870488954584\cdot10^{-15},&
 \sum_{a,\nu}|\hat U_{a\nu}|&<26.904610216896492,\\
 \sum_a\left(\sum_\nu|\hat U_{a\nu}|\right)^2&<61.515376731920361,&
 \left|\frac4{\pi^2}-{\widehat{\frac4{\pi^2}}}\right|&<1.586515951775777\cdot10^{-17},\\
 \sum_{a,\nu:\hat U_{a\nu}<10^{-18}}|\hat U_{a\nu}|&<2.655682807169783\cdot10^{-15},&
 {\widehat{\frac4{\pi^2}}}&=\frac{228155022448185}{562949953421312}.
\end{aligned}
\end{equation}
For conversion of the weights and $4/\pi^2$,
\begin{align}
 \sum_{a=0}^{681}\sum_{\nu,\mu=1}^{768}
 \left|\frac4{\pi^2}U_{a\nu}U_{a\mu}
       -{\widehat{\frac4{\pi^2}}}\hat U_{a\nu}\hat U_{a\mu}\right|
 &\le\frac4{\pi^2}
 \left(\sum_{a,\nu}|U_{a\nu}-\hat U_{a\nu}|\right)
 \left(2\sum_{a,\nu}|\hat U_{a\nu}|+
              \sum_{a,\nu}|U_{a\nu}-\hat U_{a\nu}|\right)\notag\\
 &\quad+\left|\frac4{\pi^2}-{\widehat{\frac4{\pi^2}}}\right|
       \sum_a\left(\sum_\nu|\hat U_{a\nu}|\right)^2\notag\\
 &<8.386939012631\cdot10^{-14}.\label{head:U-error}
\end{align}
Multiplication by
$(1+5.604656\cdot10^{-12})^6$ bounds the resulting coefficient error.
Omitting entries below $10^{-18}$ contributes at most
\[
\begin{aligned}
 &{\widehat{\frac4{\pi^2}}}(1+5.604656\cdot10^{-12})^6
 \sum_a\sum_{\substack{\nu,\mu\\
       \hat U_{a\nu}<10^{-18}\ \mathrm{or}\ \hat U_{a\mu}<10^{-18}}}
            |\hat U_{a\nu}\hat U_{a\mu}|\\
 &\quad\le{\widehat{\frac4{\pi^2}}}(1+5.604656\cdot10^{-12})^6
 \sum_a\left[2\left(\sum_\nu|\hat U_{a\nu}|\right)
 \left(\sum_{\nu:\hat U_{a\nu}<10^{-18}}|\hat U_{a\nu}|\right)
 +\left(\sum_{\nu:\hat U_{a\nu}<10^{-18}}|\hat U_{a\nu}|\right)^2\right]\\
 &\quad<5.791527838\cdot10^{-14}.
\end{aligned}
\]
In each transform in \cref{head:discrete-transform}, the eleven
stages $(a,b)\mapsto(a+wb,a-wb)$ satisfy
\[
 |a+wb|^2+|a-wb|^2=2(|a|^2+|b|^2)\quad(|w|=1).
\]
For $a,b\in\C$, let $\widehat{ab}$ be their product computed with four
real multiplications and two real additions. Apart from the additive
$2^{-1074}$ errors treated below,
\[
 |\widehat{ab}-ab|
 \le\sqrt2(2\cdot2^{-53}+2^{-106})|a||b|\le3\cdot2^{-53}|a||b|.
\]
The component errors in $w$ therefore give
\[
 |\widehat{\hat w b}-wb|
 \le\bigl(\sqrt2\,2^{-48}+3\cdot2^{-53}(1+\sqrt2\,2^{-48})\bigr)|b|
 \le50\cdot2^{-53}|b|.
\]
The two final additions give stage error at most
$(51+50\cdot2^{-53})2^{-53}\sqrt{2(|a|^2+|b|^2)}$,
where $51+50\cdot2^{-53}<100$.
For $0\le k\le11$, $(1+100\cdot2^{-53})^k-1\le\gamma_{100k}$.
Including rounded input products, the Euclidean error after $k$ stages is at most
\[
 2^{k/2}(\gamma_{100k}(1+2^{-53})+2^{-53})
       \sum_j|[z^j]g(z)|
 \le5.53175\cdot10^{-12}\sum_j|[z^j]g(z)|.
\]
In exact arithmetic, each intermediate entry is a sum of distinct input
coefficients with multipliers of modulus one, so its modulus is at most
$\sum_j|[z^j]g(z)|$.
For coordinatewise cubes of complex vectors,
\[
 \|v^3-w^3\|_2\le3\max(\|v\|_\infty,\|w\|_\infty)^2\|v-w\|_2.
\]
For $(zz)z$, the product bound gives rounding error at most
$((1+3\cdot2^{-53})^2-1)|z|^3\le\gamma_{20}|z|^3$.
For the computed forward-transform vector $v\in\C^{2048}$, rounding
these cubes contributes Euclidean error at most
$\gamma_{20}\sqrt{2048}\|v\|_\infty^3$.
The Euclidean operator norm
of the inverse transform is $1/\sqrt{2048}$, and
$\|v\|_1\le\sqrt{2048}\|v\|_2$. The sum of absolute coefficient
errors is therefore at most
\begin{align}
 \left(\sum_j|[z^j]g(z)|\right)^3\bigl(&3(1+5.53175\cdot10^{-12})^2
                                     (5.53175\cdot10^{-12})\notag\\
 &+\sqrt{2048}(\gamma_{20}+\gamma_{1100}(1+\gamma_{20}))
                  (1+5.53175\cdot10^{-12})^3\bigr)
 <2.222247\cdot10^{-11}\label{head:fft-derived}
\end{align}
and, including the additive rounding errors, at most
\begin{equation}\label{head:fft-allocated}
 10000\cdot2048\cdot2^{-53}\left(\sum_j|[z^j]g(z)|\right)^3+10^{-290}
 <2.273736754509\cdot10^{-9}.
\end{equation}
For each $(\nu,\mu)$, the input products, two transforms, cubes, and
inverse normalization use at most
\[
 682+22\cdot1024\cdot10+2048(12+2)<2^{22}
\]
real operations. The bounds for $0\le k\le11$ give forward values of modulus at most
\[
 (1+5.53175\cdot10^{-12})(1+5.604656\cdot10^{-12})^2
 +4\cdot2^{22}\cdot2^{11}\cdot2^{-1074}<\frac32.
\]
An error in a real product passes through at most two remaining real
additions and is copied to the two outputs, with Euclidean amplification
$\sqrt2(1+2^{-53})^2<4$. With local rounding errors separated, each
remaining linear stage map has operator norm $<2$.
The forward and cube bounds give inverse inputs of
modulus $<8$, and all inverse intermediate magnitudes are bounded by
\[
 8\cdot2048(1+100\cdot2^{-53})^{11}
 +4\cdot2^{22}\cdot2^{11}\cdot2^{-1074}<2^{15}<2^{16}.
\]
Thus the transform and cube operations do not overflow.
With the relative rounding factors fixed, each computed product is
bilinear over $\R$, and the product bound gives input-error amplification at most
\[
 3(1+3\cdot2^{-53})^2\left(\frac32\right)^2<12
 \qquad(|z|<3/2)
\]
for the two products forming $z^3$.
Errors within these multiplications have amplification at most $4\cdot12$.
Across at most $22$ transform stages and the final
division by $2048$, a local error contributes at most
\[
 4\cdot2^{22}\cdot12\cdot\frac{\sqrt{2048}}{2048}
 <2^{24}<2^{40}
\]
times its magnitude to the sum of absolute coefficient errors.
Errors in the final division have amplification one. The total additive
error is $<2^{22}2^{-1074}2^{40}<10^{-300}$.
All $B$ dot products and weighted sums together use fewer than
$2\cdot1600\cdot768\cdot682+682^2(2\cdot768^2+8\cdot768)<2^{40}$
operations. Their intermediate magnitudes are bounded, respectively, by
$2\sum_j(\|v_j\|_2+\|v_j-\hat v_j\|_2)<16$ and
$2(1+\sum_{a,\nu}|\hat U_{a\nu}|)^2<2^{12}$, so they do not overflow.
A local error reaches its dot-product or weighted-sum output through at
most two weight multiplications, the final factor, and $1185792$ rounded
operations. Its amplification is at most
\[
 (1+2^{-53})^{1185792}\max\!\left(1,{\widehat{\frac4{\pi^2}}}\right)
 \left(1+\sum_{a,\nu}|\hat U_{a\nu}|\right)^2<2^{12}.
\]
Their additive errors therefore sum to less than
$2^{40}2^{12}2^{-1074}<10^{-300}<10^{-290}$.
Taking real parts cannot increase error, since $|\Re z|\le|z|$.
Using \cref{head:U-error} with the larger coefficient bound from
\cref{head:fft-allocated}, the conversion error is at most
\[
 (8.386939012631\cdot10^{-14})
 ((1+5.604656\cdot10^{-12})^6+2.273736754509\cdot10^{-9})
 <8.386939032\cdot10^{-14}.
\]
After multiplication by the weights, the error in the computed cube
coefficients and the error in the final sums are respectively bounded by
\[
\begin{aligned}
 &{\widehat{\frac4{\pi^2}}}\sum_a\left(\sum_\nu|\hat U_{a\nu}|\right)^2
 (2.273736754509\cdot10^{-9})<5.668708385\cdot10^{-8},\\
 &\gamma_{1185792}{\widehat{\frac4{\pi^2}}}
 \sum_a\left(\sum_\nu|\hat U_{a\nu}|\right)^2
 ((1+5.604656\cdot10^{-12})^6+2.273736754509\cdot10^{-9})
 +10^{-290}\\
 &\qquad<3.282182163\cdot10^{-9},
\end{aligned}
\]
The inner sums, outer sum, and final multiplication use at most
$768(2\cdot768)+2\cdot768+1\le2\cdot768^2+8\cdot768=1185792$
real operations per coefficient. Adding the conversion, omission,
transform, and summation bounds gives

\[
\begin{aligned}
 &8.386939032\cdot10^{-14}+5.791527838\cdot10^{-14}
 +5.668708385\cdot10^{-8}+3.282182163\cdot10^{-9}\\
 &\qquad=5.996940779766870\cdot10^{-8}.
\end{aligned}
\]
The exact entries are nonnegative by \cref{head:nonnegative-squares};
$|\max(x,0)-y|\le|x-y|$ for $y\ge0$ shows that replacing negative
entries by zero cannot increase error.
\end{proof}
\endgroup

\subsection{\texorpdfstring{{Bounds above degree $681$ for $H_{14},H_{17},\ldots,H_{31}$}}{Bounds above degree 681 for H14 and H17 through H31}}
\label{other:section}

{
\begingroup
\theoremstyle{plain}
\newtheorem{otherclaim}[headclaim]{Claim}
\begin{lemma}\label{other:lemma}
\begin{equation}\label{eq:upper-certified-nonadditive}
 \lambda_{14}\sum_{m>681}|[t^m]H_{14}(t)|
     +\sum_{i=17}^{31}\lambda_i\sum_{m>681}|[t^m]H_i(t)|
 <0.000003878150168295.
\end{equation}
\end{lemma}
\begin{proof}
Claims~\ref{other:claim-staircase} and~\ref{other:claim-hyperplane} give
\[
 \begin{aligned}
 &\lambda_{14}\sum_{m>681}|[t^m]H_{14}(t)|
 +\sum_{i=17}^{31}\lambda_i\sum_{m>681}|[t^m]H_i(t)|\\
 &\quad<(1.254728050780+2.623422117515)\cdot10^{-6}
 =0.000003878150168295.
 \end{aligned}
\]
\end{proof}

\begin{otherclaim}\label{other:claim-staircase}
\begin{equation}\label{other:staircase-numeric}
 \lambda_{14}\sum_{m>681}|[t^m]H_{14}(t)|
 <1.254728050780\cdot10^{-6}<1.254729\cdot10^{-6}.
\end{equation}
\end{otherclaim}
\begin{proof}
The formulas for $J_j$ in \cref{head:hermite-integrals} follow from
\[
 \frac d{dx}\operatorname{erfc}(x/\sqrt2)=-2\phi(x),\qquad
 (\phi\He_{j-1})'=-\phi\He_j,\qquad
 \lim_{x\to\pm\infty}\phi(x)\He_{j-1}(x)=0\quad(j\ge1).
\]
Define $u_{ij}:=\E[F_*(X_1,X_2)\He_i(X_1)\He_j(X_2)]$.
By \cref{root:odd-hermite}, $u_{ij}=0$ for $i+j$ even. On the strip
$k/256\le y<(k+1)/256$, the definition of $F_*$ gives
\[
 \begin{aligned}
 \int_{\mathbb R}F_*(x,y)\He_i(x)\phi(x)\,dx
 &=\mathbf1_{\{i=0\}}+2\sum_a(-1)^aJ_i(r_{k,a}),\\
 \int_{\mathbb R}\sign(x)\He_i(x)\phi(x)\,dx
 &=2\mathbf1_{\{i\text{ odd}\}}J_i(0).
 \end{aligned}
\]
For $i+j$ odd, symmetry and $F_*(x,y)=-\sign(x)$ for $y\ge8$ yield
\begin{align}
 u_{ij}
 &=2\int_0^\infty\int_{\mathbb R}
       F_*(x,y)\He_i(x)\He_j(y)\phi(x)\phi(y)\,dx\,dy\notag\\
 &=2\sum_{k=0}^{2047}
  \left(\mathbf1_{\{i=0\}}+2\sum_a(-1)^aJ_i(r_{k,a})\right)
       \left[J_j(k/256)-J_j((k+1)/256)\right]\notag\\
 &\quad-4\mathbf1_{\{i\text{ odd}\}}J_i(0)J_j(8).
 \label{other:staircase-moments}
\end{align}
By \cref{eq:upper-staircase-kernel} and Parseval,
\begin{equation}\label{other:staircase-kernel}
 \sum_{m>681}|[t^m]H_{14}(t)|
 \le\sum_{3i+5j>681}\frac{u_{ij}^2}{i!j!}
 =1-\sum_{3i+5j\le681}\frac{u_{ij}^2}{i!j!}.
\end{equation}
Evaluating \cref{other:staircase-moments} with $384$-bit interval
arithmetic gives
\[
 \begin{gathered}
 \#\{(i,j)\in\mathbb Z_{\ge0}^2:i+j\text{ odd},\ 3i+5j\le681\}=7843,\\
 \sum_{3i+5j\le681}\frac{u_{ij}^2}{i!j!}>0.9648417014661795,\qquad
 \lambda_{14}(1-0.9648417014661795)<1.254728050780\cdot10^{-6}.
 \end{gathered}
\]
\end{proof}

\begin{otherclaim}\label{other:claim-hyperplane}
\begin{equation}\label{other:hyperplane-numeric}
 \sum_{i=17}^{31}\lambda_i\sum_{m>681}|[t^m]H_i(t)|
 <2.623422117515\cdot10^{-6}<2.623423\cdot10^{-6}.
\end{equation}
\end{otherclaim}
\begin{proof}
For $17\le i\le31$, \cref{root:linear},
\cref{root:coefficient-majorant}, and
\begin{equation}\label{other:arcsin-series}
 \arcsin z=\sum_{k\ge0}\frac{\binom{2k}{k}}{4^k(2k+1)}z^{2k+1}
 \qquad(|z|\le1)
\end{equation}
give
\[
 |[t^m]H_i(t)|\le[t^m]\frac2\pi\arcsin P_{i,+}(t),\qquad
 \sum_{m\ge0}[t^m]\frac2\pi\arcsin P_{i,+}(t)
 =\frac2\pi\arcsin P_{i,+}(1)\le1.
\]
Consequently,
\begin{equation}\label{other:positive-arcsine-tail}
 \sum_{m>681}|[t^m]H_i(t)|
 \le1-\sum_{m=1}^{681}[t^m]\frac2\pi\arcsin P_{i,+}(t).
\end{equation}
For $19\le i\le31$, evaluate at $256$ bits the finite sum
\[
 \begin{aligned}
 \sum_{m=1}^{681}[t^m]\frac2\pi\arcsin P_{i,+}(t)
 &=\frac2\pi\sum_{k=0}^{340}\frac{\binom{2k}{k}}{4^k(2k+1)}
            \sum_{m=1}^{681}[t^m]P_{i,+}(t)^{2k+1},\\
 [t^m]P_{i,+}^{2k+1}&=0\quad(m<2k+1).
 \end{aligned}
\]
For $i=17,18$, set
\[
 \alpha:=[t]P_i>0,\qquad\beta:=-[t^3]P_i>0,\qquad
 (R,\eta):=\begin{cases}(6/5,19/20),&i=17,\\
                         (21/20,49/50),&i=18.
             \end{cases}
\]
For $z=Re^{{\mathrm{i}}\theta}$ and $x=\cos\theta$,
\[
 \Re P_i(z)=(\alpha R+3\beta R^3)x-4\beta R^3x^3.
\]
Its squared endpoint and interior critical values satisfy
\[
 (\alpha R-\beta R^3)^2<\eta^2,\qquad
 \frac{(\alpha R+3\beta R^3)^3}{27\beta R^3}<\eta^2
 \quad\text{if }\frac{\alpha R+3\beta R^3}{12\beta R^3}<1.
\]
Thus $\sup_{|z|\le R}|\Re P_i(z)|<\eta$ by the maximum principle.
On $|\Re w|<1$, take the analytic branch of $\arcsin$ with
$\arcsin(0)=0$, $\arcsin'(0)=1$. For $w=x+{\mathrm{i}}y$, $|x|\le\eta$,
\[
 \begin{aligned}
 |1-w^2|^2&=(1-x^2)^2+2(1+x^2)y^2+y^4\ge(1-\eta^2)^2,\\
 |\arcsin w|&=\left|\int_0^1\frac{w}{\sqrt{1-t^2w^2}}\,dt\right|
                 \le\frac{|w|}{\sqrt{1-\eta^2}}.
 \end{aligned}
\]
Cauchy's estimate and oddness give
\[
 |[t^m]H_i(t)|\le
 \frac{2(\alpha R+\beta R^3)}{\pi\sqrt{1-\eta^2}}R^{-m},
 \qquad[t^{2m}]H_i(t)=0,
\]
and hence
\begin{equation}\label{other:cubic-tail}
 \sum_{m>681}|[t^m]H_i(t)|
 \le\frac{2(\alpha R+\beta R^3)}{\pi\sqrt{1-\eta^2}}
                         \frac{R^{-683}}{1-R^{-2}}
 <\begin{cases}6.982\cdot10^{-54},&i=17,\\
               1.223\cdot10^{-13},&i=18.
   \end{cases}
\end{equation}
{The coefficients satisfy}
\cref{root:normalization}. Summing the $256$-bit interval bounds
\cref{other:positive-arcsine-tail} and \cref{other:cubic-tail}
with weights $\lambda_i$ gives \cref{other:hyperplane-numeric}.
\end{proof}
\endgroup
}

\subsection{\texorpdfstring{{Bounds above degree $681$ for $H_1,\ldots,H_{13}$}}{Bounds above degree 681 for H1 through H13}}
\label{other:rare}

{
\begin{lemma}\label{other:rare-lemma}
\begin{equation}\label{eq:upper-certified-rare}
 \sum_{i=1}^{13}\lambda_i\sum_{m>681}|[t^m]H_i(t)|
                                  <0.000034097111321496.
\end{equation}
\end{lemma}
\begin{proof}
For $1\le i\le13$, Parseval and \cref{root:coefficient-majorant} give
\[
 \sum_{m\ge0}[t^m]\sum_{a,b\ge0}C_{\rm tr}[a,b]P_{i,+}(t)^at^b
 =\sum_{a,b\ge0}C_{\rm tr}[a,b]P_{i,+}(1)^a
 \le\E[F_{\rm tr}(W,X)^2]\le1.
\]
Therefore \cref{root:additive-kernel}, \cref{root:coefficient-majorant},
\cref{head:cutoff}, and \cref{head:joint-plane} imply
\begin{equation}\label{other:captured-mass}
\begin{aligned}
 \sum_{m>681}|[t^m]H_i(t)|
 &\le\sum_{m>681}[t^m]\sum_{a,b\ge0}
                 C_{\rm tr}[a,b]P_{i,+}(t)^at^b
        +\sum_{a,b\ge0}|C[a,b]-C_{\rm tr}[a,b]|\\
 &\le1-\sum_{m=0}^{681}[t^m]\sum_{a+b\le681}
                 \hat C[a,b]P_{i,+}(t)^at^b
        +\sum_{a+b\le681}|\hat C[a,b]-C_{\rm tr}[a,b]|\\
 &\qquad+\sum_{a,b\ge0}|C[a,b]-C_{\rm tr}[a,b]|\\
 &\le1-\sum_{m=0}^{681}[t^m]\sum_{a+b\le681}
                 \hat C[a,b]P_{i,+}(t)^at^b\\
 &\qquad+6.078879564946870\cdot10^{-8}
        +2\sqrt{3\operatorname{erfc}\left(\frac{73}{10\sqrt2}\right)}.
\end{aligned}
\end{equation}
Evaluating the finite sum at $256$ bits using the exact rational
$\hat C[a,b]$ from Claim~\ref{head:claim-numerical} gives the following
upper bounds for $\lambda_i\sum_{m>681}|[t^m]H_i(t)|$:
\[
\begin{array}{c|c@{\qquad}c|c}
 i&\text{upper bound}&i&\text{upper bound}\\ \hline
 1&2.364541\cdot10^{-7}&8&6.366458\cdot10^{-6}\\
 2&1.243556\cdot10^{-7}&9&3.079728\cdot10^{-7}\\
 3&3.261885\cdot10^{-7}&10&3.796353\cdot10^{-6}\\
 4&2.277905\cdot10^{-6}&11&2.153852\cdot10^{-7}\\
 5&3.923603\cdot10^{-6}&12&1.933337\cdot10^{-6}\\
 6&1.454752\cdot10^{-6}&13&5.492683\cdot10^{-7}\\
 7&1.258509\cdot10^{-5}&&
\end{array}
\]
Summing the interval bounds before rounding to the displayed decimals
gives \cref{eq:upper-certified-rare}.
\end{proof}
}

\subsection{\texorpdfstring{{Bounds above degree $681$ for $\sum_{i\in\{0,15,16\}}\lambda_iH_i$}}{Bounds above degree 681 for the weighted sum with indices 0, 15, and 16}}
\label{heavy:section}
{
\begingroup
\theoremstyle{plain}
\newtheorem{heavyclaim}[headclaim]{Claim}
\begin{lemma}\label{heavy:lemma}
\begin{equation}\label{heavy:final-number}
 \sum_{m>681}\left|[t^m]\sum_{i\in\{0,15,16\}}\lambda_iH_i(t)\right|
 <0.000007875733380418.
\end{equation}
\end{lemma}
\begin{proof}
Define
\begin{equation}\label{heavy:cap-definition}
 \begin{aligned}
 \chi(u)&:=1-462u^6+1980u^7-3465u^8+3080u^9-1386u^{10}+252u^{11},\\
 \widetilde r(x)&:=
 \begin{cases}
 r(x),&|x|\le27/4,\\
 {\sign}(x)\displaystyle\left(r(27/4)+
 \int_0^{\min\{|x|-27/4,1/4\}}r'(27/4+u)\chi(4u)\,du\right),&|x|>27/4.
 \end{cases}
 \end{aligned}
\end{equation}
For $i\in\{0,15,16\}$ and real $t\in[-1,1]$, let
\[
 \widetilde F(w,x):={\sign}\!\left(w+\sum_{j=1}^3\widetilde r(x_j)\right),
 \qquad \widetilde H_i(t):=\E[\widetilde F(W,X)\widetilde F(W',X')],
\]
where $(W,W'),(X_j,X'_j)$ are independent centered jointly Gaussian
pairs with unit marginal variances and
$\E[WW']=P_i(t)$, $\E[X_jX'_j]=-t$, $1\le j\le3$.
Set $\mathcal D:=t\,d/dt$ and
$\mathbb T:=\{t\in\C:|t|=1\}$, parametrized by $t=e^{{\mathrm{i}}\theta}$ with measure $d\theta/(2\pi)$.
\begin{heavyclaim}\label{heavy:claim-perturbation}
The function $\widetilde r$ is bounded and $C^6$, with
$\widetilde r^{(j)}(x)=0$ for $|x|>7$, $1\le j\le6$, and
\begin{equation}\label{heavy:weighted-perturbation}
 \sum_{m\ge0}\left|[t^m]\sum_{i\in\{0,15,16\}}
 \lambda_i(H_i-\widetilde H_i)\right|
 \le4\sqrt{{\Pr}(F\ne\widetilde F)}
 \sum_{i\in\{0,15,16\}}\lambda_i
 <8.769183\cdot10^{-7}.
\end{equation}
\end{heavyclaim}
\begin{heavyclaim}\label{heavy:claim-circle}
The Taylor series of each $\widetilde H_i$ converges absolutely and uniformly on $|t|\le1$,
its restriction to $\mathbb T$ is $C^4$, and
$d^4\widetilde H_i(e^{{\mathrm{i}}\theta})/d\theta^4
=\mathcal D^4\widetilde H_i(e^{{\mathrm{i}}\theta})$. Moreover,
\begin{equation}\label{heavy:combined-numerical}
 \frac{\left\|\mathcal D^4\sum_{i\in\{0,15,16\}}
 \lambda_i\widetilde H_i\right\|_{{L_2}(\mathbb T)}}{\sqrt{14\cdot681^7}}
 +4\sqrt{{\Pr}(F\ne\widetilde F)}
 \sum_{i\in\{0,15,16\}}\lambda_i
 <0.000007875733380418.
\end{equation}
\end{heavyclaim}
By Claim~\ref{heavy:claim-circle}, four integrations by parts and
Parseval give
\[
 \sum_{m\ge0}m^8\left|[t^m]\sum_{i\in\{0,15,16\}}
 \lambda_i\widetilde H_i\right|^2
 =\left\|\mathcal D^4\sum_{i\in\{0,15,16\}}
 \lambda_i\widetilde H_i\right\|_{{L_2}(\mathbb T)}^2.
\]
Oddness of $\widetilde r$ and \cref{root:odd-hermite} imply that
the even coefficients vanish. Thus Cauchy--Schwarz and
\[
 \sum_{\substack{m>681\\m\text{ odd}}}m^{-8}
 \le\frac12\int_{681}^\infty x^{-8}\,dx=(14\cdot681^7)^{-1}
\]
give
\[
 \sum_{m>681}\left|[t^m]\sum_{i\in\{0,15,16\}}
 \lambda_i\widetilde H_i\right|
 \le\frac{\left\|\mathcal D^4\sum_{i\in\{0,15,16\}}
 \lambda_i\widetilde H_i\right\|_{{L_2}(\mathbb T)}}{\sqrt{14\cdot681^7}}.
\]
Adding \cref{heavy:weighted-perturbation} and using
\cref{heavy:combined-numerical} proves the lemma.
\end{proof}

\begin{proof}[Proof of Claim~\ref{heavy:claim-perturbation}]
We use the inequalities
\begin{equation}\label{heavy:profile-checks}
 r(27/4)>7,\qquad r^{(j)}(27/4)>0\quad(1\le j\le15),\qquad
 |r(x)|<33/20\quad(|x|\le7/2).
\end{equation}
For $s\ge0$, \cref{heavy:cap-definition} and
$0\le\chi(u)\le1$ for $0\le u\le1$ give
\[
 \begin{aligned}
 r'(27/4+s)&=\sum_{j=1}^{15}\frac{r^{(j)}(27/4)}{(j-1)!}s^{j-1}>0,\\
 \min\{r(27/4+s),\widetilde r(27/4+s)\}&\ge r(27/4)>7.
 \end{aligned}
\]
Therefore
\[
 \begin{aligned}
 \max_j|X_j|\le27/4&\implies F=\widetilde F,\\
 |X_i|>27/4,\quad\max_{j\ne i}|X_j|\le7/2,\quad |W|\le37/10
 &\implies F=\widetilde F={\sign}(X_i),
 \end{aligned}
\]
because $7=2(33/20)+37/10$. Hence
\[
 \{F\ne\widetilde F\}\subseteq
 \bigcup_{i\ne j}\{|X_i|>27/4,\ |X_j|>7/2\}
 \cup\bigcup_i\{|X_i|>27/4,\ |W|>37/10\},
 \qquad 1\le i,j\le3.
\]
Independence and the union bound give
\begin{equation}\label{heavy:disagreement-bound}
 {\Pr}(F\ne\widetilde F)
 \le6\operatorname{erfc}\!\left(\frac{27}{4\sqrt2}\right)
 \operatorname{erfc}\!\left(\frac7{2\sqrt2}\right)
 +3\operatorname{erfc}\!\left(\frac{27}{4\sqrt2}\right)
 \operatorname{erfc}\!\left(\frac{37}{10\sqrt2}\right).
\end{equation}
Parseval and \cref{head:square-difference} yield
\[
 \begin{aligned}
 &\sum_{a\ge0,\,\beta\in\{0,1,\ldots\}^3}
 \left|\bigl(\E[F\psi_a(W)\psi_\beta(X)]\bigr)^2
 -\bigl(\E[\widetilde F\psi_a(W)\psi_\beta(X)]\bigr)^2\right|\\
 &\qquad\le\|F-\widetilde F\|_2\|F+\widetilde F\|_2
 \le4\sqrt{{\Pr}(F\ne\widetilde F)}.
 \end{aligned}
\]
Together with \cref{root:coefficient-majorant} and \cref{eq:correlation},
this proves \cref{heavy:weighted-perturbation}; the numerical evaluation
of \cref{heavy:disagreement-bound} uses interval rounding at $256$ bits.

The conclusions
\[
 0\le\chi(u)\le1\quad(0\le u\le1),\qquad
 \widetilde r\in C^6(\R)
\]
follow from
\[
 \chi(0)=1,\quad\chi(1)=0,\quad\chi'(u)=-2772u^5(1-u)^5,\quad
 \chi^{(j)}(0)=\chi^{(j)}(1)=0\quad(1\le j\le5).
\]
For $c=-7/2+(2k+1)/32$, $0\le k<112$, the Taylor bounds
\[
 |x-c|\le1/32\implies
 |r(x)|\le|r(c)|+\sum_{j=1}^{15}\frac{|r^{(j)}(c)|}{j!32^j}<33/20
\]
and the values $r(27/4),r^{(j)}(27/4)$, $1\le j\le15$, verify
\cref{heavy:profile-checks} with interval rounding at $256$ bits.
The definition of $\widetilde r$ gives boundedness and vanishing
derivatives outside $[-7,7]$.
\end{proof}

\begin{proof}[Proof of Claim~\ref{heavy:claim-circle}]
To bound $|\mathcal D^4\widetilde H_i(t)|$ when $|P_i(t)|<1$, define
\[
 \widetilde C[a,b]:=\sum_{|\beta|=b}
 \bigl(\E[\widetilde F(W,X)\psi_a(W)\psi_\beta(X)]\bigr)^2
 \qquad(a,b\ge0).
\]
For $\alpha\in\{0,1,\ldots\}^3$, write
$\alpha!:=\prod_i\alpha_i!$ and
$\partial^\alpha:=\partial_{x_1}^{\alpha_1}\partial_{x_2}^{\alpha_2}
\partial_{x_3}^{\alpha_3}$; $[z^\alpha]$ denotes the coefficient of
$\prod_i z_i^{\alpha_i}$. Derivatives at $X$ are evaluated at $x=X$. Set
\begin{equation}\label{heavy:chain-rule}
 B_{\alpha,l}(x):=\frac{\alpha!}{l!}[z^\alpha]
 \left(\sum_{i=1}^3\sum_{h=1}^4
 \frac{\widetilde r^{(h)}(x_i)}{h!}z_i^h\right)^l
 \qquad(1\le l\le|\alpha|\le4).
\end{equation}
For $1\le l\le j\le4$, set
\begin{equation}\label{heavy:T-definition}
 T_{j,l}:=\sum_{|\alpha|=j}\frac{j!}{\alpha!}\E[B_{\alpha,l}(X)^2],
 \qquad
 \begin{array}{c|rrrr}
 j\backslash l&1&2&3&4\\\hline
 1&0.099&&&\\
 2&0.221&0.335&&\\
 3&0.36&3.38&4.88&\\
 4&1.11&24.3&124&110
 \end{array}
\end{equation}
Use $(n)_0:=1$, $(n)_j:=n(n-1)\cdots(n-j+1)$ for $1\le j\le n$,
$(n)_j:=0$ for $j>n\ge0$, and $(-1)!!:=1$.
We set $0^0=1$.
\begin{heavyclaim}\label{heavy:claim-endpoint}
Each entry in \cref{heavy:T-definition} is a strict upper bound for $T_{j,l}$, and
\[
 \widetilde C[a,b]\ge0,\qquad\sum_{a,b}\widetilde C[a,b]=1,\qquad
 \widetilde H_i(t)=\sum_{a,b}\widetilde C[a,b]P_i(t)^a(-t)^b.
\]
This series defines the extension of $\widetilde H_i$ to complex $|t|\le1$.
For nonnegative integers $p,j$ and $0\le\rho<1$,
\begin{equation}\label{heavy:M-estimates}
 \begin{aligned}
 \sum_{a\ge p,b\ge0}(a)_p\rho^{a-p}\widetilde C[a,b]
 &\le\frac{2(2p-3)!!}{\pi(1-\rho^2)^{p-1/2}}&&(1\le p\le4),\\
 \sum_{a\ge p,b\ge j}(a)_p(b)_j\rho^{a-p}\widetilde C[a,b]
 &\le\left(\sum_{l=1}^j
 \sqrt{\frac{2T_{j,l}(2p+2l-3)!!}
 {\pi(1-\rho^2)^{p+l-1/2}}}\right)^2&&(j\ge1,\ p+j\le4).
 \end{aligned}
\end{equation}
For $i\in\{0,15,16\}$, $|t|=1$ and $|P_i(t)|<1$,
\begin{equation}\label{heavy:D4-endpoint}
 \begin{aligned}
 |\mathcal D^4\widetilde H_i(t)|\le{}&24
 \sum_{\substack{p,j\ge0\\1\le p+j\le4}}\frac1{p!j!}
 [z^4]\left(\sum_{h=1}^4\frac{|\mathcal D^hP_i(t)|}{h!}z^h\right)^p
 \left(\sum_{h=1}^4\frac{z^h}{h!}\right)^j\\
 &\quad\cdot\sum_{a\ge p,b\ge j}
 (a)_p(b)_j|P_i(t)|^{a-p}\widetilde C[a,b].
 \end{aligned}
\end{equation}
The series and its first four $\mathcal D$-derivatives converge uniformly
on $\{t:|t|\le1,\ |P_i(t)|\le\rho\}$ for every $0\le\rho<1$.
\end{heavyclaim}

For $z\in\mathbb C$ with $-1<\Re z<1$, the Gaussian density in coordinates
$u=(x+y)/\sqrt2$, $v=(x-y)/\sqrt2$ continues to
\begin{equation}\label{heavy:complex-density}
 \frac{\exp\bigl(-u^2/(2(1+z))-v^2/(2(1-z))\bigr)}
 {2\pi\sqrt{(1+z)(1-z)}},
\end{equation}
with the analytic square-root branch positive for real $-1<z<1$.
For $t=e^{{\mathrm{i}}\theta}$, $0<\theta<\pi$, and
$\Re((1\pm P_i(t))^{-1})>0$, multiply the four densities
\cref{heavy:complex-density} at correlations $P_i(t),-t,-t,-t$.
For $\sigma=\pm1$, under the probability law whose density is the
absolute value of this product divided by its integral over $\R^8$, set
\[
 \begin{aligned}
 Y_{0,\sigma}&:=\Re((1+\sigma P_i(t))^{-1})\frac{(W+\sigma W')^2}{2},\\
 Y_{1,\sigma}&:=\frac14\sum_{j=1}^3(X_j+\sigma X'_j)^2,\qquad
 Y:=(Y_{0,+},Y_{0,-},Y_{1,+},Y_{1,-}).
 \end{aligned}
\]
For $1\le k\le4$, define
\begin{equation}\label{heavy:ell}
 \begin{aligned}
 \ell_k(Y):=-\frac{k!}{2}[\epsilon^k]\sum_{\sigma=\pm1}\biggl(
 &\ln(1+\sigma P_i(te^\epsilon))+3\ln(1-\sigma te^\epsilon)\\
 &+\frac{Y_{0,\sigma}}
 {\Re((1+\sigma P_i(t))^{-1})(1+\sigma P_i(te^\epsilon))}
 +\frac{2Y_{1,\sigma}}{1-\sigma te^\epsilon}\biggr),
 \end{aligned}
\end{equation}
holding $Y$ and $\Re((1+\sigma P_i(t))^{-1})$ constant in $\epsilon$.
\begin{heavyclaim}\label{heavy:claim-arc}
For $i\in\{0,15,16\}$ and every compact interval $I\subset(0,\pi)$ with
$\inf_{\theta\in I}\Re((1\pm P_i(e^{{\mathrm{i}}\theta}))^{-1})>0$,
$\widetilde H_i$ continues analytically to a neighborhood of the arc.
For $\theta\in I$ and $t=e^{{\mathrm{i}}\theta}$,
\begin{equation}\label{heavy:score-bound}
 |\mathcal D^4\widetilde H_i(e^{{\mathrm{i}}\theta})|\le
 \frac{(2/\sqrt{2\sin\theta})^3
 \sqrt{\E|\ell_1^4+6\ell_1^2\ell_2+3\ell_2^2+4\ell_1\ell_3+\ell_4|^2}}
 {\sqrt{|1-P_i(t)^2|\Re((1+P_i(t))^{-1})\Re((1-P_i(t))^{-1})}}.
\end{equation}
The coordinates of $Y$ are independent, with
\[
 \E e^{sY_{0,\sigma}}=(1-2s)^{-1/2},\qquad
 \E e^{sY_{1,\sigma}}=(1-2s)^{-3/2}
 \quad(\sigma=\pm1,\ s<1/2).
\]
\end{heavyclaim}

Use the $6284$ closed intervals
\[
 I\in\left\{\left[\frac{k}{4000},\frac{k+1}{4000}\right]:0\le k\le6282\right\}
 \cup\left\{\left[\frac{6283}{4000},\frac\pi2\right]\right\}.
\]
The interval computations specified at the end of this subsection give, for every
$(i,I)$, an upper bound
\[
 \sup_{\theta\in I}|\mathcal D^4\widetilde H_i(e^{{\mathrm{i}}\theta})|
 \le U_{i,I}<\infty,
\]
with at least one of
Claims~\ref{heavy:claim-endpoint} and~\ref{heavy:claim-arc}
applicable throughout $I$. Before rounding to the displayed decimal
bounds, the interval computations give
\begin{equation}\label{heavy:panel-numerics}
 \begin{aligned}
 \left(\frac2\pi\sum_I|I|
 \left(\sum_{i\in\{0,15,16\}}\lambda_iU_{i,I}\right)^2\right)^{1/2}
 &<215825.836929,\\
 \frac{\left(\frac2\pi\sum_I|I|
 \left(\sum_{i\in\{0,15,16\}}\lambda_iU_{i,I}\right)^2\right)^{1/2}}
 {\sqrt{14\cdot681^7}}
 +4\sqrt{{\Pr}(F\ne\widetilde F)}\sum_{i\in\{0,15,16\}}\lambda_i
 &<0.000007875733380418.
 \end{aligned}
\end{equation}
Equation~\eqref{root:normalization} and Claim~\ref{heavy:claim-endpoint} give
\[
 \sum_{m\ge0}|[t^m]\widetilde H_i|
 \le\sum_{a,b}\widetilde C[a,b]
 \left(\sum_{m\ge1}|[t^m]P_i|\right)^a\le1,
\]
so convergence is uniform on the closed disk. {The polynomials}
$P_i$ specified by Table~\ref{root:recipe} satisfy
$[t]P_i>0$, $[t^3]P_i<0$, $[t]P_i-[t^3]P_i=1$, whence
\[
 |P_i(e^{{\mathrm{i}}\theta})|^2
 =1+4([t]P_i)([t^3]P_i)\cos^2\theta<1\quad(e^{{\mathrm{i}}\theta}\ne\pm {\mathrm{i}}),
 \qquad P_i({\mathrm{i}})={\mathrm{i}},\quad\Re((1\pm P_i({\mathrm{i}}))^{-1})=\tfrac12.
\]
Claim~\ref{heavy:claim-endpoint} gives continuous derivatives through
order four away from $\pm {\mathrm{i}}$; Claim~\ref{heavy:claim-arc} and conjugation
give analytic continuation near $\pm {\mathrm{i}}$. Since
\[
 |P_i(t)|\le([t]P_i-[t^3]P_i)|t|=|t|<1\qquad(|t|<1),
\]
the integral against the product Gaussian density is holomorphic on the
open unit disk. It agrees there with the Hermite series by
\cref{eq:correlation} on the real interval and the identity theorem;
the arc integral continues this same function.
Thus $\widetilde H_i(e^{{\mathrm{i}}\theta})$ is $C^4$ and
$d^4/d\theta^4=({\mathrm{i}}\mathcal D)^4=\mathcal D^4$.
Oddness and real coefficients give
$\widetilde H_i(-t)=-\widetilde H_i(t)$ and
$\widetilde H_i(\bar t)=\overline{\widetilde H_i(t)}$, hence
\begin{equation}\label{heavy:panel-sum}
 \begin{aligned}
 \left\|\mathcal D^4\sum_{i\in\{0,15,16\}}\lambda_i\widetilde H_i
 \right\|_{{L_2}(\mathbb T)}^2
 &=\frac2\pi\int_0^{\pi/2}
 \left|\sum_{i\in\{0,15,16\}}\lambda_i
 \mathcal D^4\widetilde H_i(e^{{\mathrm{i}}\theta})\right|^2d\theta\\
 &\le\frac2\pi\sum_I|I|
 \left(\sum_{i\in\{0,15,16\}}\lambda_iU_{i,I}\right)^2.
 \end{aligned}
\end{equation}
Thus \cref{heavy:panel-numerics} gives
\begin{equation}\label{heavy:circle-numerical}
 \left\|\mathcal D^4\sum_{i\in\{0,15,16\}}\lambda_i\widetilde H_i
 \right\|_{{L_2}(\mathbb T)}<215825.836929
\end{equation}
and \cref{heavy:combined-numerical}.
\end{proof}

\begin{proof}[Proof of Claim~\ref{heavy:claim-endpoint}]
\cref{heavy:M-estimates} implies \cref{heavy:D4-endpoint} by
the triangle inequality and the expansions
\[
 \begin{aligned}
 \mathcal D^4(P_i(t)^a(-t)^b)
 &=24[z^4]P_i(te^z)^a(-t)^be^{bz},\\
 P_i(te^z)^a
 &=\sum_{p=0}^a\frac{(a)_p}{p!}P_i(t)^{a-p}(P_i(te^z)-P_i(t))^p,\\
 P_i(te^z)-P_i(t)&=\sum_{h=1}^4\frac{\mathcal D^hP_i(t)}{h!}z^h+O(z^5),\\
 e^{bz}&=\sum_{j=0}^b\frac{(b)_j}{j!}(e^z-1)^j,\qquad
 \mathcal D^hP_i(t)=([t]P_i)t+3^h([t^3]P_i)t^3.
 \end{aligned}
\]
For $|z|\le1$ and $1\le h\le4$,
\[
 |\mathcal D^hP_i(z)|\le[t]P_i-3^h[t^3]P_i.
\]
The expansions and \cref{heavy:M-estimates} therefore give, for
$0\le j\le4$ and $0\le\rho<1$,
\[
 \sum_{a,b}\widetilde C[a,b]
 \sup_{\substack{|z|\le1\\|P_i(z)|\le\rho}}
 \left|\left.\mathcal D^j(P_i(t)^a(-t)^b)\right|_{t=z}\right|<\infty,
\]
proving uniform convergence, in particular near $\pm1$ since
$|P_i(\pm1)|<1$.

To prove \cref{heavy:M-estimates}, conditioning on $X$ gives
\[
 \widetilde C[a,b]=\sum_{|\beta|=b}
 \left(\E\left[q_a\!\left(\sum_i\widetilde r(X_i)\right)
 \psi_\beta(X)\right]\right)^2,
\]
where $q_a$ is defined in \cref{head:smoothed-sign}.
Parseval and \cref{eq:correlation} give $\sum_{a,b}\widetilde C[a,b]=1$
and the expansion of $\widetilde H_i$ in Claim~\ref{heavy:claim-endpoint}.
For nonnegative integers $p,l$ with $p+l\ge1$, $s\in\mathbb R$, and
$0\le\rho<1$, we use
\begin{equation}\label{heavy:gram-bound}
 \sum_{a\ge p}(a)_p\rho^{a-p}|q_a^{(l)}(s)|^2
 \le\frac{2(2p+2l-3)!!}{\pi(1-\rho^2)^{p+l-1/2}}.
\end{equation}
For $j=0$, Parseval and \cref{heavy:gram-bound} with $l=0$ give
\[
 \begin{aligned}
 \sum_{a\ge p,b\ge0}(a)_p\rho^{a-p}\widetilde C[a,b]
 &=\E\!\left[\sum_{a\ge p}(a)_p\rho^{a-p}
 \left|q_a\!\left(\sum_i\widetilde r(X_i)\right)\right|^2\right]\\
 &\le\frac{2(2p-3)!!}{\pi(1-\rho^2)^{p-1/2}}
 \qquad(1\le p\le4).
 \end{aligned}
\]
For $a\ge0$, $|\alpha|=j\in\{1,2,3,4\}$, and
$\beta\in\{0,1,\ldots\}^3$, Taylor expansion and Gaussian integration
by parts using $(\phi\He_n)'=-\phi\He_{n+1}$ give
\[
 \begin{aligned}
 \partial^\alpha q_a\!\left(\sum_i\widetilde r(x_i)\right)
 &=\sum_{l=1}^{|\alpha|}B_{\alpha,l}(x)
 q_a^{(l)}\!\left(\sum_i\widetilde r(x_i)\right),\\
 \E\!\left[\partial^\alpha q_a\!\left(\sum_i\widetilde r(X_i)\right)
 \psi_\beta(X)\right]
 &=\sqrt{\frac{(\beta+\alpha)!}{\beta!}}\,
 \E\!\left[q_a\!\left(\sum_i\widetilde r(X_i)\right)
 \psi_{\beta+\alpha}(X)\right],\\
 \sum_{|\alpha|=j}\frac{j!}{\alpha!}
 \left\|\partial^\alpha q_a\!\left(\sum_i\widetilde r(X_i)\right)\right\|_2^2
 &=\sum_b(b)_j\widetilde C[a,b],
 \end{aligned}
\]
The last equality follows from Parseval and
$\sum_{|\alpha|=j}(j!/\alpha!)\prod_i(\beta_i)_{\alpha_i}=(|\beta|)_j$.
Boundedness of $\widetilde r,\ldots,\widetilde r^{(4)}$ and
\cref{head:smoothed-sign} justify integration by parts. The triangle inequality gives
\[
 \left(\sum_{a\ge p,b\ge j}(a)_p(b)_j\rho^{a-p}\widetilde C[a,b]\right)^{1/2}
 \le\sum_{l=1}^j\left(T_{j,l}\sup_{s\in\mathbb R}
 \sum_{a\ge p}(a)_p\rho^{a-p}|q_a^{(l)}(s)|^2\right)^{1/2}.
\]
Apply \cref{heavy:gram-bound}; finite sums and monotone convergence
justify summation over $a$.

For \cref{heavy:gram-bound}, with $u,v\in\mathbb R$ and $n=p+l-1$, use
\begin{equation}\label{heavy:gram-identity}
 \sum_{a\ge p}(a)_p\rho^{a-p}q_a^{(l)}(u)q_a^{(l)}(v)
 =4\partial_u^n\partial_v^n
 \frac{e^{-(u^2-2\rho uv+v^2)/(2(1-\rho^2))}}
 {2\pi\sqrt{1-\rho^2}}.
\end{equation}
Fourier inversion, Cauchy--Schwarz, and the Gaussian moment formula give
\[
 \begin{aligned}
 4\left|\partial_u^n\partial_v^n
 \frac{e^{-(u^2-2\rho uv+v^2)/(2(1-\rho^2))}}
 {2\pi\sqrt{1-\rho^2}}\right|
 &\le\frac1{\pi^2}\int_{\R^2}|\xi\eta|^n
 e^{-(\xi^2+2\rho\xi\eta+\eta^2)/2}\,d\xi\,d\eta\\
 &\le\frac1{\pi^2}\int_{\R^2}\xi^{2n}
 e^{-(\xi^2+2\rho\xi\eta+\eta^2)/2}\,d\xi\,d\eta
 =\frac{2(2n-1)!!}{\pi(1-\rho^2)^{n+1/2}}.
 \end{aligned}
\]
Setting $u=v=s$ proves \cref{heavy:gram-bound}. To prove
\cref{heavy:gram-identity}, use \cref{head:smoothed-sign},
$(\phi\He_j)'=-\phi\He_{j+1}$, and Mehler's identity:
\[
 \begin{aligned}
 \sum_{a\ge p}(a)_p\rho^{a-p}q_a^{(l)}(u)q_a^{(l)}(v)
 &=4\phi(u)\phi(v)\sum_{k\ge0}\frac{\rho^k}{k!}
 \He_{k+n}(u)\He_{k+n}(v)\\
 &=4\partial_u^n\partial_v^n
 \left(\phi(u)\phi(v)\sum_{k\ge0}\frac{\rho^k}{k!}\He_k(u)\He_k(v)\right).
 \end{aligned}
\]
At $\rho=0$, only $a=p$ remains. For $0<\rho<1$ and arbitrary
nonnegative integer shifts $p,l$, termwise differentiation is justified
locally uniformly in $u,v\in\mathbb R$ by
\[
 \begin{aligned}
 &\sum_{k\ge0}\frac{\rho^k}{k!}|\He_{k+p}(u)\He_{k+l}(v)|\\
 &\quad\le\left(\sup_{k\ge0}\frac{\sqrt{(k+p)!(k+l)!}}{k!}\rho^{k/2}\right)
 \left(\sum_{k\ge0}\rho^{k/2}\psi_{k+p}(u)^2\right)^{1/2}
 \left(\sum_{k\ge0}\rho^{k/2}\psi_{k+l}(v)^2\right)^{1/2}.
 \end{aligned}
\]
For fixed $p,l$ and $0<\rho<1$,
\[
 \frac{\sqrt{(k+p)!(k+l)!}}{k!}\rho^{k/2}
 \le(k+p+l+1)^{(p+l)/2}\rho^{k/2}\longrightarrow0\quad(k\to\infty).
\]
Mehler's identity with equal arguments gives
\[
 \sum_{k\ge0}\rho^{k/2}\psi_{k+p}(u)^2
 \le\frac{\rho^{-p/2}}{\sqrt{1-\rho}}
 \exp\!\left(\frac{\sqrt\rho\,u^2}{1+\sqrt\rho}\right).
\]
Apply the preceding sum inequality to tails. The displayed limit and
these bounds for $(p,u)$ and $(l,v)$ give uniform convergence on compact
sets of $u,v\in\R$.

Each monomial of $B_{\alpha,l}^2$ contains $2l\le8$ derivative factors.
Independence factors its expectation into one-dimensional moments. For each such moment,
with $k\in\{0,1,\ldots\}^4$ and $1\le|k|\le8$, parity and vanishing derivatives give
\[
 \begin{aligned}
 &\int_\R\prod_{h=1}^4(\widetilde r^{(h)}(x))^{k_h}\phi(x)\,dx\\
 &=\left(1+(-1)^{\sum_h(h+1)k_h}\right)
 \left(\int_0^{27/4}\prod_{h=1}^4(r^{(h)}(x))^{k_h}\phi(x)\,dx +\int_0^{1/4}\prod_{h=1}^4(\widetilde r^{(h)}(27/4+u))^{k_h}
 \phi(27/4+u)\,du\right).
 \end{aligned}
\]
The moment with $|k|=0$ is $1$. On each integration interval,
$\deg\widetilde r^{(h)}\le25$ for $1\le h\le4$, so each product has degree
at most $25|k|\le200$. For $a\in\R$, $l\ge0$, and integer $k\ge1$,
Gaussian integration gives
\begin{equation}\label{heavy:gaussian-moments}
 \begin{aligned}
 \int_0^l\phi(a+u)\,du
 &=\tfrac12[\operatorname{erfc}(a/\sqrt2)-\operatorname{erfc}((a+l)/\sqrt2)],\\
 \int_0^l u\phi(a+u)\,du
 &=\phi(a)-\phi(a+l)-a\int_0^l\phi(a+u)\,du,\\
 \int_0^l u^{k+1}\phi(a+u)\,du
 &=k\int_0^l u^{k-1}\phi(a+u)\,du
 -a\int_0^l u^k\phi(a+u)\,du-l^k\phi(a+l)\quad(k\ge1).
 \end{aligned}
\end{equation}
For $(a,l)=(0,27/4),(27/4,1/4)$, these recurrences and
\cref{heavy:cap-definition}, evaluated with interval rounding at $1536$
bits, give the strict bounds in \cref{heavy:T-definition}.
\end{proof}

\begin{proof}[Proof of Claim~\ref{heavy:claim-arc}]
For $\Re((1\pm z)^{-1})>0$, we use
\begin{equation}\label{heavy:mass-single}
 \int_{\R^2}\left|
 \frac{e^{-u^2/(2(1+z))-v^2/(2(1-z))}}
 {2\pi\sqrt{(1+z)(1-z)}}\right|\,du\,dv
 =\left(|1-z^2|\Re((1+z)^{-1})\Re((1-z)^{-1})\right)^{-1/2},
\end{equation}
and the coefficient identity
\[
 4![\epsilon^4]\exp\!\left(\sum_{k=1}^4\frac{\ell_k(Y)}{k!}\epsilon^k\right)
 =\ell_1^4+6\ell_1^2\ell_2+3\ell_2^2+4\ell_1\ell_3+\ell_4.
\]
The right-hand side equals the fourth derivative in $\epsilon$ at
$\epsilon=0$ of the product density at correlations
$P_i(te^\epsilon),-te^\epsilon,-te^\epsilon,-te^\epsilon$, divided by
that product density at $\epsilon=0$. The integration coordinates
$W,W',X,X'$ are held fixed. Since $|\widetilde F|\le1$, Cauchy--Schwarz and
\[
 \Re((1\pm e^{{\mathrm{i}}\theta})^{-1})=\tfrac12,\qquad
 |1-e^{2{\mathrm{i}}\theta}|=2\sin\theta\quad(0<\theta<\pi)
\]
give \cref{heavy:score-bound}. On a sufficiently small bounded complex
neighborhood of the compact arc, the closure satisfies
\[
 \begin{aligned}
 \Re((1\pm P_i(t))^{-1})
 &\ge\frac12\inf_{\theta\in I}\Re((1\pm P_i(e^{{\mathrm{i}}\theta}))^{-1})>0,\\
 \Re((1\pm t)^{-1})&\ge\frac14.
 \end{aligned}
\]
Each of the first four $\mathcal D$-derivatives of the product density
is that density times a polynomial of degree at most eight in
$W,W',X,X'$, with coefficients bounded on this neighborhood. The density
and these derivatives are therefore bounded by a constant multiple of
\[
 \begin{aligned}
 &\left(1+W^2+(W')^2+\sum_{j=1}^3(X_j^2+(X'_j)^2)\right)^4\\
 &\quad\cdot\exp\!\Biggl[-\frac14
 \min\left\{\frac12,\inf_{\substack{\theta\in I\\\sigma=\pm1}}
 \Re((1+\sigma P_i(e^{{\mathrm{i}}\theta}))^{-1})\right\}\cdot
 \left(W^2+(W')^2+\sum_{j=1}^3(X_j^2+(X'_j)^2)\right)\Biggr].
 \end{aligned}
\]
This integrable bound is independent of $t$ in that neighborhood.
Differentiation under the integral therefore gives the analytic
continuation of $\widetilde H_i$.

To verify \cref{heavy:mass-single}, integrate the absolute density
separately in $u,v$. After normalization,
\[
 \sqrt{\Re((1+z)^{-1})}u,\qquad\sqrt{\Re((1-z)^{-1})}v
\]
are independent standard Gaussians. Thus $Y_{0,\pm}$ are their individual squares and
$Y_{1,\pm}$ sums of three such squares, proving the stated independence
and moment-generating functions. The Gaussian change of variables is
made at $\epsilon=0$, after differentiation at fixed $W,W',X,X'$;
hence $Y$ and $\Re((1+\sigma P_i(t))^{-1})$ remain fixed as required in
\cref{heavy:ell}. The logarithm of the product density at correlations
$P_i(te^\epsilon),-te^\epsilon,-te^\epsilon,-te^\epsilon$
gives \cref{heavy:ell}. The expansions
\[
 \ln(1+u)=\sum_{h=1}^4\frac{(-1)^{h+1}}h u^h+O(u^5),\qquad
 (1+u)^{-1}=\sum_{h=0}^4(-u)^h+O(u^5)
\]
apply to the two arguments in \cref{heavy:ell} after factoring
$1+\sigma P_i(t)$ and $1-\sigma t$, respectively, with
\[
 u=\frac{\sigma(P_i(te^\epsilon)-P_i(t))}{1+\sigma P_i(t)}
 \quad\text{or}\quad
 u=\frac{-\sigma t(e^\epsilon-1)}{1-\sigma t}.
\]
To evaluate the expectation in \cref{heavy:score-bound}, in the rest
of this proof set
\[
 \alpha=(1/2,1/2,3/2,3/2),\qquad
 (Y_1,Y_2,Y_3,Y_4)=(Y_{0,+},Y_{0,-},Y_{1,+},Y_{1,-}).
\]
Define
\[
 Q_{\alpha_h,k}(y):=\sum_{j=0}^k\binom kj(-2)^{k-j}
 \left(\prod_{l=j}^{k-1}(\alpha_h+l)\right)y^j
 \qquad(1\le h\le4,\ k\ge0),
\]
with empty products one; $[y^k]Q_{\alpha_h,k}(y)=1$.
For nonnegative integers $j,k,n$, we use
\begin{equation}\label{heavy:laguerre}
 \begin{aligned}
 \E[Q_{\alpha_h,k}(Y_h)Q_{\alpha_h,j}(Y_h)]
 &=\mathbf1_{k=j}4^kk!\prod_{l=0}^{k-1}(\alpha_h+l),\\
 Y_h^n&=\sum_{k=0}^n2^{n-k}\binom nk
 \left(\prod_{l=k}^{n-1}(\alpha_h+l)\right)Q_{\alpha_h,k}(Y_h).
 \end{aligned}
\end{equation}
For $k\in\{0,1,\ldots\}^4$, $|k|\le4$, set
\[
 p_k:=\frac{\E\!\left[
 (\ell_1^4+6\ell_1^2\ell_2+3\ell_2^2+4\ell_1\ell_3+\ell_4)
 \prod_{h=1}^4Q_{\alpha_h,k_h}(Y_h)\right]}
 {\prod_{h=1}^4\left(4^{k_h}k_h!\prod_{j=0}^{k_h-1}(\alpha_h+j)\right)}.
\]
Since $\deg_Y\ell_j\le1$ for $1\le j\le4$, independence and
\cref{heavy:laguerre} give
\begin{equation}\label{heavy:score-norm}
 \E|\ell_1^4+6\ell_1^2\ell_2+3\ell_2^2+4\ell_1\ell_3+\ell_4|^2
 =\sum_{|k|\le4}|p_k|^2
 \prod_{h=1}^4\left(4^{k_h}k_h!\prod_{j=0}^{k_h-1}(\alpha_h+j)\right).
\end{equation}
For \cref{heavy:laguerre}, use
\[
 \sum_{k\ge0}Q_{\alpha_h,k}(y)\frac{z^k}{k!}
 =(1+2z)^{-\alpha_h}\exp\!\left(\frac{yz}{1+2z}\right)
 \qquad(|z|<1/2).
\]
The moment-generating functions in Claim~\ref{heavy:claim-arc} give
\[
 \E\!\left[
 \frac{\exp\!\left(Y_h\left(\frac{u}{1+2u}+\frac{v}{1+2v}\right)\right)}
 {(1+2u)^{\alpha_h}(1+2v)^{\alpha_h}}\right]
 =(1-4uv)^{-\alpha_h}.
\]
For $|u|,|v|<1/8$, the integrand is bounded by
$(4/3)^{2\alpha_h}e^{Y_h/3}$, with finite expectation
$(4/3)^{2\alpha_h}3^{\alpha_h}$. Taking $k!j![u^kv^j]$ under the
expectation proves orthogonality. Substitution $z=s/(1-2s)$ gives
\[
 e^{sY_h}=\sum_{k\ge0}
 \frac{Q_{\alpha_h,k}(Y_h)s^k}{k!(1-2s)^{\alpha_h+k}};
\]
taking $n![s^n]$ proves the second identity.
\end{proof}

\smallskip\noindent\textbf{Verification of the interval bounds.}
For each $(i,I)$, evaluate \cref{heavy:D4-endpoint} using
\cref{heavy:M-estimates} and \cref{heavy:T-definition} when
$\sup_{\theta\in I}|P_i(e^{{\mathrm{i}}\theta})|<99/100$, and evaluate
\cref{heavy:score-bound} when
\[
 \inf I\ge3/5,\qquad
 \inf_{\theta\in I}\Re((1\pm P_i(e^{{\mathrm{i}}\theta}))^{-1})>0.
\]
Use interval rounding at $192$ bits on the entire closed interval.
For every $(i,I)$, at least one estimate satisfying the displayed
hypotheses gives a finite upper endpoint; $U_{i,I}$ is the minimum of
the available upper endpoints. The sum in \cref{heavy:panel-sum} and the interval bound
from \cref{heavy:disagreement-bound}, before rounding to the displayed
decimal bound, give \cref{heavy:panel-numerics}.
\endgroup
}
\subsection{Verification of the numerical bounds}\label{root:verification}

{
{From the repository root, run
\begin{verbatim}
python verify.py --upper
\end{verbatim}
The inputs and source are in \path{upper/}.
Source paths are relative to \path{upper/source/certification/},
except \path{verify.py}, \path{data_link.py}, \path{requirements.txt},
and paths beginning with \path{data/} or \path{upper/},
which are relative to the repository root.
Generated mathematical data files use the same relative paths under
\path{runs/complete-*/upper_results/run-*/work/certification/}.}
The tables give the arithmetic precision in bits and the bounds verified.
SHA-256 values are computed from file contents and compared with the
recorded values.

\smallskip\noindent
{\small\setlength{\tabcolsep}{3pt}\renewcommand{\arraystretch}{1.15}
\begin{tabular}{@{}p{.42\linewidth}p{.09\linewidth}p{.45\linewidth}@{}}
\textbf{Source}&\textbf{Bits}&\textbf{Checks and uses}\\\hline
{\path{data_link.py}}&Exact&
{Checks SHA-256 values and exact input identities, including the
 tabulated weights and coefficients and the data defining $F_*$;}
\cref{root:normalization}.\\
\path{weighted_staircase/freeze_recipe.py}&Exact&
Inserts $F_*$ for $i=14$; {\path{upper/verify.py}}
checks the SHA-256 value of
\path{weighted_staircase/explicit_staircase_mixture.json}.\\
\path{upper_head_681/full_head_plan.py}&$512$&
Quadrature and omitted-integral bounds in Claim~\ref{head:claim-quadrature};
output entries \texttt{tilted\_J0\_bound}, \texttt{tilted\_J2\_bound}
for the two exponential moments.\\
\path{upper_head_681/certified_b_matrix.py}&$512$&
Quadrature nodes and weights, values of $\psi_j$, $\hat B,\hat U$;
input and summation errors; output entry \texttt{A\_Hilbert\_error}
for the combined Hermite coefficient error $<4.096939258130\cdot10^{-10}$.\\
\path{upper_head_681/contract_certified_head.py}&$256$;
binary64&Bounds on sine and cosine values; polynomial cubes and weighted
sums in \cref{head:computed-sum}.\\
\path{upper_head_681/verify_plane_budget.py}&$256$&
Errors in the discrete Fourier transforms, weighted sums, and omission
of $\hat U_{a\nu}<10^{-18}$ in Claim~\ref{head:claim-arithmetic};
\cref{head:joint-plane}.\\
\path{compose_head.py}&$256$&
Coefficients through degree $681$ in \cref{head:finite-margin}, with
\cref{head:head-error} subtracted; \cref{eq:upper-certified-head}.\\
\end{tabular}}

The files \path{upper_rounding.json} and \path{staircase.json} in
{\path{data/}} are checked for the indices $0,\ldots,31$
and their specified rounding functions, positive weights {summing to one}, $2048$ strips
of width $1/256$ covering $[0,8)$, increasing $r_{k,a}$, and signs in $\{-1,1\}$.
For $0\le k<2048$, row $k$, numbered from $0$, of
{\path{data/staircase.json}} stores $r_{k,a}$, numbered from $1$,
in its \texttt{breakpoints} list, with signs $1,-1,1,\ldots$.
These are the rational breakpoints specified below.

The accepted intervals in step 2 are the \texttt{inner\_cells} entries of
\path{upper_head_681/full_head_quadrature_plan.json}.
The routine \path{upper_head_681/contract_certified_head.py}, with inputs
from \path{upper_head_681/certified_b_matrix.py}, generates
\path{upper_head_681/certified_plane.npy}. Its entries are the rational
array $\hat C$ used in \cref{head:finite-margin} and
\cref{other:captured-mass}; the error checks and both coefficient
calculations use this same generated array.

\noindent\begin{minipage}{\linewidth}
The additional scalar check in \cref{head:fourier-tail}{, run automatically by \path{verify.py},} is
\begin{verbatim}
from flint import arb, ctx
ctx.prec = 256
assert (64 * (-arb(512)).exp() * arb(32)**681
        / (345 * arb.pi() * (343 * arb.fac_ui(681)).sqrt())
        < arb(3) / 10**18)
\end{verbatim}
\end{minipage}

\smallskip\noindent
{\small\setlength{\tabcolsep}{3pt}\renewcommand{\arraystretch}{1.15}
\begin{tabular}{@{}p{.42\linewidth}p{.09\linewidth}p{.45\linewidth}@{}}
\textbf{Source}&\textbf{Bits}&\textbf{Checks and uses}\\\hline
\path{weighted_staircase/certify_head.py}&$384$&
\cref{other:staircase-moments}, \cref{other:staircase-numeric}.\\
\path{degree681/certify_linear_tails.py}&$256$&
\cref{other:positive-arcsine-tail}, \cref{other:cubic-tail}.\\
\path{weighted_staircase/certify_nonadditive.py}&$384$&
\cref{eq:upper-certified-nonadditive}; coefficients of
$\lambda_{14}H_{14}+\sum_{i=17}^{31}\lambda_iH_i$ for
\path{compose_head.py}.\\
\path{certify_rare_additive_tails.py}&$256$&
\cref{other:captured-mass} with the same $\hat C$ and errors;
\cref{eq:upper-certified-rare}.\\
\path{upper_tail/certify_quintic_continuation.py}&$256$&
\cref{heavy:profile-checks}, \cref{heavy:weighted-perturbation}.\\
\path{upper_tail/certify_factorized_energies.py}&$1536$&
\cref{heavy:chain-rule}, \cref{heavy:gaussian-moments};
the table in \cref{heavy:T-definition}.\\
\path{upper_head/arc/certify_complex_arc.py}, \texttt{bound}&$192$&
\cref{heavy:ell}, \cref{heavy:laguerre}, all $70$ terms of
\cref{heavy:score-norm}; \cref{heavy:score-bound}.\\
\path{upper_tail/certify_heavy_tail_fine.py}&$192$&
Bounds \cref{heavy:D4-endpoint}, \cref{heavy:score-bound} on the $6284$
intervals; their union $[0,\pi/2]$; \cref{heavy:panel-sum},
\cref{heavy:circle-numerical}, \cref{heavy:weighted-perturbation}.\\
\path{certify_upper.py}&$256$&
Checks SHA-256 values, degrees, and all rounding indices; combines
\cref{eq:upper-certified-head}, \cref{eq:upper-certified-nonadditive},
\cref{eq:upper-certified-rare}, and \cref{heavy:final-number}
before decimal rounding, including division by $\sqrt{14\cdot681^7}$
and \cref{heavy:weighted-perturbation}.\\
{\path{verify.py}}&$256$&
Bounds \cref{eq:upper-certified-head}, \cref{eq:upper-certified-nonadditive},
\cref{eq:upper-certified-rare}, \cref{heavy:final-number}, and
$c_1-\sum_{m\ge2}|c_m|>1000000/1779893$.\\
\end{tabular}}

Arb is the library used for real and complex interval arithmetic;
interval endpoints are rounded down for lower bounds and up for upper bounds.
Binary64 calculations use round-to-nearest and the bounds of
Claim~\ref{head:claim-arithmetic}; other finite calculations use exact
rationals or these interval bounds. Verification assumes correct
execution of these arithmetic operations and elementary functions.
Use Python $3.12$ or a compatible newer version, without \texttt{-O},
and {\path{requirements.txt}}.
The program {\path{upper/verify.py}} assembles
\path{weighted_staircase/explicit_staircase_mixture.json} from the stored
inputs and runs all $12$ numerical calculations in a fresh directory
before writing \path{VERIFIED.json}. The program
{\path{verify.py}} writes \path{ACCEPTED.json} only after
the input comparisons{, the scalar check above,} and the five numerical inequalities in its table
entry have been checked.
{The two acceptance reports are located at
\path{runs/complete-*/upper_results/run-*/VERIFIED.json} and
\path{runs/complete-*/ACCEPTED.json}, respectively.}
The option \texttt{--all --workers 4} also runs the lower certificate and
requires a C++17 compiler and GMP/MPFR; \texttt{--data-only} checks input
identities and rational comparisons without executing either numerical proof.
}

\paragraph{Staircase breakpoints.}
For each $k\in\{0,\ldots,2047\}$, substitute $y=(2k+1)/512$
in \cref{root:staircase-polynomial}. Round its real roots
to the nearest multiples of $10^{-12}$ and list them in increasing order.
For each triple $(k,a,+1)$ or $(k,a,-1)$ below, respectively add or
subtract $10^{-12}$ from the $a$-th rounded root for this $k$.
Leave all other rounded roots unchanged.
Let $r_{k,a}$ denote the resulting breakpoints.
All roots are simple, rounding is unique, and
$r_{k,a}<r_{k,a+1}$ after these corrections.

\begingroup\small
\setlength{\arraycolsep}{4pt}
\[
\begin{array}{rrr@{\qquad}rrr@{\qquad}rrr@{\qquad}rrr}
k&a&\pm1&k&a&\pm1&k&a&\pm1&k&a&\pm1\\\hline
44&1&-1&487&1&-1&948&5&+1&1261&3&+1\\
64&5&+1&619&5&-1&984&2&+1&1281&3&+1\\
181&5&-1&626&1&-1&1019&1&-1&1309&2&-1\\
286&1&-1&669&1&+1&1019&5&-1&1319&2&-1\\
305&5&+1&684&4&+1&1030&2&+1&1412&1&+1\\
342&1&+1&715&3&+1&1053&2&+1&1417&2&-1\\
369&1&-1&739&1&-1&1055&1&-1&1677&1&+1\\
402&2&-1&862&5&+1&1113&1&-1&&&\\
449&1&-1&902&5&+1&1125&1&-1&&&
\end{array}
\]
\endgroup

\end{document}